\PassOptionsToPackage{colorlinks=true, citecolor=cyan}{hyperref}
\documentclass[copyright,creativecommons]{eptcs}
\pdfoutput=1
\providecommand{\event}{QPL 2023}
\usepackage[utf8]{inputenc}
\usepackage[english]{babel}
\usepackage[T1]{fontenc}
\usepackage{amsmath}
\usepackage{parskip}
\usepackage[numbers,sort&compress]{natbib}
\usepackage{tikz}
\usepackage{lipsum}

\usepackage{pdfsync}
\usepackage{tikzit}
\usepackage{amsfonts}
\usepackage{amsthm}
\usepackage{amssymb}
\usepackage{thmtools}
\usepackage{dsfont}
\usepackage{braket}
\usepackage{mathtools}
\usepackage[export]{adjustbox}
\usepackage{scalerel}
\usepackage[nameinlink]{cleveref}
\crefname{subsection}{subsection}{subsections}
\usepackage{xstring}
\usepackage{ifthen}
\usepackage{stmaryrd}
\usepackage{subcaption}
\usepackage[normalem]{ulem}

\theoremstyle{plain}
\newtheorem{theorem}{Theorem}[section]
\newtheorem{plottwist}[theorem]{Plot-twist}

\newtheorem{definition}[theorem]{Definition}

\newtheorem{lemma}[theorem]{Lemma}
\newtheorem{corollary}[theorem]{Corollary}

\newtheorem{openquestion}[theorem]{Open question}

\declaretheoremstyle[qed={$\blacksquare$}]{exampleStyle}

\declaretheorem[style=exampleStyle, sibling=theorem, name=Example, numbered=no]{example*}

\newcommand{\ccite}[2][]{%
    \IfSubStr{#2}{,}{refs.~}{ref.~}%
    \ifthenelse{\equal{#1}{}}{\cite{#2}}{\cite[#1]{#2}}}
\newcommand{\Ccite}[2][]{%
    \IfSubStr{#2}{,}{Refs.~}{Ref.~}%
    \ifthenelse{\equal{#1}{}}{\cite{#2}}{\cite[#1]{#2}}}
\newcommand{\zoo}[1]{%
    \href{%
        https://errorcorrectionzoo.org/c/#1%
    }{\includegraphics[height=.75\baselineskip, valign=c]{%
        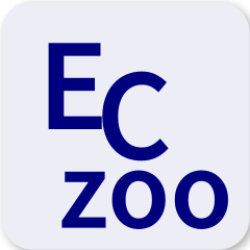}}}

\newcommand{\ZZ}{\mathbb{Z}}

\newcommand{\GGG}{\mathcal{G}}

\newcommand{\MMM}{\mathcal{M}}

\newcommand{\PPP}{\mathcal{P}}
\newcommand{\SSS}{\mathcal{S}}

\newcommand{\one}{\mathds{1}}

\newcommand{\defeq}{\coloneqq}
\newcommand{\after}{\circ}

\newcommand{\allIdentitiesExcept}[1]{I \otimes \ldots \otimes I \otimes #1 \otimes I \otimes \ldots \otimes I}
\newcommand{\isg}[1]{\SSS_{#1}}
\newcommand{\lpg}[1]{N(\isg{#1}) / \isg{#1}}
\newcommand{\sublpg}[1]{N(\SSS) / #1}
\newcommand{\qcode}[1]{\llbracket#1\rrbracket}

\newcommand{\stabCase}[1]{\hyperlink{stabilizer_formalism_case_#1}{\textit{Case #1}}}
\newcommand{\normCase}[1]{\hyperlink{normalizer_formalism_case_#1}{\textit{Case #1}}}

\newcommand{\ol}[1]{\overline{#1}}
\newcommand{\qqquad}{\quad\qquad}
\newcommand{\qqqquad}{\quad\qqquad}
\newcommand{\modPhases}[1]{#1^\star}
\newcommand{\addPhases}[1]{\langle i \rangle #1}

\DeclareMathAlphabet{\mathcal}{OMS}{cmsy}{m}{n}

\begin{document}
    \title{Floquetifying the Colour Code}

\author{
    Alex Townsend-Teague, Julio Magdalena de la Fuente and Markus Kesselring
    \institute{Dahlem Center for Complex Quantum Systems, Freie Universität Berlin, 14195 Berlin, Germany}
    \email{alex.townsend-teague@outlook.com}
}
\def\titlerunning{Floquetifying the Colour Code}
\def\authorrunning{Alex Townsend-Teague, Julio Magdalena de la Fuente \& Markus Kesselring}
\def\copyrightholders{Alex Townsend-Teague, \\ Julio Magdalena de la Fuente \& Markus Kesselring.}
\maketitle
    \begin{abstract}
    \noindent
    Floquet codes are a recently discovered type of quantum error correction code.
    They can be thought of as generalising stabilizer codes and subsystem codes,
    by allowing the logical Pauli operators of the code to vary dynamically over time.
    In this work, we use the ZX-calculus to create new Floquet codes that are
    in a definable sense equivalent to known stabilizer codes.
    In particular, we find a Floquet code that is equivalent to the colour code,
    but has the advantage that all measurements required to implement it are of weight one or two.
    Notably, the qubits can even be laid out on a square lattice.
    This circumvents current difficulties with implementing the colour code fault-tolerantly,
    while preserving its advantages over other well-studied codes,
    and could furthermore allow one to benefit from extra features exclusive to Floquet codes.
    On a higher level, as in \Ccite{ZxFaultTolerancePsiQuantum},
    this work shines a light on the relationship between
    `static' stabilizer and subsystem codes and `dynamic' Floquet codes;
    at first glance the latter seems a significant generalisation of the former,
    but in the case of the codes that we find here,
    the difference is essentially just a few basic ZX-diagram deformations.
\end{abstract}
    \section{Introduction}\label{sec:introduction}

In 2021, Hastings and Haah discovered the honeycomb code~\cite[\zoo{honeycomb}]{DynamicallyGeneratedLogicalQubits}.
At first glance, it looked a lot like a subsystem code~\cite[\zoo{oecc}]{SubsystemCodesOriginalPaper};
it was defined via a sequence of non-commuting Pauli measurements whose individual outcomes were random,
but combined to give deterministic outcomes suitable for use in catching
errors that occur during quantum computation.
But something didn't quite add up. 
Viewed exactly as a subsystem code, it encoded no logical information.
This was because the logical information was instead `dynamically' encoded.
This made it the first member of a new class of codes,
which have come to be called \textit{Floquet codes}~\cite[\zoo{floquet}]{PlanarFloquetCodes}.
Since its publication, it seems a number of people have been holed up in their offices thinking about Floquet
codes, because of late a new one has popped up every month or
so~\cite{AnyonCondensation, FloquetCodesWithoutParents, AutomorphismCodes, XCubeFloquetCode, AndiFixedPointPathIntegrals, FloquetTwistDefectNetworks}.
In particular, in both \Ccite{AnyonCondensation} and \Ccite{FloquetCodesWithoutParents}, a whole family of Floquet
codes is described, of which the honeycomb code is a single member.
We call this family the \textit{condensed colour codes}, as in~\cite{AnyonCondensation}.

\begin{plottwist}
    Before the honeycomb code paper was published,
    another Floquet code had already independently been discovered.
\end{plottwist}

In \Ccite{ZxFaultTolerancePsiQuantum},
the authors write that Hector Bombín had previously discovered an equivalence between
the rotated surface code~\cite[\zoo{rotated_surface}]{RotatedSurfaceCodeOriginalPaper}
and a (then unnamed) condensed colour code.
This equivalence is shown in \Ccite{ZxFaultTolerancePsiQuantum} using the ZX-calculus,
a flexible but rigorous graphical formalism for quantum mechanics~\cite{ZxCalculusOriginalPaper, DodoBook}.
Given any ZX-diagram representing an error correction protocol,
something the ZX-calculus is particularly good at is
identifying ways of rewriting high-weight Pauli measurements as weight-two or weight-one Pauli measurements.
The authors of \Ccite{ZxFaultTolerancePsiQuantum} applied this idea to
a ZX-diagram representing multiple rounds of rotated surface code measurements.
Interestingly, Craig Gidney had himself previously done something very similar in \Ccite{PairwiseMeasurementSurfaceCode}
to find an implementation of the rotated surface code that uses only weight-two measurements.
The only difference is that Gidney applied the idea to a circuit representing a \emph{single} round of
rotated surface code measurements.
In both cases all measurements were reduced to weight-two, but only in the former case was the result a Floquet
code\footnote{%
    Similar in spirit to \Ccite{ZxFaultTolerancePsiQuantum} is \Ccite{AndiFixedPointPathIntegrals};
    both can be viewed as using graphical tensor network approaches to construct new Floquet codes from existing stabilizer codes.
    However, the latter's author approaches things through the lens of \textit{topological spacetime path-integrals},
    and uses a tensor network approach that is closely related to, but isn't exactly the same as, the ZX-calculus.
    Nonetheless, for \textit{topological codes}, we believe the two approaches are equivalent.}.

The colour code~\cite[\zoo{color}]{ColourCodeOriginalPaper}
has certain properties that make it arguably more appealing than the surface code,
such as a higher encoding rate~\cite{ColourCodeLatticeSurgery},
transversal Clifford gates~\cite{ColourCodeOriginalPaper}
and more efficient lattice surgery operations~\cite{LowOverheadColourCode}.
The high-weight measurements naively required for its implementation, however,
are the major obstacle to realising it practically~\cite{ColourCodeFlagQubits}.
So in the same way that \Ccite{ZxFaultTolerancePsiQuantum} `Floquetified' the rotated surface code,
it would be great to have a `Floquetified' colour code -
that is, a Floquet code that is in a definable sense equivalent to the colour code -
in which all measurements are weight-two or less.
Perhaps even more excitingly,
Floquet codes can come equipped with the ability to perform logical Clifford gates
both fault-tolerantly and at no extra effort;
this is discussed in detail in \Ccite{AutomorphismCodes}.
The honeycomb code, for example, naturally performs a fault-tolerant logical Hadamard gate every three timesteps.
Exactly which logical gates can be implemented by a Floquet code in this manner is restricted by
the automorphism group of what condensed matter theorists call the \textit{anyonic defects} of the code.
The honeycomb code can be shown to have exactly one non-trivial such automorphism,
which corresponds exactly to the logical Hadamard gate.
A Floquetified colour code, however, would in principle inherit its automorphism group from the colour code -
this is much richer, containing 72 elements~\cite{YoshidaTopologicalColourCode,ColourCodeBoundariesAndTwists}.
So in such a code,
the set of logical Clifford gates that could potentially be fault-tolerantly implemented in this way is larger.

To this end, we aimed to use the ideas from \Ccite{ZxFaultTolerancePsiQuantum}
to Floquetify the colour code.
We succeeded, finding a Floquet code with period 13 whose qubits can be laid out on a square lattice,
and in which all measurements are weight one or two.
This paper proceeds as follows.
In \Cref{sec:preliminaries}, we introduce the definitions and notation we'll need throughout the paper.
First, we introduce \textit{ISG codes},
of which stabilizer codes~\cite[\zoo{stabilizer}]{GottesmanPhdThesis},
subsystem codes
and Floquet codes are subtypes
(thus far, we're not aware of any universally accepted formal definition of a Floquet code in the literature).
We also import the graphical formalism of \textit{Pauli webs} from \Ccite{ZxFaultTolerancePsiQuantum}
(generalised to \textit{stabilizer flow} in \Ccite{MattTimeDynamics}),
which allows us to reason graphically about stabilizers, logical operators and \textit{detectors}.
In \Cref{sec:floquetifying_422}, we jump in the shallow end by Floquetifying the $\qcode{4, 2, 2}$
code~\cite[\zoo{stab_4_2_2}]{422Code},
demonstrating the key ideas behind this Floquetification process on a simple example.
The deep end awaits in \Cref{sec:floquetifying_colour_code}, where we Floquetify the colour code.
We include many extra details in appendices;
these will be signposted throughout the main text.

    \section{Preliminaries}\label{sec:preliminaries}

We will assume familiarity with
stabilizer codes and the stabilizer formalism~\cite{GottesmanPhdThesis},
as well as the ZX-calculus.
For the uninitiated, good introductory references are \Ccite[Section 10.5]{NielsenChuang}
or \Ccite{GottesmanQecLectureNotes} for the former,
and \Ccite{ZxWorkingScientist} for the latter.
In the appendix, we also include a reminder of how measurement works in the stabilizer formalism -
see \Cref{thm:stabilizer_formalism}.

\subsection{ISG codes}\label{subsec:isg_codes}

We begin by introducing \textit{ISG codes},
where `ISG' stands for \textit{instantaneous stabilizer group},
a term introduced in \Ccite{DynamicallyGeneratedLogicalQubits}.
But first, some notation;
we'll use $\sigma_0 = I, \sigma_1 = X, \sigma_2 = Y$ and $\sigma_3 = Z$ to denote the \textit{Pauli matrices},
and $\PPP_1$ to denote the \textit{single qubit Pauli group} they form under composition.
$\PPP_n$ will then denote the \textit{$n$-qubit Pauli group}
$\PPP_n = \{p_1  \otimes \ldots \otimes p_n \mid \forall j: p_j \in \PPP_1\}$ for $n > 1$.
Every element of this group can be written in the form
$i^\ell (\sigma_{j_1} \otimes \ldots \otimes \sigma_{j_n})$, for $\ell \in \{0, 1, 2, 3\}$.
This $i^\ell$ is often called a \textit{phase}.
One such element we'll use a lot is $(\sigma_j)_k = \allIdentitiesExcept{\sigma_j}$,
which has an $I$ in each tensor factor except the $k$-th, where we insert the Pauli matrix $\sigma_j$.
Another common element is $\one = I \otimes \ldots \otimes I$.
By a slight abuse of notation, we will usually write $i^\ell \one$ or $i^\ell I$ as just $i^\ell$.
The \textit{weight} of any element $i^\ell (\sigma_{j_1} \otimes \ldots \otimes \sigma_{j_n})$ is
the number of Pauli matrices $\sigma_{j_k}$ that aren't $I$.

Rather than defining qubits via Hilbert spaces, we'll stick to group theory;
we'll simply define that we have \textit{a system of $n$ qubits} whenever we have any group $\GGG$ isomorphic to $\PPP_n$.
A particularly important example will be $\GGG = N(\SSS) / \SSS \cong \PPP_{n-r}$,
for any stabilizer group $\SSS \leq \PPP_n$
with \textit{rank} (size of any minimal generating set) $r$~\cite[Section 3.4]{GottesmanQecLectureNotes}.
Recall that a stabilizer group is just any subgroup of $\PPP_n$ that doesn't contain $-\one$,
and this forces it to be Abelian.
Here, $N(\SSS) = \{p \in \PPP_n \mid p\SSS p^{-1} = \SSS \}$ denotes the \textit{normalizer} of $\SSS$ in $\PPP_n$.
So $N(\SSS) / \SSS$ is the quotient group consisting of \textit{left cosets} of $\SSS$ in $N(\SSS)$,
whose elements $p\SSS = \{ps \mid s \in \SSS \}$ will often be denoted $\ol{p}$ for short.
The \textit{weight} of such a coset is the minimum weight over all its elements.
For any groups $\GGG_1, \ldots, \GGG_{\ell}$,
we'll write the product $\GGG_1 \ldots \GGG_{\ell}$ to mean
the group generated by the union of generating sets for $\GGG_1, \ldots, \GGG_{\ell}$.

It can be shown that $\PPP_n$ has presentation
$\PPP_n = \langle i, X_1, Z_1, X_2, Z_2, \ldots, X_n, Z_n \rangle$.
Thus any group $\GGG$ isomorphic to $\PPP_n$ has presentation $\langle \iota, x_1, z_1, x_2, z_2, \ldots, x_n, z_n \rangle$,
for some elements $\iota, x_j, z_j$ in $\GGG$,
with group isomorphism $\GGG \cong \PPP_n$ given by $\iota \mapsto i, x_j \mapsto X_j$ and $z_j \mapsto Z_j$.
In particular, all the $x_j$ and $z_j$ generators obey the same commutativity relations as the Paulis $X_j$ and $Z_j$.
If we think of $\GGG$ as defining $n$ qubits,
we can identify qubit $j$ with the subgroup $\langle \iota, x_j, z_j \rangle \cong \langle i, X_j, Z_j \rangle \cong \PPP_1$.

We can now define an \textit{ISG code}.
Given $n$ qubits, an ISG code is defined entirely by a \textit{measurement schedule} $\MMM$,
which is an ordered list $[\MMM_0, \MMM_1, \ldots]$ of Abelian subgroups of $\PPP_n$.
The schedule $\MMM$ can be finite or infinite.
If it's finite, with length $\ell$, say, then we let the subscript in $\MMM_j$ be modulo $\ell$.
Given any such $\MMM$, there exists a subgroup $\SSS_t$ of $\PPP_n$ for all $t \in \ZZ$ called the \textit{instantaneous stabilizer group} (ISG).
This is defined recursively: for $t < 0$, it's always the trivial group $\{\one\}$.
For $t \geq 0$, $\isg{t}$ is formed from $\isg{t-1}$ by measuring a generating set for $\MMM_t$;
the effect of this can be determined using the stabilizer formalism (\Cref{thm:stabilizer_formalism}).
This is well-defined, in that it doesn't depend on the choice of generating set for $\MMM_t$.
We'll often call $t$ the \textit{timestep} (or just \textit{time}).
At every timestep $t \in \ZZ$, let $r_t$ denote the rank of $\isg{t}$,
and let $k_t \defeq n-r_t$.
Then $\lpg{t} \cong \PPP_{k_t}$.
That is, we can consider ourselves to have a system of $k_t$ qubits.
This is the idea behind an \textit{ISG code}.

\begin{definition}\label{defn:ISG_code}
    An $\qcode{n, k, d}$ \textbf{ISG code} is given by a measurement schedule $\MMM$,
    with the property that for some $T \in \ZZ_{\geq 0}$ and all $t \geq T$,
    the group $\isg{t}$ has some fixed rank $r$.
    We say that for all such timesteps $t \geq T$, the code is \textbf{established}.
    This encodes $k = n-r$ \textbf{logical qubits} whenever $t \geq T$,
    via the \textbf{logical Pauli group} $\lpg{t} \cong \PPP_{k}$.
    The \textbf{distance} $d$ is the minimum weight of any element of $\lpg{t}$, over all $t \geq T$.
    The \textbf{period} is the length $\ell \in \ZZ_{\geq 0} \cup \{\infty\}$ of the list $\MMM = [\MMM_0, \MMM_1, \ldots]$,
    and can be finite or infinite.
\end{definition}

For a slightly longer discussion of this definition, see \Cref{sec:isg_code_remarks}.
As advertised, stabilizer codes, subsystem codes and Floquet codes are types of ISG code.

\begin{definition}\label{defn:stabilizer_code_as_ISG_code}
    A \textbf{stabilizer code} is an ISG code
    $\MMM = [\MMM_0, \MMM_1, \ldots]$ such that
    the group $\MMM_0 \MMM_1 \ldots$ generated by
    the union of generating sets of all $\MMM_t$ in $\MMM$ is itself Abelian.
\end{definition}

A stabilizer code has the property that, after establishment at time $T$,
the ISG $\SSS_t$ is the same for all $t \geq T$.
There thus exist fixed Paulis $x_1, z_1, \ldots, x_k, z_k \in \PPP_n$ such that
$\langle \ol{i}, \ol{x_1}, \ol{z_1}, \ldots, \ol{x_k}, \ol{z_k} \rangle$
is a presentation for $\lpg{t} \cong \PPP_k$ for all $t \geq T$.
This latter property is baked into the next definition, which is admittedly a bit of a mouthful,
and can be skipped by any readers not already familiar with subsystem codes.
Therein, for a group $\GGG$, we let
$Z(\GGG) = \{g \in \GGG \mid gh = hg ~ \forall h \in \GGG \}$ denote its \textit{center},
and $\PPP_k^\circ$ be the `almost Pauli group' $\langle X_1, Z_1, \ldots, X_k, Z_k \rangle$.

\begin{definition}\label{defn:subsystem_code_as_ISG_code}
    A \textbf{subsystem code} is an ISG code $\MMM = [\MMM_0, \MMM_1, \ldots]$ that establishes at time $T$,
    such that the group $\MMM_0 \MMM_1 \ldots$ generated by
    the union of generating sets of all $\MMM_t$ in $\MMM$ satisfies the following:
    letting $\GGG = \addPhases{\MMM_0 \MMM_1 \ldots}$
    and $\SSS$ be a stabilizer group such that $\addPhases{\SSS} = Z(\GGG)$,
    there exist fixed Paulis $x_1, z_1, \ldots, x_k, z_k$ such that
    $\langle x_1 \GGG, z_1 \GGG, \ldots, x_k \GGG, z_k \GGG \rangle$
    is a presentation for $\sublpg{\GGG} \cong \PPP_k^\circ$ and,
    for all $t \geq T$,
    $\langle i \isg{t}, x_1 \isg{t}, z_1 \isg{t}, \ldots, x_k \isg{t}, z_k \isg{t} \rangle$
    is a presentation for $\lpg{t} \cong \PPP_k$.
\end{definition}

Unlike a stabilizer code,
after establishment at time $T$,
the ISG $\SSS_t$ of a subsystem code may change from one timestep to another (while always having the same rank).
However, such a code still has the property that
there exist fixed Paulis that can represent $\lpg{t} \cong \PPP_k$ for all $t \geq T$.
This is what is meant when stabilizer and subsystem codes are labelled \textit{static}.

Our stabilizer code definition above agrees exactly with the usual one;
though our definition requires a measurement schedule $\MMM = [\MMM_0, \MMM_1, \ldots]$ to be specified,
the fact that $\MMM_0 \MMM_1 \ldots$ is Abelian makes this irrelevant.
Our subsystem code definition, however, slightly deviates from the usual one;
here the fact that a measurement schedule is required is very relevant.
We go into more detail on this in \Cref{sec:subsystem_code_remarks}.

\begin{definition}\label{defn:floquet_code}
    A \textbf{Floquet code} is an ISG code with a finite period.
\end{definition}

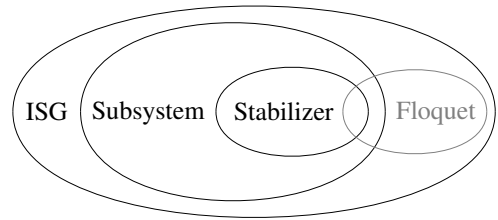
\begin{figure}[b]
    \centering
    \begin{tabular}{c | c | c | c}
        \hline
        Code type & $\isg{t}$ & $\lpg{t}$ & \textcolor{gray}{Period} \\
        \hline
        Stabilizer & Static & Static & \textcolor{gray}{Finite or infinite} \\
        Subsystem & Dynamic & Static & \textcolor{gray}{Finite or infinite} \\
        \textcolor{gray}{Floquet} & \textcolor{gray}{Dynamic} & \textcolor{gray}{Dynamic} & \textcolor{gray}{Finite} \\
        ISG & Dynamic & Dynamic & \textcolor{gray}{Finite or infinite} \\
        \hline
    \end{tabular}
    \quad
    \scalebox{0.9}{\begin{tikzpicture}
	\begin{pgfonlayer}{nodelayer}
		\node [style=none] (2) at (7.5, 0) {};
		\node [style=none] (3) at (-6.75, 0) {};
		\node [style=none] (4) at (4.25, 0) {};
		\node [style=none] (5) at (-4.75, 0) {};
		\node [style=none] (6) at (3.75, 0) {};
		\node [style=none] (7) at (-0.75, 0) {};
		\node [style=none] (8) at (7.25, 0) {};
		\node [style=none] (9) at (3, 0) {};
		\node [style=none] (10) at (-5.75, 0) {ISG};
		\node [style=none] (12) at (1.25, 0) {Stabilizer};
		\node [style=none] (13) at (5.75, -0.05) {\textcolor{gray}{Floquet}};
		\node [style=none] (14) at (-2.75, -0.05) {Subsystem};
	\end{pgfonlayer}
	\begin{pgfonlayer}{edgelayer}
		\draw [bend left=90, looseness=0.75] (2.center) to (3.center);
		\draw [bend left=90, looseness=0.75] (3.center) to (2.center);
		\draw [bend left=90] (4.center) to (5.center);
		\draw [bend left=90] (5.center) to (4.center);
		\draw [bend left=90] (6.center) to (7.center);
		\draw [bend left=90] (7.center) to (6.center);
		\draw [color=gray, bend left=90] (8.center) to (9.center);
		\draw [color=gray, bend left=90] (9.center) to (8.center);
	\end{pgfonlayer}
\end{tikzpicture}}
    \caption{%
        Table and Venn diagram showing relationships between stabilizer, subsystem, Floquet and ISG codes.
        We attribute little importance to whether an ISG code has finite period,
        hence this column in the table is drawn in grey.
        If this column is ignored, there is no difference between a Floquet code and a general ISG code,
        hence the former's row in the table, as well as its bubble in the Venn diagram, are also drawn in grey.
        The Venn diagram shows that all stabilizer codes are subsystem codes, and all subsystem codes are ISG codes.
        Though all Floquet codes are ISG codes,
        only stabilizer and subsystem codes with finite period are Floquet codes.}
    \label{fig:isg_relationships}
    \vspace{-10pt}
\end{figure}

We do not attribute so much importance to whether or not an ISG code has finite period,
and hence whether it's labelled a Floquet code or not.
Indeed, there are ISG codes that don't fall into any of the three categories above -
examples include the \textit{dynamic tree codes} of \Ccite{FloquetCodesWithoutParents}.
More interesting to us is the fact that for a general ISG code that establishes at time $T$,
there need not exist fixed Paulis $x_1, z_1, \ldots, x_k, z_k \in \PPP_n$ such that
$\langle \ol{i}, \ol{x_1}, \ol{z_1}, \ldots, \ol{x_k}, \ol{z_k} \rangle$
is a presentation for $\lpg{t} \cong \PPP_k$ for all $t \geq T$.
This is what is meant when such codes are labelled \textit{dynamic}.

In \Cref{fig:isg_relationships},
we show a table and a Venn diagram characterising the relationships between these code types,
and in \Cref{sec:isg_code_example} we give a simple example of an ISG code and its evolution.
When working with ISG codes,
calculating the effect on $\lpg{t}$ of measuring a Pauli $p \in \PPP_n$ is paramount.
To this end, a vital tool is a corollary of the stabilizer formalism which we informally call the \textit{normalizer formalism};
we state and prove it in \Cref{sec:normalizer_formalism}.

\subsection{Pauli webs}\label{subsec:stabilizer_flows}

Though we've now defined ISG codes,
we haven't said how to actually detect errors on them,
nor how to perform logical (Pauli) operations.
Both of these can be viewed elegantly in the ZX-calculus via \textit{Pauli webs},
as defined in \Ccite{ZxFaultTolerancePsiQuantum}.
This is analogous to \textit{firing} spiders in \Ccite{BorghansMastersThesis},
and is generalised to \textit{stabilizer flow} in \Ccite{MattTimeDynamics}.
Here we'll only introduce it in a limited and informal way,
since this is all we'll need for \Cref{sec:floquetifying_422,sec:floquetifying_colour_code}.
For a more rigorous discussion, see \Ccite{ZxFaultTolerancePsiQuantum} or \Ccite{MattTimeDynamics}.

Given a \textit{Clifford ZX-diagram} (one in which all spider phases are integer multiples of $\frac{\pi}{2}$),
we'll define an \textit{(unsigned CSS) Pauli web} to be a highlighting of wires green or red
(corresponding to $Z$ and $X$, respectively), according to certain rules.
Essentially, the green highlighted edges correspond to how a $Z$ gate can propagate through the diagram,
and likewise for red edges and the $X$ gate.
Specifically, a highlighted wire can only end at a \textit{Pauli spider}
(one whose phase is an integer multiple of $\pi$), a Hadamard box\footnote{%
    Under the hood, a Hadamard box is actually 
    a composition of three spiders with phases $\pm \frac{\pi}{2}$.
    The fact that an unsigned CSS Pauli web can end here
    might seem to contradict the fact we just said it can only end at Pauli spiders.
    However, this is just a consequence of the limited way in which we've imported Pauli webs here.
    More general (unsigned) Pauli webs can terminate at any \textit{Clifford spider} -
    one with phase $k \frac{\pi}{2}$, for $k \in \ZZ$.},
or an input or output node of the overall diagram;
a green Pauli spider must have an even number of legs highlighted green (and likewise for red Pauli spiders and red edges);
a green Pauli spider must have no legs or every leg highlighted red (and likewise for red Pauli spiders and green edges);
and if one leg of a Hadamard box is highlighted green, the other must be red.
Examples of Pauli webs on small ZX-diagrams are shown below.
Throughout this paper, ZX-diagrams should be read bottom-to-top:
\begin{equation}\label{eq:simple_stabilizer_flows}
    \begin{aligned}
        \includegraphics[width=350pt]{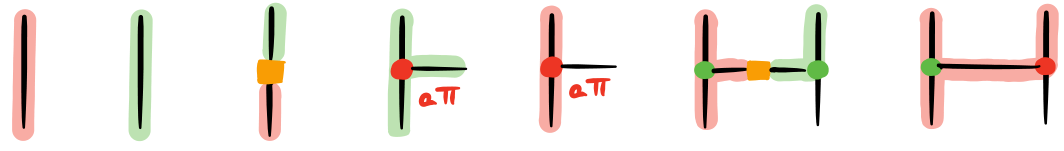}
    \end{aligned}
\end{equation}

\subsubsection{Detectors}\label{subsubsec:detectors}

A detector is a set of measurement outcomes $m_j \in \{-1, 1\}$
whose product is deterministic in the absence of noise~\cite{StimPaper, SparseBlossom}.
We can write them as \textit{formal products}\footnote{%
    By \textit{formal product},
    we mean we forget that symbols like $m_j$ are actually stand-ins for values like $-1$ and $1$,
    and treat the symbols just as objects to be moved around algebraically.
    By taking powers modulo 2, we mean - for example -
    the formal product $m^3$ is the same as $m^1 = m$.}
with powers taken modulo 2.
For example, given a qubit in state $\ket{0}$, a $Z$-basis measurement should deterministically give outcome $m = 1$.
Thus the formal product $m$ is a detector.
On the other hand, if the qubit is in state $\ket{+}$, a $Z$-basis measurement's outcome $m_1$ is completely random.
But a second $Z$-basis measurement should give outcome $m_2$ identical to $m_1$; that is, $m_1 m_2$ should be 1.
So the formal product $m_1 m_2$ is a detector.
Specifically, it detects Pauli $X$ errors -
if one occurs between the first and second measurement, we'll get $m_2 = -m_1$ and hence $m_1 m_2 = -1$.
We say in this case that the detector $m_1 m_2$ is \textit{violated}.
In an ISG code context,
detectors occur whenever we measure a Pauli $p$ such that $p$ or $-p$ is in the ISG $\SSS_t$
(\stabCase{2} in \Cref{sec:normalizer_formalism}).
Following~\cite{MattTimeDynamics}, we'll define an \textit{(unsigned CSS) detecting region} to be
a Pauli web with the additional constraint that no input or output nodes of the overall diagram are incident to highlighted edges.
Now, recall that in the ZX-calculus,
$Z \otimes Z \otimes \ldots \otimes Z$ and $X \otimes X \otimes \ldots \otimes X$ measurements
with outcome $m$ can be represented respectively as:
\begin{equation}\label{eq:zx_CSS_measurements}
    \begin{aligned}
        \includegraphics[width=350pt]{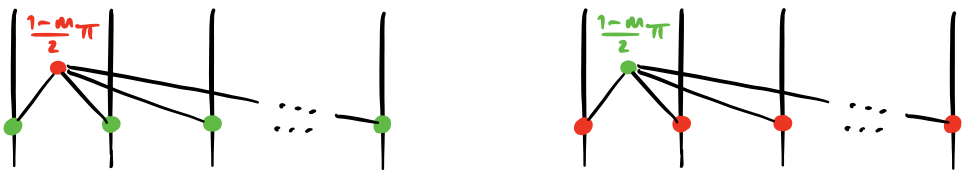}
    \end{aligned}
\end{equation}
In particular, note that the measurement outcome $m$ parametrises a spider phase.
A detecting region then corresponds to a detector $m_1 \ldots m_d$ as follows:
the measurement outcomes $m_j$ in the detector are all those that parametrise
a red spider incident to a green highlighted edge, or a green spider incident to a red highlighted edge.
In fact, throughout this paper we will always be able to post-select;
that is, we can assume all measurement outcomes $m_j$ are 1.
See \Cref{sec:measurement_vs_post_selection} for a longer discussion of this.
Below are some simple detecting regions;
in each diagram, horizontal wires correspond to measurements
(or rather, post-selections; we can thus omit the spider phases that correspond to the measurement outcomes).
The resulting detectors consist of exactly the outcomes of the measurements represented by these horizontal wires.
\begin{equation}\label{eq:simple_detecting_regions}
    \begin{aligned}
        \includegraphics[width=300pt]{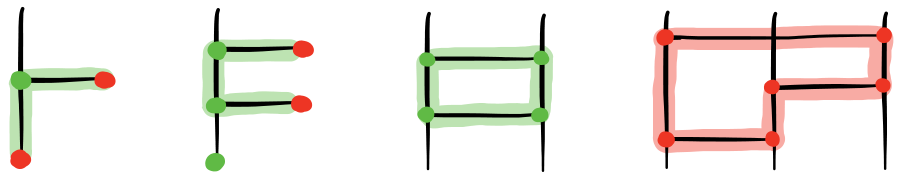}
    \end{aligned}
\end{equation}

\subsubsection{Stabilizers and logical operators}\label{subsubsec:stabilizers_and_logicals}

Given an ISG code,
the stabilizers (elements of $\isg{t}$) and logical operators (members of cosets of $\lpg{t}$)
can also be seen via Pauli webs.
In analogy with a detecting region,
we can define an \textit{(unsigned CSS) stabilizing region} on a ZX-diagram to be a Pauli web in which
none of the diagram's input nodes are incident to a highlighted edge, but at least one output node is.
Supposing the ZX-diagram has $n$ output wires,
the stabilizer corresponding to such a stabilizing region is (up to $\pm1$ sign) the Pauli
$p = \sigma_{j_1} \otimes \ldots \otimes \sigma_{j_n} \in \PPP_n$, where $\sigma_{j_\ell}$ is
$I$ if output wire $\ell$ isn't highlighted, $Z$ if it's highlighted green, and $X$ if it's highlighted red.
For an ISG code with measurement schedule $\MMM = [\MMM_0, \MMM_1, \ldots]$,
we can draw a ZX-diagram that corresponds to measuring a generating set for $\MMM_0$, then $\MMM_1$, and so on.
If we do this up to $\MMM_t$, the non-trivial elements of $\isg{t}$ are exactly the stabilizers derived from the stabilizing regions for this diagram.
Below we show this for timesteps $t \in \{-1, 0, 1\}$ of the distance-two repetition code.
This is a $\qcode{2, 1, 1}$ stabilizer code defined by $\MMM = [\langle Z_1 Z_2 \rangle]$.
Its ISG $\isg{t}$ is thus trivial for $t < 0$ and $\langle m Z_1 Z_2 \rangle$ for $t \geq 0$, for some measurement outcome $m$.
\begin{equation}\label{eq:rep_code_stabilizing_regions}
    \begin{aligned}
        t = -1:\quad
        \includegraphics[width=50pt, valign=c]{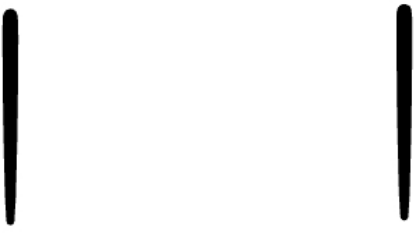}
        \qqqquad
        t = 0:
        \quad
        \includegraphics[width=50pt, valign=c]{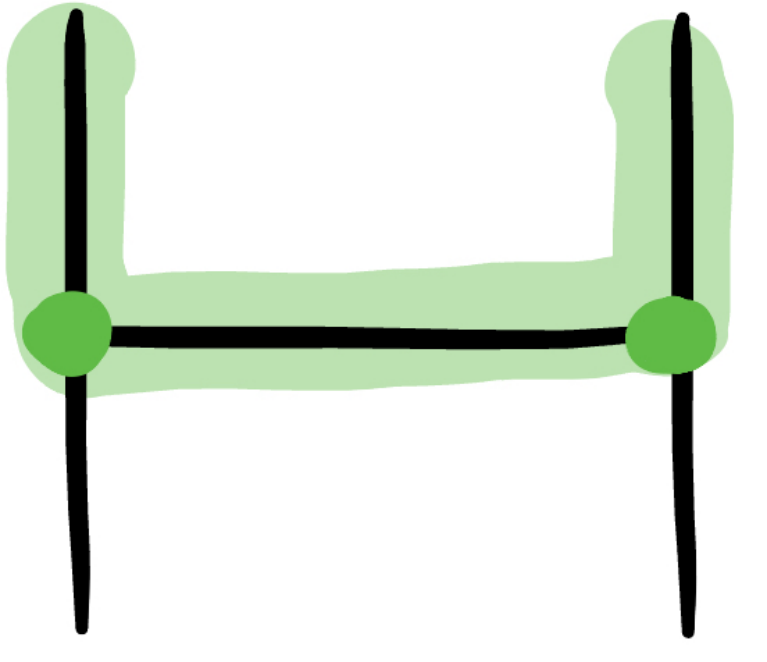}
        \qqqquad
        t = 1:
        \quad
        \includegraphics[width=50pt, valign=c]{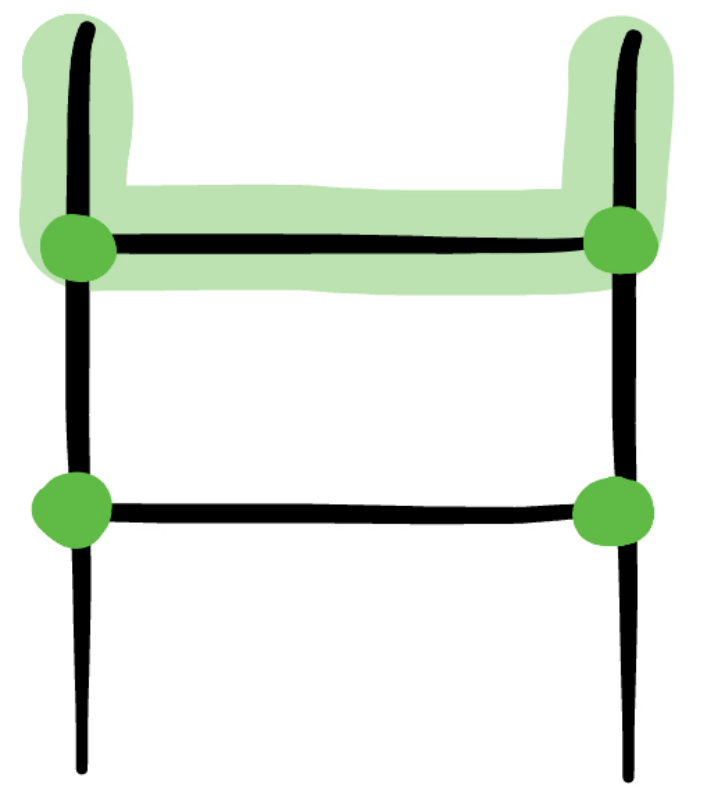}
    \end{aligned}
\end{equation}
Logical operators get a very similar treatment.
We can define an \textit{(unsigned CSS) operating region} to be a Pauli web in which
at least one input \textit{and} output node of the diagram are incident to highlighted edges.
If the diagram has $n$ output legs, the corresponding operator $p \in \PPP_n$ is again found by looking at the output wires,
in exactly the same way as for a stabilizer above.
Given any ISG code,
if we again draw a ZX-diagram that corresponds to sequentially measuring generating sets for $\MMM_0$ up to $\MMM_t$,
then representatives of non-trivial elements of $\lpg{t}$ are exactly the operators derived from the operating regions for this diagram.
The distance-two repetition code has
$\lpg{t} = \langle \ol{i}, \ol{Z_1}, \ol{X_1}, \ol{Z_2}, \ol{X_2} \rangle \cong \PPP_2$ for $t < 0$, and
$\lpg{t} = \langle \ol{i}, \ol{Z_1}, \ol{X_1 X_2} \rangle \cong \PPP_1$ for $t \geq 0$.
Below, we show operating regions at time $t=1$:
\begin{equation}\label{eq:rep_code_operating_regions}
    \begin{aligned}
        t = 1:
        \qquad
        \includegraphics[width=50pt, valign=c]{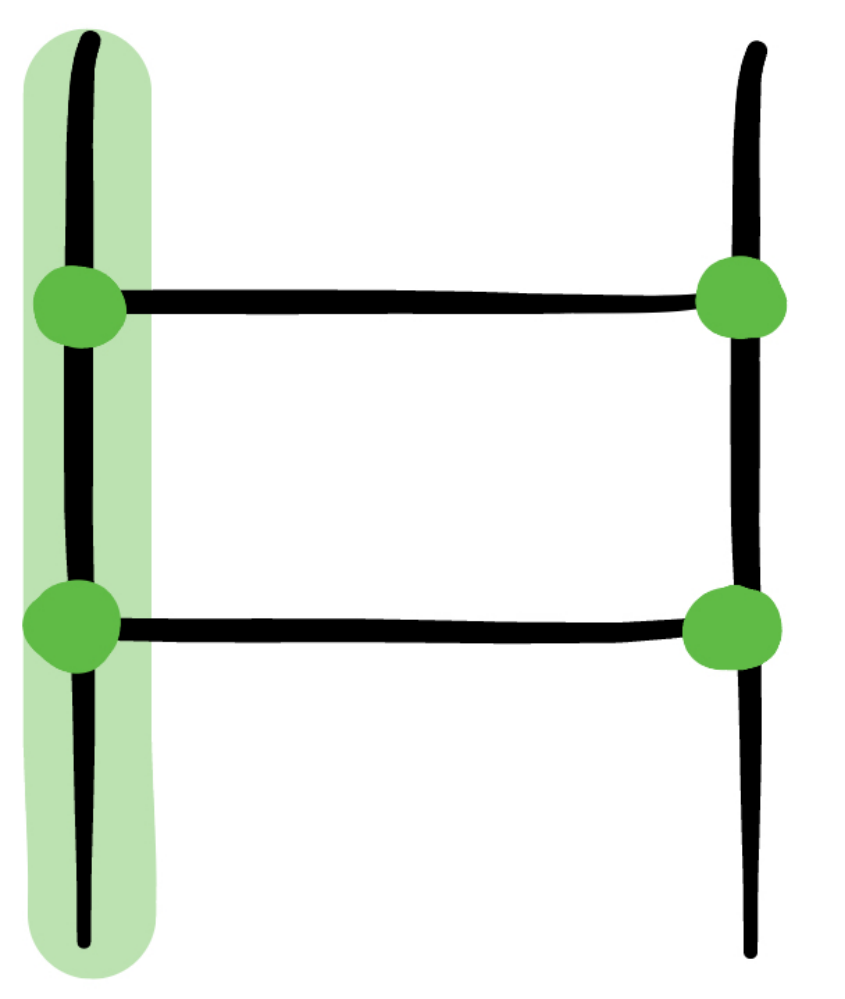}
        \qquad
        \text{and}
        \qquad
        \includegraphics[width=50pt, valign=c]{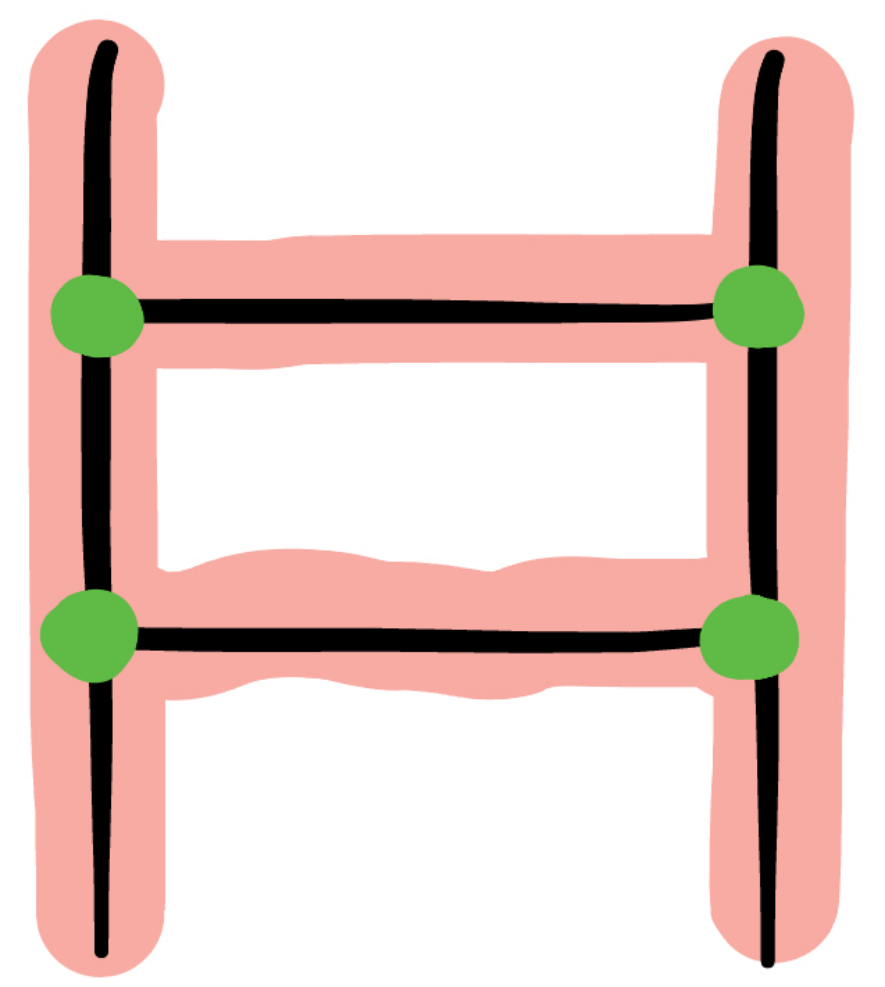}
    \end{aligned}
\end{equation}
We close this preliminary section with the comment that detectors and logical operators together provide an alternative view of an ISG code.
That is, one can think of such a code not as a group-theoretic object,
but as a Clifford ZX-diagram that is suitably covered by detecting regions,
and contains a non-empty set of pairs of operating regions whose corresponding operators satisfy the Pauli commutativity relations.
This corresponds to the unifying view of fault-tolerance put forward recently in \Ccite{ZxFaultTolerancePsiQuantum},
and is in the same spirit as the \textit{spacetime codes} of \Ccite{DelfosseSpacetimeCodes}.

    \section{Floquetifying the $\mathbf{\qcode{4, 2, 2}}$ code}\label{sec:floquetifying_422}

Let's warm up with one of the simplest interesting codes around: the $\qcode{4, 2, 2}$ code.
This is a stabilizer code, which we'll define as
$\MMM = [\langle Z_1 Z_2 Z_3 Z_4 \rangle, \langle X_1 X_2 X_3 X_4 \rangle]$.
The aim of this section is to prove the following:
\begin{theorem}\label{thm:4_2_2_double_hexagon_equivalence}
    The $\qcode{4, 2, 2}$ stabilizer code is equivalent as a ZX-diagram to a $\qcode{12, 2, 2}$ Floquet code with period 6,
    which we call the \textbf{double hexagon code}.
\end{theorem}
In \Cref{fig:4_2_2_rewritten}, we show three equivalent ZX-diagrams depicting seven timesteps of this code.
Here the grey squares, grey numbers, wire colours and wire styles (solid/dashed) have no meaning in the ZX-calculus;
they're just visual aids.
The grey squares denote timesteps of the $\qcode{4, 2, 2}$ code, and the grey numbers and coloured/styled lines will be explained shortly.
The leftmost diagram is the most natural one;
it shows measurements of $Z_1 Z_2 Z_3 Z_4$ and $X_1 X_2 X_3 X_4$ alternating at each timestep.
The second diagram is obtained from the first by unfusing every spider in the center of a grey square into two spiders,
and unfusing every spider in the corner of a grey square into three spiders.
The third is identical to the second in the ZX-calculus - all we've done is coloured and styled certain wires,
and labelled all one-legged spiders and black wires with an integer.
Now, in the leftmost diagram,
we interpret the four vertical lines as the world-lines of the four qubits of the code.
But we need not do this!
The ZX-diagram remains equivalent if we choose to interpret different wires as qubit world-lines.
This is exactly what the colours and styles in the rightmost diagram are for;
each colour-style pair denotes a different qubit world-line.
Since there are twelve world-lines, we're now viewing this as a system of twelve qubits, rather than four.
\begin{figure}[t]
    \centering
    \includegraphics[width=340pt]{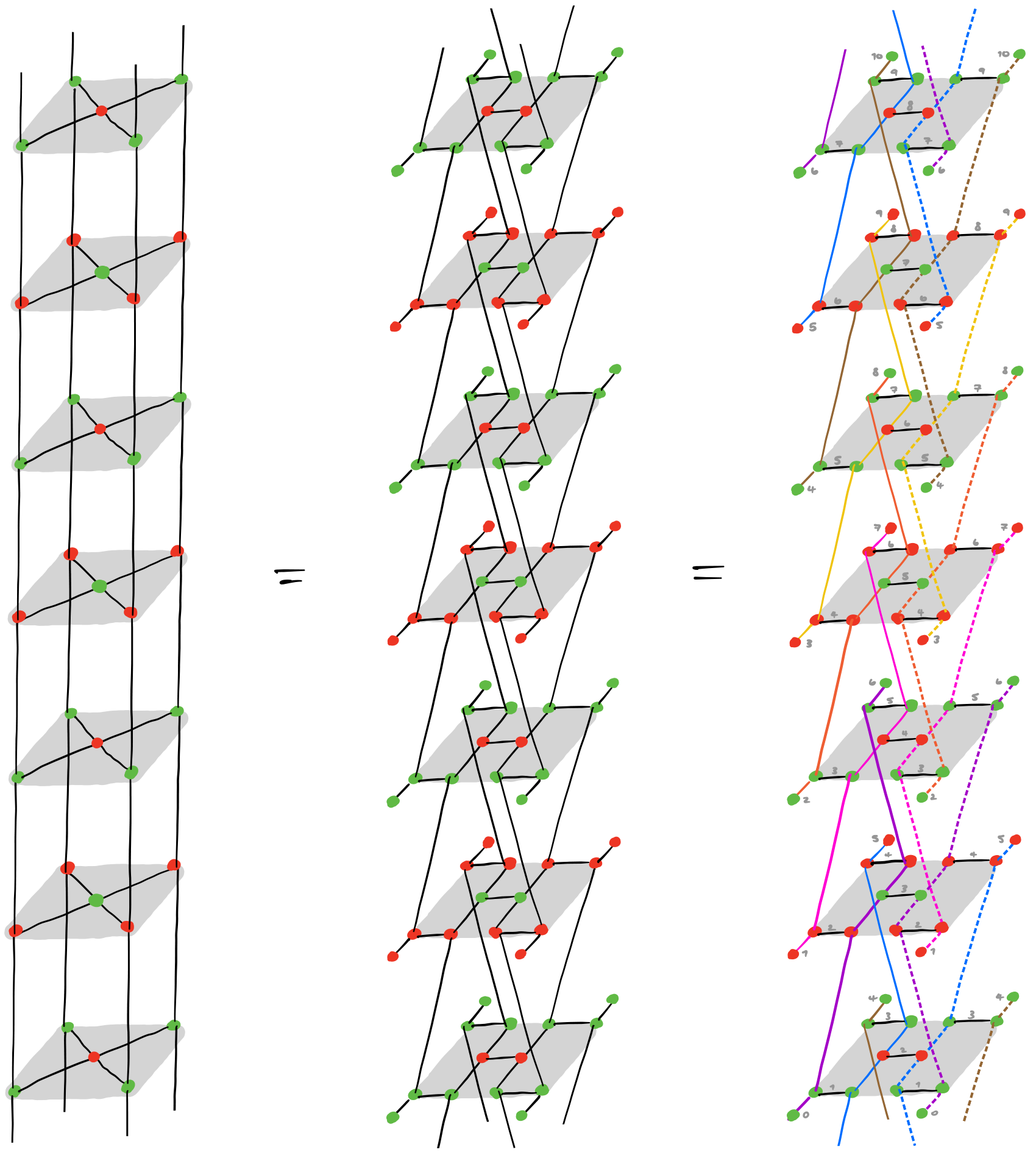}
    \caption{Three equivalent ZX-diagrams for seven timesteps of the $\qcode{4, 2, 2}$ code.}
    \label{fig:4_2_2_rewritten}
\end{figure}

Let's make some observations about this rightmost diagram.
Firstly, if we follow the world-line of any particular qubit up the page,
the integer labels incident to it form an increasing sequence.
For example, starting from the bottom of the diagram and following the solid purple qubit upwards,
the integers incident to it form the sequence $[0, 1, \ldots, 6]$.
We can thus think of these integers as a new set of timesteps for this diagram.
Next, notice that we can interpret all the uncoloured black wires like
\ \includegraphics[height=1\baselineskip, valign=c]{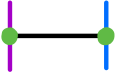}
\ and
\ \includegraphics[height=1\baselineskip, valign=c]{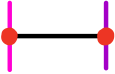}
\ as weight two $Z_u Z_v$ and $X_u X_v$ measurements (respectively) between qubits.
Furthermore, the qubit world-lines only have limited interactions with each other via these measurements.
Specifically, if we define ordered lists
$\mathbf{colours} = [
\mathbf{purple}, ~
\mathbf{pink}, ~
\mathbf{orange}, ~
\mathbf{yellow}, ~
\mathbf{brown}, ~
\mathbf{blue}]$
and
$\mathbf{styles} = [
\mathbf{solid}, ~
\mathbf{dashed}]$,
and let qubit $(i, j)$ denote the qubit with the $i$-th colour and $j$-th style,
where $i$ and $j$ are taken modulo 6 and 2 respectively,
then looking closely we see that qubit $(i, j)$ is only ever involved in
a measurement with the three qubits $(i+1, j), (i-1, j)$ and $(i, j+1)$.
So supposing we now wanted to lay out these qubits on a planar 2D chip,
a natural geometry would be a `double hexagon', as in the rightmost diagram of \Cref{fig:detector_4_2_2_double_hexagon}.

In fact, recalling the diagrammatic equation\
\includegraphics[height=1.1\baselineskip, valign=c]{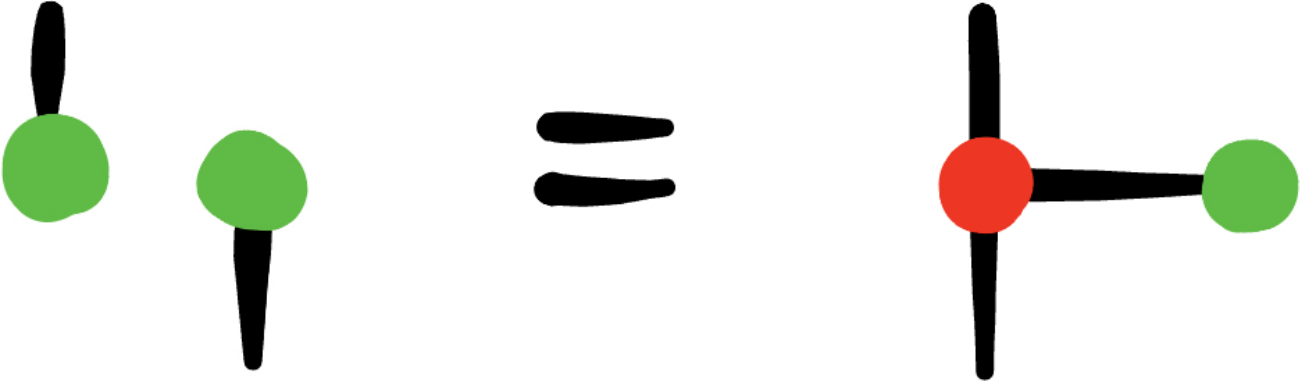}\ ,
which says non-destructive single qubit Pauli measurements disconnect wires,
we can also interpret all one-legged spiders as one half of such a measurement.
This interpretation is valid,
in that the timestep at which the world-line of qubit $(i, j)$ `ends' at a one-legged spider
is the same as the timestep at which it `resumes' via another one-legged spider further up the page.
Finally, we can see that this pattern of colours and styles repeats itself every six timesteps.

\begin{figure}[b]
    \centering
    \includegraphics[width=330pt]{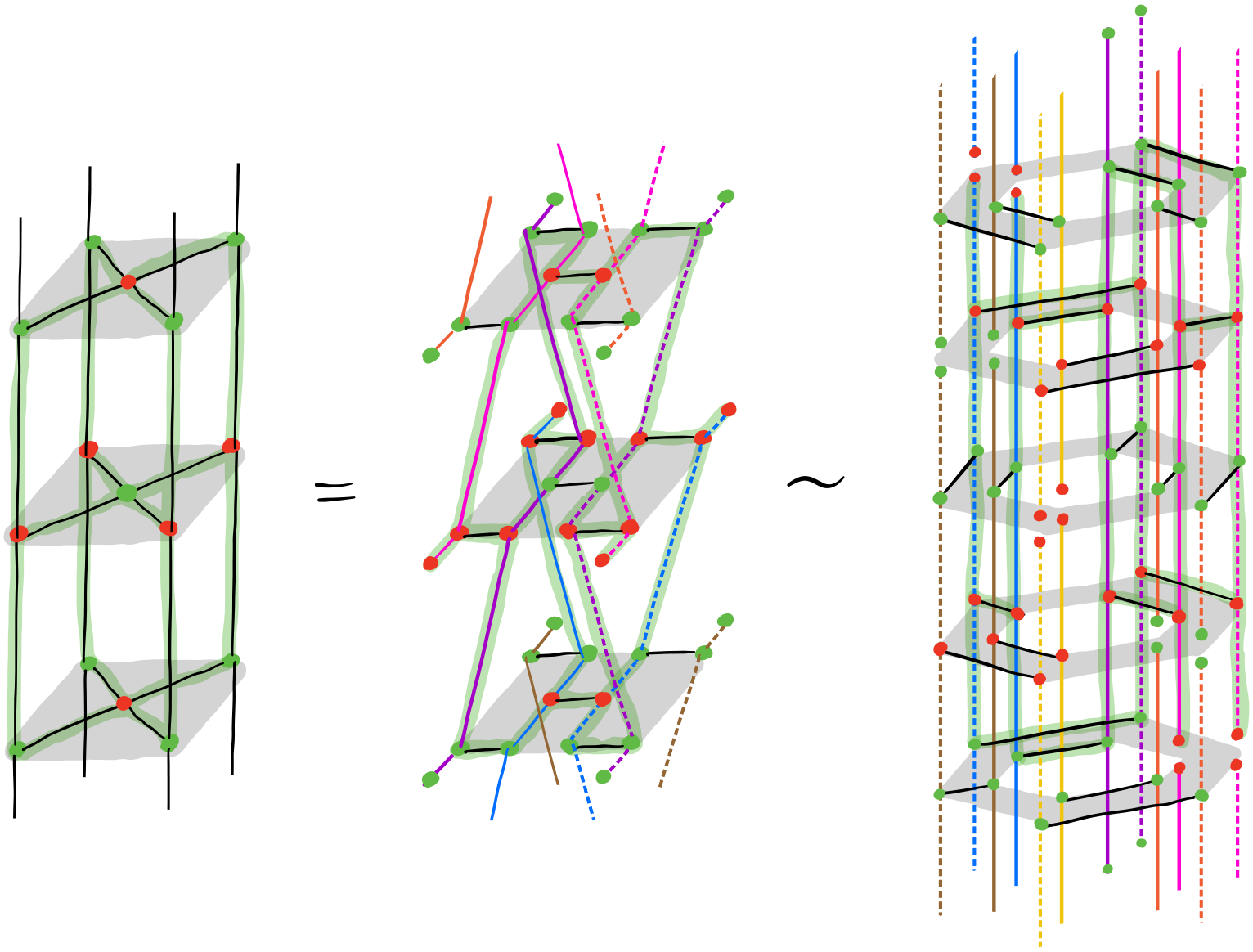}
    \caption{
        A detecting region in the $\qcode{4, 2, 2}$ code and its image in the double hexagon code.
        We use the ($\sim$) symbol between the second and third diagrams rather than an equals sign because,
        although the two codes can be viewed as having equivalent ZX-diagrams,
        these two particular subdiagrams are not equivalent -
        the rightmost one has more input and output wires, for example.}
    \label{fig:detector_4_2_2_double_hexagon}
\end{figure}

From this analysis, we can now interpret the rightmost ZX-diagram as an ISG code;
we can write down the measurements that each qubit undergoes at each timestep,
and consequently we can define the measurement schedule described by this diagram.
We get $\MMM = [\MMM_0, \ldots, \MMM_5]$, where:
\begin{equation}\label{eq:floquetified_4_2_2_measurement_schedule}
    \MMM_t = \stretchleftright[1000]
        {<}
            {\begin{array}{c}
                 X_{(t, 0)}, \\
                 X_{(t, 1)}, \\
                 X_{(t+1, 0)}X_{(t+2, 0)}, \\
                 X_{(t+1, 1)}X_{(t+2, 1)}, \\
                 X_{(t+3, 0)}X_{(t+3, 1)}, \\
                 X_{(t+4, 0)}X_{(t+5, 0)}, \\
                 X_{(t+4, 1)}X_{(t+5, 1)} \\
            \end{array}}
        {>}
    \text{~~if $t$ even,}
    \qquad\qquad
    \MMM_t = \stretchleftright[1000]
        {<}
        {\begin{array}{c}
             Z_{(t, 0)}, \\
             Z_{(t, 1)}, \\
             Z_{(t+1, 0)}Z_{(t+2, 0)}, \\
             Z_{(t+1, 1)}Z_{(t+2, 1)}, \\
             Z_{(t+3, 0)}Z_{(t+3, 1)}, \\
             Z_{(t+4, 0)}Z_{(t+5, 0)}, \\
             Z_{(t+4, 1)}Z_{(t+5, 1)} \\
        \end{array}}
        {>}
    \text{~~if $t$ odd.}
\end{equation}
One can then calculate the group $\isg{t}$, and consequently $\lpg{t}$.
It turns out $\isg{t}$ is established whenever $t \geq T = 3$,
and can be minimally generated by the seven generators of $\MMM_t$,
plus three weight-six Paulis $s_{t-1}, s_{t-2}$ and $s_{t-3}$, where:
\begin{equation}\label{eq:floquetified_4_2_2_s_t}
    s_t = \begin{cases}
              X_{(t, 0)}X_{(t, 1)}X_{(t-1, 0)}X_{(t-1, 1)}X_{(t-2, 0)}X_{(t-2, 1)} &\text{ if $t$ even} \\
              Z_{(t, 0)}Z_{(t, 1)}Z_{(t-1, 0)}Z_{(t-1, 1)}Z_{(t-2, 0)}Z_{(t-2, 1)} &\text{ if $t$ odd} \\
    \end{cases}
\end{equation}

Since we have 10 independent generators on 12 qubits, we can conclude that $\lpg{t} \cong \PPP_2$ for all $t \geq 3$.
This then proves most of \Cref{thm:4_2_2_double_hexagon_equivalence};
namely that the double hexagon code encodes 2 logical qubits and has period 6.
The proof that the distance of the new Floquet code remains two is deferred to \Cref{sec:double_hexagon_code}.

In addition to being able to work with detectors, stabilizers and logical operators algebraically, as above,
the mapping of these objects from the $\qcode{4, 2, 2}$ code to the double hexagon code
can be seen graphically via Pauli webs.
In \Cref{fig:detector_4_2_2_double_hexagon} we show a detecting region in the $\qcode{4, 2, 2}$ code
and its image in the double hexagon code.
From the rightmost diagram of this figure,
and recalling the rules for mapping a detecting region to a detector from \Cref{subsubsec:detectors},
one can see that the corresponding detector in the double hexagon code consists of eight measurements.
Specifically, in the bottom layer, we include two $Z \otimes Z$ measurements between blue and purple qubits,
and two single qubit $Z$ measurements on pink qubits.
In the top layer, we include two $Z \otimes Z$ measurements between purple and pink qubits,
and two single qubit $Z$ measurements on blue qubits.
Since every detecting region in the $\qcode{4, 2, 2}$ code is equivalent to the one on the left of this figure
(up to a space-time translation and exchanging the roles of $Z$ and $X$),
every detecting region in the double hexagon code is equivalent to the one on the right
(again up to a space-time translation and $Z \leftrightarrow X$ interchange).

One could justifiably point out here that we seem to have made things worse;
we've taken a $\qcode{4, 2, 2}$ ISG code and turned it into a $\qcode{12, 2, 2}$ ISG code,
and what's more, each detector now consists of eight measurements rather than two, so would seem to be noisier!
The trade-off is that now every measurement is weight-two or weight-one, rather than weight-four.
As a general rule, the higher the measurement weight, the noisier it will be.
In particular, weight-one and weight-two measurements can be performed natively in some architectures,
whereas higher-weight measurements are implemented via extraction circuits,
which give more opportunities for noise to interfere.

On a higher-level, one can view this Floquetification process as a reinterpretation of the time direction
in a ZX-diagram.
Below, we use two blue prisms (square and hexagonal)
as abstractions of ZX-diagrams for the $\qcode{4, 2, 2}$ code and double hexagon code respectively.
In grey we show how a timeslice in the double hexagon code corresponds to an angled slice of the $\qcode{4, 2, 2}$ code:
\begin{equation}\label{eq:timeslice_4_2_2_double_hexagon}
    \begin{aligned}
        \includegraphics[width=230pt]{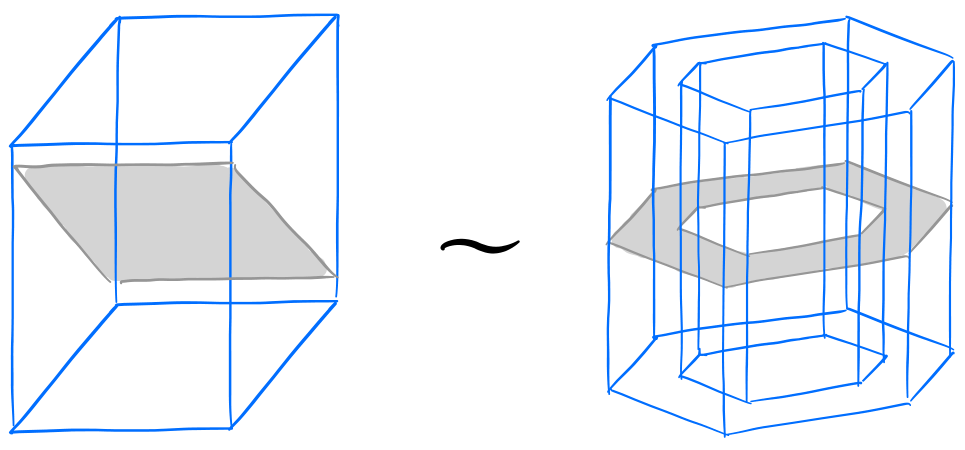}
    \end{aligned}
\end{equation}

    \section{Floquetifying the colour code}\label{sec:floquetifying_colour_code}

We can apply the same ideas as in the last section to stabilizer codes
more interesting than the $\qcode{4, 2, 2}$ code.
In this section, we take this more interesting stabilizer code to be the colour code.
Or, more accurately, we take it to be the \textit{bulk} of the colour code -
i.e.\ ignoring the code's global topology (whether it lives on a torus, or is planar).
A discussion of global topology is deferred to
\Cref{subsec:floquetifying_colour_code_beyond_bulk} at the end of this section,
and continued in detail in \Cref{sec:floquetifying_colour_code_beyond_bulk}.

\subsection{The bulk}\label{subsec:floquetifying_colour_code_bulk}

The $\mathit{6.6.6}$ \textit{colour code} is a stabilizer code defined on a honeycomb lattice, with qubits placed at vertices.
We write $v \in f$ to mean that a vertex $v$ is incident to a hexagonal face $f$.
For any such face $f$, we define weight-six Paulis $X_f = \prod_{v \in f} X_v$ and $Z_f = \prod_{v \in f} Z_v$.
On a torus, the code is defined by the measurement schedule
$\MMM = [\langle \{Z_f: \text{face $f$}\} \rangle,~ \langle \{X_f: \text{face $f$}\} \rangle]$.
On a planar geometry, slightly different Paulis are measured at the boundaries~\cite{ColourCodeBoundariesAndTwists}.

\begin{figure}[t]
    \centering
    \includegraphics[width=\linewidth]{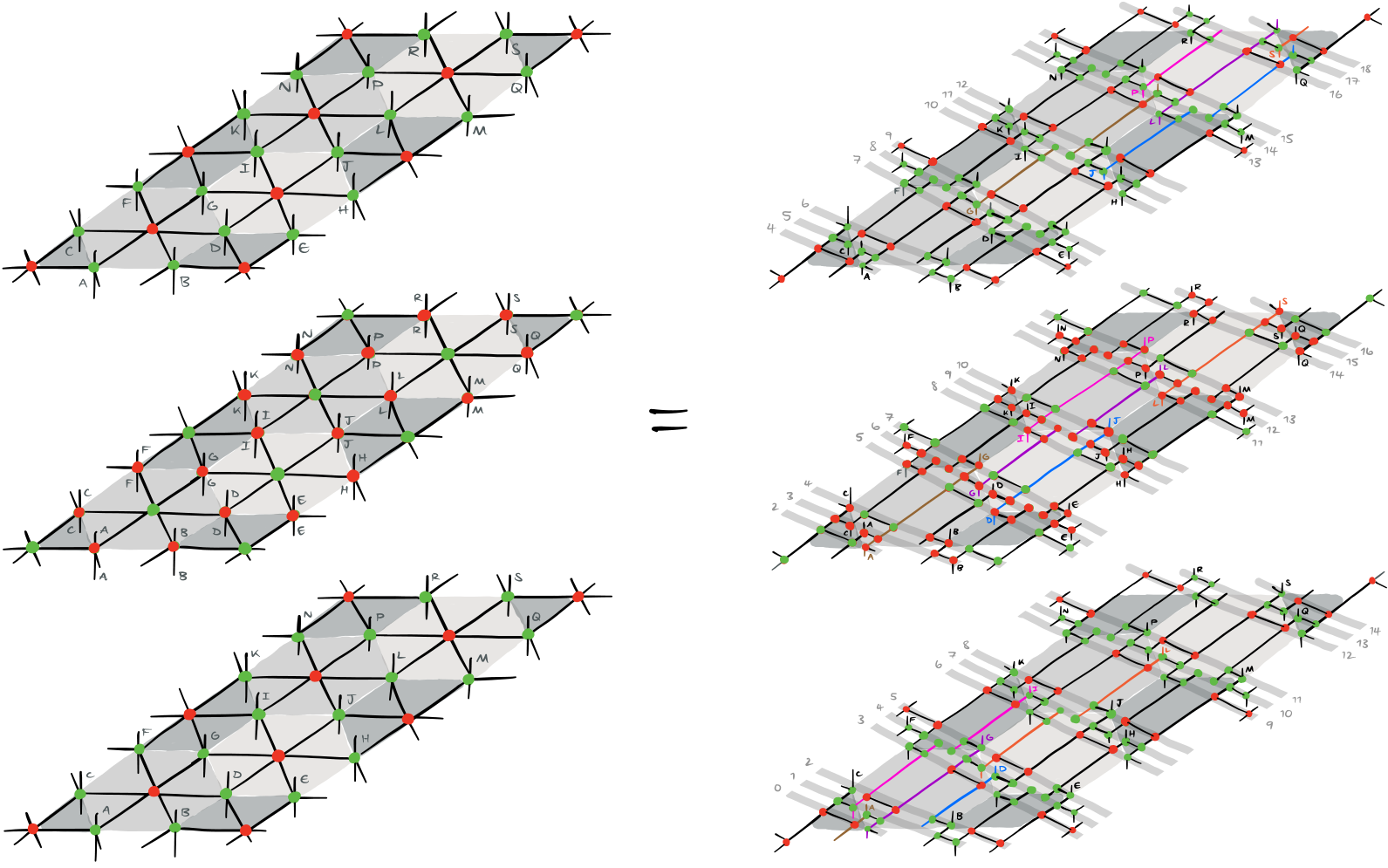}
    \caption{
        Two equivalent ZX-diagrams for a patch of the colour code at three consecutive timesteps.
        Ideally we'd draw vertical wires connecting each column of patches,
        like we were able to do for the $\qcode{4, 2, 2}$ code in the last section,
        but doing this renders the diagrams pretty much unreadable.
        Instead, we draw very small vertical wires going up and/or down from certain spiders, labelled by letters.
        These letters define how a wire going vertically upwards in one layer is actually connected to
        a wire coming vertically downwards from the layer above it;
        wires labelled by the same letter are connected.
    }
    \label{fig:colour_code_rewritten}
\end{figure}

As before, we start with a ZX-diagram of the colour code over multiple timesteps;
see the left hand side of \Cref{fig:colour_code_rewritten}.
We then unfuse every spider into four or five spiders to get the diagram on the right of the figure.
In these diagrams, grey hexagons, bars and integers, as well as letter labels and coloured wires, all have no meaning in the ZX-calculus; they're just visual aids.
In the diagram on the right, wires within grey bars correspond to weight-two measurements,
while all other wires correspond to qubit world-lines.
The focus is on a single qubit's world-line, which we've coloured purple;
we can see that it's only involved in measurements with four other qubits,
which we've coloured orange, pink, blue and brown.
It turns out that all qubit world-lines have this property of only interacting with four other qubits.
Furthermore, these interactions are such that the qubits of the new code can be laid out on a square lattice.
So henceforth we'll label qubits of the new code with a pair of integer coordinates $(x, y)$.

The grey bars labelled by integers denote the timesteps of the new code.
These are well-defined;
picking any qubit world-line and following it up the page
while noting down the integer label of every grey bar it crosses produces an increasing sequence.
For example, doing this for the purple qubit produces the contiguous sequence $[0, 1, 2, \ldots]$.
As before, we can view all one-legged spiders as halves of single-qubit measurements.
The pattern repeats after every two timesteps of the old code (the colour code);
one can see this by noting that the bottom and top of the diagram are identical,
up to a translation in space and an increase by 13 in the labels of the grey bars.
In other words, this new code has a period of 13.

\begin{figure}[t]
    \centering
    \includegraphics[width=450pt]{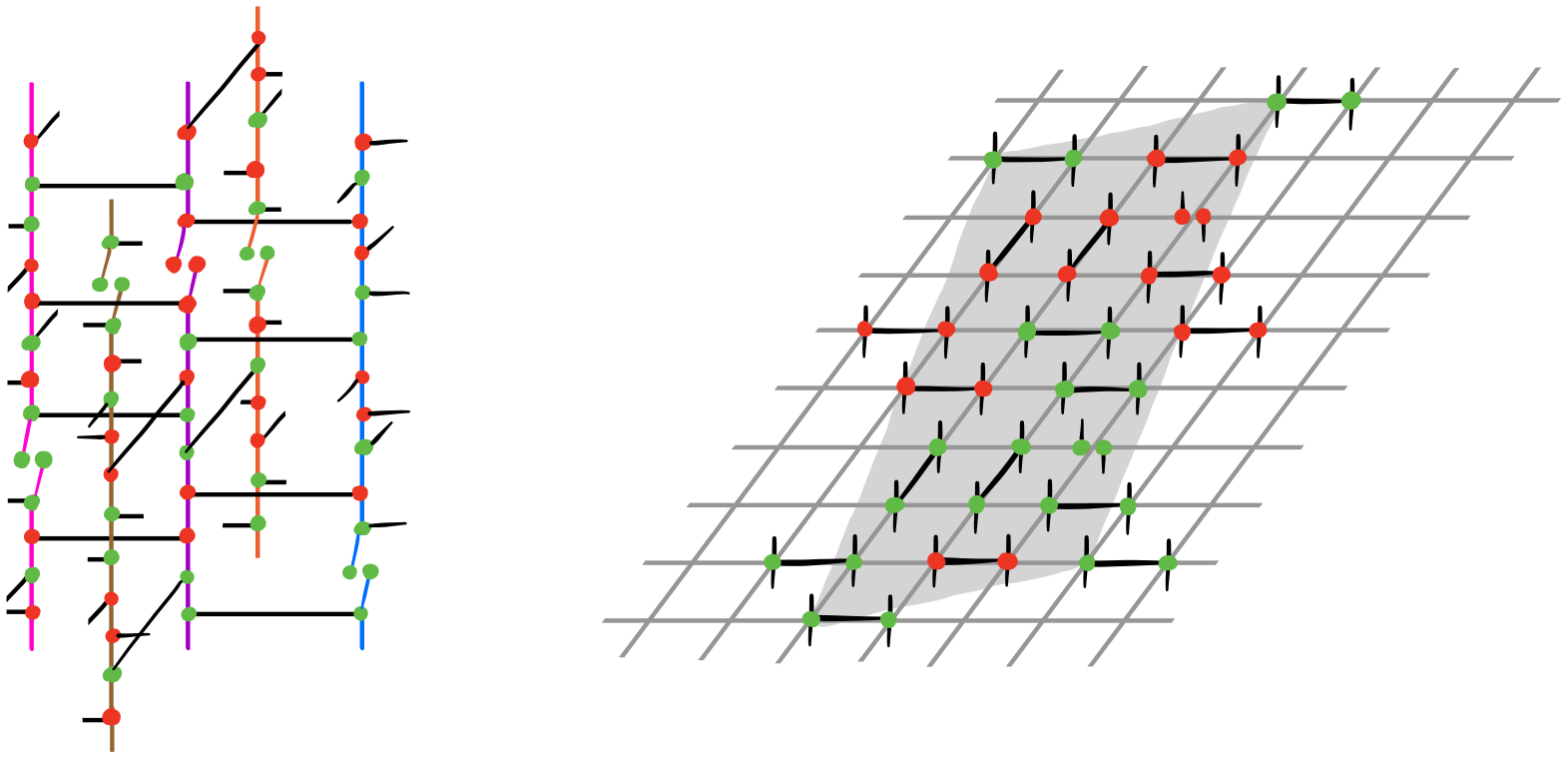}
    \caption{
        Left: measurements undergone by a single qubit and its four nearest neighbours over one full period of 13 timesteps.
        Right: One `tile' of one timestep of the measurement schedule in the bulk of the Floquetified colour code.
        Here the grey lines are \textit{not} ZX-diagram wires -
        they are just visual aids that show this all sits on a square lattice.}
    \label{fig:floquetified_colour_code_measurement_schedule}
\end{figure}
Again, we can now write down the measurements that each qubit undergoes at each timestep.
In \Cref{fig:floquetified_colour_code_measurement_schedule},
we show a ZX-diagram of this, focused on the purple qubit from \Cref{fig:colour_code_rewritten}.
Each qubit undergoes essentially the same pattern of measurements in each period.
Specifically, whatever measurement qubit $(x, y)$ undergoes at timestep $t$,
qubit $(x, y+1)$ undergoes it at time $t-2$, but with the roles of $Z$ and $X$ exchanged.
Likewise for the remaining neighbours $(x+1, y)$, $(x, y-1)$ and $(x-1, y)$,
but at times $t-8$, $t+2$ and $t+8$ respectively.
Knowing this, we can then write down the measurement schedule for the bulk of the new code
(i.e.\ what measurements are happening at any single timestep).
This has a periodic structure; we draw a ZX-diagram for a single timestep $t$ and single `tile' of this on the right of \Cref{fig:floquetified_colour_code_measurement_schedule}.
To see the measurements happening across the whole bulk at this timestep, one should tile these grey rectangles across the square lattice.
That is, whatever measurement qubit $(x, y)$ undergoes at time $t$, qubits $(x+3, y+1)$ and $(x-2, y+8)$ undergo this too.
Then to get the measurements happening at the next timestep, one should translate the measurements from time $t$ by $(-1, -3)$.
That is, whatever measurement qubit $(x, y)$ undergoes at time $t$, qubit $(x-1, y-3)$ undergoes it at time $t+1$.
\begin{figure}[t]
    \centering
    \includegraphics[width=450pt]{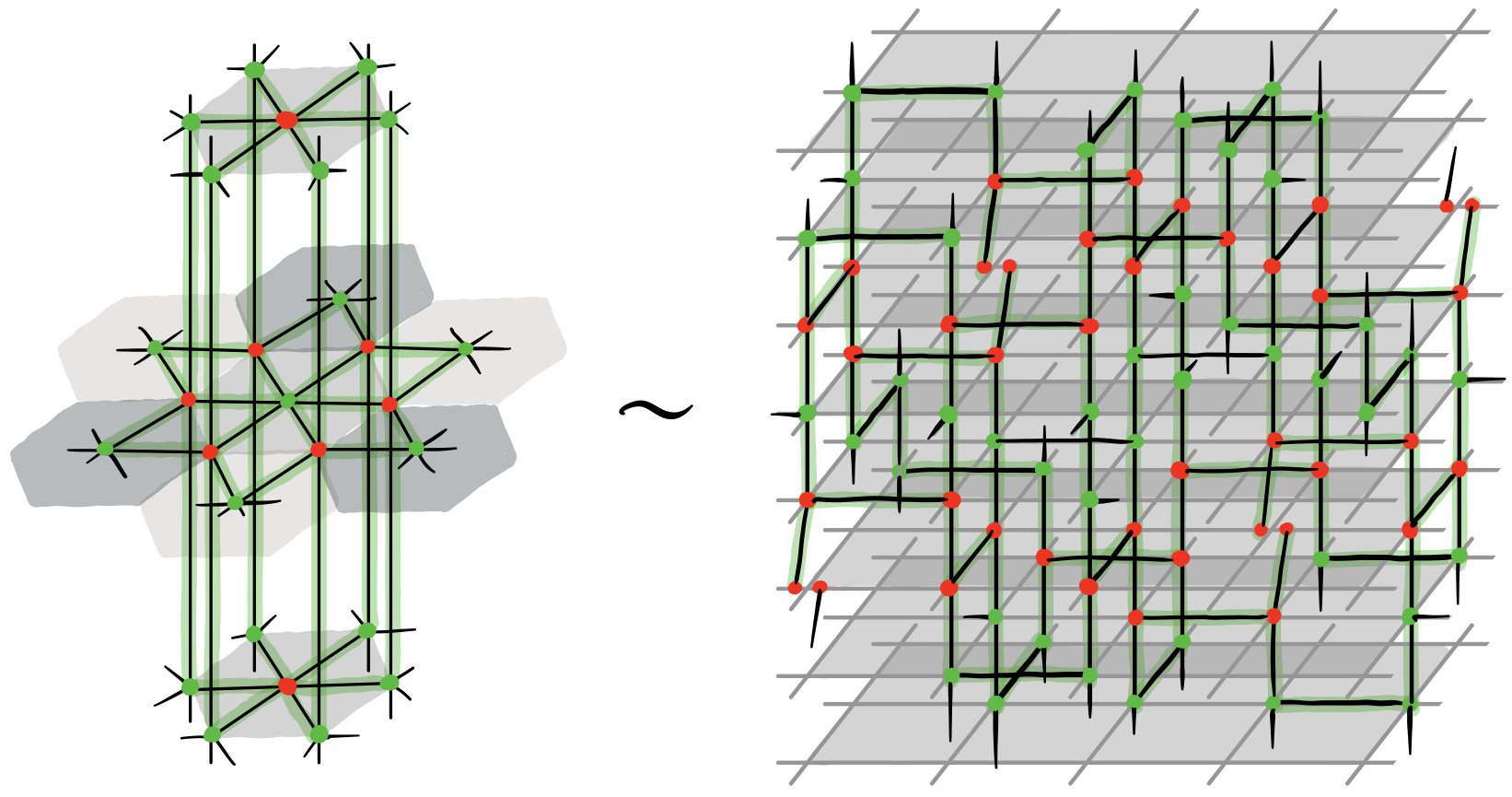}
    \caption{
        A detecting region in the colour code and the corresponding detecting region in the Floquetification.
        Again, grey lines are \textit{not} ZX-diagram wires -
        they are just a visual guide showing the square lattice.}
    \label{fig:detector_floquetified_colour_code}
\end{figure}

Detectors and logical operators can again not only be worked with algebraically,
but also graphically, via Pauli webs.
In \Cref{fig:detector_floquetified_colour_code} we show a detecting region in the colour code bulk and its image in the bulk of the new Floquet code.
The corresponding detector in the new code consists of 14 weight-two measurements and 4 single-qubit measurements, spread over 8 timesteps.
Every detector in the bulk of the new code is identical to this one,
up to a space-time translation and exchanging $Z$ and $X$.
Just as in the colour code, these detecting regions are tiled such that unique errors
violate unique sets of detectors, so decoding can be performed,
though we leave investigating specific decoding strategies to future work.
A similar exercise can be repeated for the logical operators;
one can draw the operating regions corresponding to the known logical operators of the colour code,
and see how these map to operating regions in the Floquetified code,
from which one can write down the new code's logical operators.

\subsection{Beyond the bulk}\label{subsec:floquetifying_colour_code_beyond_bulk}

The Floquetification process described above can be applied directly to planar colour codes,
and will lead to a new code that is itself planar.
But since the boundaries of the original code look different to the bulk,
extra work needs to be done to Floquetify these correctly.
Furthermore, the resulting Floquet code can exhibit a `drifting' behaviour,
which we discuss in more detail in \Cref{subsec:drift}.
There are potential perks of this behaviour - e.g.\ for removing leakage -
but it's also handy to have a code that doesn't drift.
One might think we could get around these two issues by starting with the colour code defined on a torus,
which has no boundaries to worry about.
But viewing our procedure as tilting the time direction in a ZX-diagram for a code,
as described in the last section,
we find we can only give well-defined new timesteps when the code we start with is planar -
this is discussed in more depth in \Cref{subsec:periodic_boundaries}.
Instead, if we want to avoid boundaries,
what we can do is Floquetify only the bulk,
which will give the bulk of a potential new Floquet code,
then see if placing this new bulk on a torus still encodes logical qubits.
In \Cref{subsec:colour_code} we apply the two ideas described above -
we Floquetify a planar colour code,
and place the Floquetified colour code bulk from this section on a torus.

    \section{Conclusion and future work}\label{sec:conclusion}
In this work, we introduced ISG codes,
which describe a large family of codes driven by sequentially measuring sets of Pauli operators;
this includes stabilizer, subsystem\footnote{%
    Up to the caveat that our definition of a subsystem code 
    as in \Cref{defn:subsystem_code_as_ISG_code}
    differs slightly from the usual subsystem code definition, 
    a point which we explore in \Cref{sec:subsystem_code_remarks}.
} and Floquet codes, and more.
We then used the ZX-calculus to find a new ISG code (specifically, a Floquet code) that is equivalent to the colour code,
and can be implemented on a square lattice.
The main disadvantage of the colour code versus the surface code is its high-weight measurements -
our construction removes this obstacle, since all its measurements are of weight one or two.
For it to be a genuine candidate for practical implementation,
we would need to investigate its decoding capabilities, its boundaries and its logical gates -
we leave these for future work.

One direction we find very interesting relates to the latter;
in \Ccite{AutomorphismCodes},
it is shown that Floquet codes can natively implement
certain logical Clifford gates fault-tolerantly at no extra effort.
The set of such implementable gates is restricted by the automorphisms of the code's anyonic defects;
we believe our Floquetified colour code should inherit a rich set of such automorphisms from the colour code.
We would like to verify whether this is the case,
and then see whether this can be leveraged to perform fault-tolerant logical gates.

As it happens, the authors of an upcoming paper~\cite{MargaritaFloquetifiedColourCode} do exactly this
for an independently discovered Floquet code that is also in a definable sense equivalent to the colour code.
Specifically, using the \textit{anyon condensation} framework of \Ccite{AnyonCondensation},
they construct a code that can implement the full logical Clifford group via sequences of measurements of weight at most three.
What's more, a further upcoming paper~\cite{TylerFloquetifiedColourCode}
contains a third independent `Floquetified colour code` construction,
this time by starting with a subsystem code (specifically, that of \Ccite{BombinRubyCodes})
and passing to an \textit{associated ISG code} (\Cref{defn:associated_ISG_code}).
We are excited to learn more about both of these works,
and to think about the connections between our various different constructions.
More generally, understanding our own Floquetified colour code in terms of
topological phases and symmetries of the ZX-diagram representing it
seems an exciting research avenue that bridges between the fields of
diagrammatic calculi, (topological) error correction and condensed matter physics.

Further work is warranted around ISG codes more broadly.
For example, there are interesting questions to be answered around equivalences of such codes,
how best to define a notion of distance on them,
and how subsystem codes as they're usually defined fit into this framework.
Though we go into more detail on these in \Cref{sec:isg_code_remarks,sec:subsystem_code_remarks},
a more thorough investigation would be welcome.
This could perhaps lead to interesting re-evaluations of long-established quantum codes.

The other obvious further research direction would be
to develop the ideas here into a full-blown general-purpose Floquetification algorithm,
which can take as input a stabilizer code (presumably satisfying certain conditions),
and return a Floquetified version of it, with all the advantages and trade-offs that brings;
e.g.\ lower-weight measurements, but more measurements per detector, typically.
One could then compare the performance of these stabilizer codes with their Floquetifications,
to see if any practical advantages emerge in the general case.

On a higher level, this work makes the same point as \Ccite{ZxFaultTolerancePsiQuantum},
in that it suggests \textit{static} stabilizer and subsystem codes are perhaps
not so different from \textit{dynamic} ISG codes (like Floquet codes) after all,
and that a unifying way to think of ISG codes could be as ZX-diagrams that are suitably covered by detecting regions
and contain pairs of operating regions satisfying Pauli commutativity relations.
An avenue for further work would be to develop this perspective further -
for example, by taking it as a starting point for designing new codes.

    \section{Acknowledgements}\label{sec:acknowledgements}
We wish to thank
Daniel Litinski and Fernando Pastawski above all, for chatting with us about
applications of the ZX-calculus to quantum error correction~\cite{ZxFaultTolerancePsiQuantum},
and to David Aasen for discussions about his automorphism codes~\cite{AutomorphismCodes}.
We also thank the reviewers from QPL 2023, who gave very helpful feedback on our initial submission,
as well as Jens Eisert, Peter-Jan Derks and Daniel Litinski,
who gave helpful feedback on the final draft.
We are grateful to the authors of \Ccite{MargaritaFloquetifiedColourCode} and \Ccite{TylerFloquetifiedColourCode}
for discussing their own Floquetified colour code constructions with us.
Alex wishes to thank Drew Vandeth, Anamaria Rojas and all those involved in organising the IBM QEC Summer School
for the inspiring four weeks that sowed the seeds of this project.
Julio wants to thank Tyler Ellison for inspiring discussion on Floquet codes in the colour code phase.
This work was supported by
the Einstein Foundation (Einstein Research Unit on quantum devices),
the DFG (CRC 183),
the Munich Quantum Valley (K8) and the BMBF (RealistiQ, QSolid).

    \newcommand{\etalchar}[1]{$^{#1}$}

    \clearpage
    \appendix
    \section{A corollary of the stabilizer formalism}\label{sec:normalizer_formalism}

Here we state and prove the effect that measuring a Hermitian Pauli $p \in \PPP_n$ has on the group $N(\SSS) / \SSS$,
for any stabilizer group $\SSS \leq \PPP_n$.
We will often refer to this as (measurement in) the \textit{normalizer formalism}.
A few lemmas will be required in order to prove this;
the first two are so fundamental that we will use them without explicitly referencing them.

\begin{lemma}\label{lem:paulis_commute_or_anticommute}
    Any two elements $p, q \in \PPP_n$ either commute or anti-commute.
    That is, $pq = \pm qp$.
\end{lemma}

\begin{corollary}\label{cor:stabilizer_normalizer_is_centralizer}
    If $\SSS \leq \PPP_n$ is a stabilizer group, the normalizer $N(\SSS) = \{p \in \PPP_n : p \SSS p^{-1} = \SSS\}$
    is equal to the centralizer $C(\SSS) = \{p \in \PPP_n : \forall s \in \SSS, psp^{-1} = s\}$.
\end{corollary}

The next two are less elementary; their proofs can be found in \Ccite[Section 3.4]{GottesmanQecLectureNotes}.

\begin{lemma}\label{lem:normalizer_iso_pauli_group}
    If $\SSS \leq \PPP_n$ is a stabilizer group with rank $r$,
    then $N(\SSS)/\SSS \cong \PPP_{k}$, where $k = n - r$.
\end{lemma}

\begin{lemma}\label{lem:normalizer_generator_choice}
    If $\SSS \leq \PPP_n$ is a stabilizer group with rank $r$, and $k = n-r$,
    then for any maximally Abelian subgroup $\langle~ \ol{x_1}, \ldots, \ol{x_k} ~\rangle$ of $N(\SSS)/\SSS \cong \PPP_k$ we might choose,
    there exists a second Abelian subgroup $\langle~ \ol{z_1}, \ldots, \ol{z_k} ~\rangle$, such that
    $\langle ~
        \ol{i},~
        \ol{x_1},~
        \ol{z_1},~
        \ldots,~
        \ol{x_k},~
        \ol{z_k} ~
    \rangle$
    is isomorphic to $\PPP_k = \langle i, X_1, Z_1, \ldots, X_k, Z_k \rangle$
    via the map $\ol{i} \mapsto i, \ol{x_j} \mapsto X_j$ and $\ol{z_j} \mapsto Z_j$ for all $j$.
\end{lemma}

First, we'll remind ourselves of how measurement works in the \textit{stabilizer formalism}.

\begin{theorem}[Measurement in the stabilizer formalism]\label{thm:stabilizer_formalism}
    Suppose we have a stabilizer group $\SSS = \langle s_1, \ldots, s_r \rangle \leq \PPP_n$
    with rank $r$, and let $k = n-r$.
    Measuring a Hermitian Pauli $p$ produces a \textit{measurement outcome} $m \in \{-1, 1\}$
    and a new stabilizer group $\SSS' \leq \PPP_n$.
    We have three cases:

    \hypertarget{stabilizer_formalism_case_1}{\textbf{Case 1}}
    (only possible when $r < n$): $p$ commutes with all generators $s_j$ but $\pm p \notin \SSS$.
    In this case, the measurement outcome $m \in \{1, -1\}$ is random,
    and $\SSS' = \langle m p, s_1, s_2, \ldots, s_r \rangle$.

    \hypertarget{stabilizer_formalism_case_2}{\textbf{Case 2:}}
    $p$ commutes with all generators $s_j$ and $\pm p \in \SSS$.
    Here the outcome $m$ is deterministically $\pm 1$, and $\SSS' = \SSS$.

    \hypertarget{stabilizer_formalism_case_3}{\textbf{Case 3:}}
    $p$ anti-commutes with at least one $s_j$.
    In fact, it can be shown that we can always pick a generating set such that $p$ anti-commutes with exactly one generator $s_1$,
    and commutes with the remaining generators $s_2, \ldots, s_r$.
    In this case, the outcome $m$ is again random, and $\SSS' = \langle m p, s_2, \ldots, s_r \rangle$.
\end{theorem}

We can then define measurement in the normalizer formalism via the same three cases.

\begin{theorem}[Measurement in the normalizer formalism]\label{thm:normalizer_formalism}
    Given a stabilizer group $\SSS = \langle s_1, \ldots, s_r \rangle \leq \PPP_n$ with rank $r$,
    and letting $k = n-r$, we know that:
    \begin{equation}\label{eq:normalizer_formalism_presentation}
        N(\SSS) / \SSS = \langle ~
            \ol{i},~
            \ol{x_1},~
            \ol{z_1},~
            \ol{x_2},~
            \ol{z_2},~
            \ldots,~
            \ol{x_k},~
            \ol{z_k} ~ \rangle
        \cong \PPP_k
    \end{equation}
    for some Paulis $x_j$ and $z_j$ in $N(S)$.
    Measuring a Hermitian Pauli $p \in \PPP_n$
    produces a measurement outcome $m \in \{1, -1\}$ and a new stabilizer group $\SSS'$;
    we can describe the new logical Pauli group $N(\SSS') / \SSS'$ in terms of the old one $N(\SSS) / \SSS$
    according to the same three cases:

    \hypertarget{normalizer_formalism_case_1}{\textbf{Case 1}}
    (only possible when $r < n$): $p$ commutes with all generators $s_j$ but $\pm p \notin \SSS$.
    Equivalently, $\ol{p}$ is an element of $N(\SSS) / \SSS$ other than $\ol{\one}$ or $-\ol{\one}$.
    By choosing a different generating set for $N(\SSS) / \SSS$ if needed,
    we may assume without loss of generality that $\ol{p} = \ol{x_1}$.
    We then find that:
    \begin{equation}\label{eq:normalizer_formalism_case_1}
        N(\SSS') / \SSS' = \langle ~
            \ol{i},~
            \ol{x_2},~
            \ol{z_2},~
            \ldots,~
            \ol{x_k},~
            \ol{z_k} ~ \rangle
        \cong \PPP_{k-1}
    \end{equation}
    where $\SSS'$ is now $\langle mp, s_1, \ldots, s_r \rangle$.

    \hypertarget{normalizer_formalism_case_2}{\textbf{Case 2:}}
    $p$ commutes with all generators $s_j$ and $\pm p \in \SSS$.
    Equivalently, $\ol{p}$ is either $\ol{\one}$ or $-\ol{\one}$ in $N(\SSS) / \SSS$.
    This case is trivial: $\SSS' = \SSS$, hence $N(\SSS') / \SSS' = N(\SSS) / \SSS$.

    \hypertarget{normalizer_formalism_case_3}{\textbf{Case 3:}}
    $p$ anti-commutes with at least one $s_j$.
    Equivalently, $\ol{p} \notin N(\SSS) / \SSS$.
    Without loss of generality, we assume $p$ anti-commutes with $s_1$ and commutes with $s_2, \ldots, s_r$.
    Then by multiplying representatives of generators of $N(\SSS) / \SSS$ by $s_1$ if needed,
    we can also assume $x_1, z_1, \ldots, x_k, z_k$ all commute with $p$.
    We then find that:
    \begin{equation}\label{eq:normalizer_formalism_case_3}
        N(\SSS') / \SSS' = \langle ~
            \ol{i},~
            \ol{x_1},~
            \ol{z_1},~
            \ol{x_2},~
            \ol{z_2},~
            \ldots,~
            \ol{x_k},~
            \ol{z_k} ~ \rangle
        \cong \PPP_{k}
    \end{equation}
    where $\SSS'$ is now $\langle mp, s_2, \ldots, s_r \rangle$.
\end{theorem}

\begin{proof}
    Let's first note that since $p$ is Hermitian, we have $p^2 = \one$.
    Thus if $\ol{p}$ is in $N(\SSS)/\SSS$, it must be of order 2.
    This is important because it means
    that the three cases above are disjoint and cover all possibilities;
    $\ol{p}$ is either in $N(\SSS)/\SSS$ or it isn't,
    and if it is, it's either $\ol{\one}$, $-\ol{\one}$ or another order-2 element.
    In the following, we will occasionally need to be verbose
    as to whether $\ol{q}$ denotes $q\SSS \in N(\SSS)/\SSS$ or $q\SSS' \in N(\SSS')/\SSS'$.

    We start with \normCase{1}.
    If $p$ commutes with all generators of $\SSS$ but $\pm p \notin \SSS$,
    then by definition this means $\ol{p}$ is an order-2 element of $N(\SSS)/\SSS$ other than $\ol{\one}$ or $-\ol{\one}$.
    The converse also holds, hence our use of the word `equivalently' above was justified.
    By \Cref{lem:normalizer_generator_choice}, we can always choose a presentation
    $\langle ~
        \ol{i},~
        \ol{x_1},~
        \ol{z_1},~
        \ldots,~
        \ol{x_k},~
        \ol{z_k} ~
    \rangle$
    for $N(\SSS)/\SSS$ such that $\ol{x_1} = \ol{p}$
    and the Pauli commutativity relations are satisfied by all $\ol{x_j}$ and $\ol{z_j}$.
    Since $\SSS'$ has minimal presentation $\langle mp, s_1, \ldots, s_r \rangle$,
    \Cref{lem:normalizer_iso_pauli_group} tells us that $N(\SSS') / \SSS'$ is isomorphic to $\PPP_{k-1}$.

    We then claim that a presentation for $N(\SSS')/\SSS'$ is
    $\langle
        \ol{i},~
        \ol{x_2},~
        \ol{z_2},~
        \ldots,~
        \ol{x_k},~
        \ol{z_k}
    \rangle$.
    Let's first check this is well-defined, in that the representative $q$ of each generator $\ol{q}$ is in $N(\SSS')$.
    This follows from the fact that each $q$ was in $N(\SSS)$, so commutes with $s_1, \ldots, s_r$,
    and because each $\ol{q}$ commuted with $\ol{x_1} = \ol{p}$ in $N(\SSS)/\SSS$, hence $q$ commutes with $mp$.
    Next, we check these generators remain independent in $N(\SSS')/\SSS'$.
    Similarly, this follows from the fact that they were
    independent of one another and of $\ol{p} = \ol{x_1}$ in $N(\SSS) / \SSS$.
    Specifically, let's assume for a contradiction that a generator $\ol{q}$ can be written
    in terms of the remaining $2(k-1)$ generators $\ol{q_1},~ \ldots,~ \ol{q_{2(k-1)}}$
    as a word $\ol{w} \coloneqq \ol{q_1^{a_1} \ldots q_{2(k-1)}^{a_{2(k-1)}}}$.
    Then this would imply $q = ws'$ for some $s' \in \SSS'$.
    This $s'$ can in turn be written as $(mp)^{b_0} s_1^{b_1} \ldots s_r^{b_r}$.
    But noting that $m = i^{1-m}$,
    we can define $w' = w (i^{1-m}p)^{b_0}$ as a new word over representatives of generators of $N(\SSS)/\SSS$,
    and can also define $s = s_1^{b_1} \ldots s_r^{b_r} \in \SSS$, so that $q = w's$.
    But then this says $\ol{q} = \ol{w'}$ in $N(\SSS)/\SSS$, contradicting the fact that
    the generators ${\ol{p}, \ol{q}, \ol{q_1}, \ldots, \ol{q_{2(n-k-1)}}}$ were independent in $N(\SSS)/\SSS$.

    Finally, we check that they obey the Pauli commutativity relations in $N(\SSS')/\SSS'$;
    this follows directly from the fact that they did so in $N(\SSS) / \SSS$.
    Thus this presentation is a subgroup of $N(\SSS') / \SSS'$ isomorphic to $\PPP_{k-1}$.
    But since $N(\SSS') / \SSS'$ is itself isomorphic to $\PPP_{k-1}$,
    this presentation must be exactly $N(\SSS') / \SSS'$, as claimed.

    Next up is \normCase{2}, where $p$ commutes with all generators $s_j$ and $\pm p \in \SSS$.
    By definition, if $\pm p \in \SSS$ then $\ol{p} = \pm \ol{\one}$ in $N(\SSS)/\SSS$, and vice versa,
    so again our use of `equivalently' in the theorem statement was justified.
    Then there's nothing left to prove; since $\SSS' = \SSS$ in this case,
    we know $N(\SSS')/\SSS' = N(\SSS)/\SSS$.

    Finally, in \normCase{3}, we know $p$ anti-commutes with at least one generator of $\SSS$.
    Thus it cannot be in $N(\SSS)$, and hence $\ol{p} \notin N(\SSS)/\SSS$.
    Again, the converse also holds, justifying our use of `equivalently' above.
    We may assume without loss of generality that $\SSS = \langle s_1, \ldots, s_r \rangle$
    where $p$ anti-commutes with $s_1$ but commutes with the remaining generators.
    We have a presentation
    $\langle ~
        \ol{i},~
        \ol{x_1},~
        \ol{z_1},~
        \ldots,~
        \ol{x_{k}},~
        \ol{z_{k}} ~
    \rangle$
    for $N(\SSS)/\SSS$.
    If any generator $q\SSS = \ol{q} \in \{\ol{x_1},~ \ol{z_1},~ \ldots,~ \ol{x_{k}},~ \ol{z_{k}}\}$
    is such that its representative $q$ anti-commutes with $p$,
    we can pick a new representative $qs_1 \in q\SSS$ that must commute with $p$.
    Hence without loss of generality we may assume all representatives $x_1,~ z_1,~ \ldots,~ x_{k},~ z_{k}$
    commute with $p$.
    Since $\SSS'$ has minimal presentation $\langle mp, s_2, \ldots, s_r \rangle$,
    we know $N(\SSS')/\SSS'$ is isomorphic to $\PPP_{k}$.

    We claim that a presentation for $N(\SSS')/\SSS'$ is
    $\langle
        \ol{i},~
        \ol{x_1},~
        \ol{z_1},~
        \ldots,~
        \ol{x_{k}},~
        \ol{z_{k}}
    \rangle$.
    To prove this, we repeat similar steps as for \normCase{1} above;
    we first show this is well-defined, in that each generator representative is in $N(\SSS')$.
    This follows from the fact that each one was in $N(\SSS)$, so commutes with $s_2, \ldots, s_r$,
    and from the assumption above that we picked them to commute with $p$, and hence $mp$ too.
    Then we must show these generators remain independent.
    To see this, suppose for a contradiction that some generator $\ol{q}$
    can be written in terms of the other $2k$ generators $\ol{q_1}, \ldots, \ol{q_{2k}}$,
    as some word $\ol{w} \coloneqq \ol{q_1^{a_1} \ldots q_{2k}^{a_{2k}}}$, for $a_j \in \{0, 1\}$.
    This would mean $q = ws'$, for some $s' \in \SSS'$.
    This $s'$ can itself be written as $(mp)^{b_1} s_2^{b_2}, \ldots, s_r^{b_r}$.
    If $b_1 = 0$, then this contradicts the fact that the generators ${\ol{q}, \ol{q_1}, \ldots, \ol{q_{2k}}}$
    were independent in $N(\SSS)/\SSS$.
    But equally if $b_1 = 1$, then $s_1$ anticommutes with $ws'$ but not $q$,
    which is again a contradiction, since $q = ws'$.
    So the generators must all be independent.
    Finally, we must prove the generators still satisfy the Pauli commutation relations;
    this follows directly from the fact that they did so in $N(\SSS)/\SSS$.
    We can conclude this presentation is a subgroup of $N(\SSS')/\SSS'$ isomorphic to $\PPP_{n-k}$,
    and since $N(\SSS')/\SSS'$ is itself isomorphic to $N(\SSS')/\SSS'$,
    it must be that this presentation is the whole of $N(\SSS')/\SSS'$, as claimed.
\end{proof}
    \section{ISG code example}\label{sec:isg_code_example}

Here we give a simple example of an ISG code
and the evolution over time of its ISG $\isg{t}$ and logical Pauli group $\lpg{t}$.
Once we know the ISG's evolution, we can then immediately write down detectors for the code too.
We do all this algebraically, using the stabilizer and normalizer formalisms of \Cref{sec:normalizer_formalism}.
We will consider the ISG code defined by the schedule
$\MMM = [\langle X_1 X_2,~ X_3 X_4 \rangle, \langle  Z_1 Z_3,~ Z_2 Z_4 \rangle]$.
Those who've seen a bit of error correction before may recognise this as the distance-two Bacon-Shor subsystem code,
which can in turn be thought of as a subsystem implementation of the distance-two rotated surface code.

\subsection{Instantaneous stabilizer group}\label{subsec:isg_example_isg}

The initial ISG $\isg{-1}$ is the trivial group $\{\one\}$, as ever.
Using the stabilizer formalism we can see how $\isg{t}$ evolves.
In the zero-th round we measure $X_1 X_2$ and $X_3 X_4$, so we are in \stabCase{1} for both;
we get random outcomes $m_{X_1 X_2}^{(0)}$ and $m_{X_3 X_4}^{(0)}$ respectively,
and the ISG immediately afterwards is:
\begin{equation}\label{eq:eg:ISG_0}
\isg{0} = \langle m_{X_1 X_2}^{(0)} X_1 X_2,~ m_{X_3 X_4}^{(0)} X_3 X_4 \rangle
\end{equation}
By multiplying the second generator by the first,
we get a new presentation for the same group that'll be more convenient for figuring out what $\isg{1}$ is in a second:
\begin{equation}\label{eq:eg:ISG_0'}
\isg{0} = \langle m_{X_1 X_2}^{(0)} X_1 X_2,~ m_{X_1 X_2}^{(0)} m_{X_3 X_4}^{(0)} X_1 X_2 X_3 X_4 \rangle
\end{equation}
In the next round we measure $Z_1 Z_3$ and $Z_2 Z_4$.
First, consider $Z_1 Z_3$.
Since it anticommutes with $X_1 X_2$ but commutes with $X_1 X_2 X_3 X_4$, we're in \stabCase{3};
$m_{X_1 X_2}^{(0)} X_1 X_2$ is removed from the generating set of $\isg{1}$,
and replaced by $m_{Z_1 Z_3}^{(1)}Z_1 Z_3$, where $m_{Z_1 Z_3}^{(1)}$ is the random outcome of measuring $Z_1 Z_3$.
Next, consider $Z_2 Z_4$:
since it commutes with $Z_1 Z_3$ and $X_1 X_2 X_3 X_4$ but neither it nor its negation is in $\isg{0}$,
we're in \stabCase{1}.
So $m_{Z_2 Z_4}^{(1)}Z_2 Z_4$ is added as a generator of $\isg{1}$,
where $m_{Z_2 Z_4}^{(1)}$ is the random outcome of measuring $Z_2 Z_4$.
So altogether:
\begin{equation}\label{eq:eg:ISG_1}
\isg{1} = \langle
m_{Z_1 Z_3}^{(1)} Z_1 Z_3,~
m_{Z_2 Z_4}^{(1)} Z_2 Z_4,~
m_{X_1 X_2}^{(0)} m_{X_3 X_4}^{(0)} X_1 X_2 X_3 X_4 \rangle
\end{equation}
For every subsequent timestep $t > 1$, the ISG $\SSS_t$ will have rank 3.
That is, this code is established after $T=1$ timesteps.
Let's spell this out for $t=2$, wherein we measure $X_1 X_2$ and $X_3 X_4$ again.
We'll first use a different generating set for $\isg{1}$, like we did just a second ago:
\begin{equation}\label{eq:eg:ISG_1'}
\isg{1} = \langle
m_{Z_1 Z_3}^{(1)} Z_1 Z_3,~
m_{Z_1 Z_3}^{(1)} m_{Z_2 Z_4}^{(1)} Z_1 Z_2 Z_3 Z_4,~
m_{X_1 X_2}^{(0)} m_{X_3 X_4}^{(0)} X_1 X_2 X_3 X_4 \rangle
\end{equation}
We measure $X_1 X_2$, which anticommutes with $Z_1 Z_3$, commutes with the other generators,
and neither it nor its negation is in $\isg{1}$.
So we're in \stabCase{3};
we get random outcome $m_{X_1 X_2}^{(2)}$,
and $m_{X_1 X_2}^{(2)} X_1 X_2$ replaces $m_{Z_1 Z_3}^{(1)} Z_1 Z_3$ as a generator.
In particular, the product
$m_{X_1 X_2}^{(2)} m_{X_1 X_2}^{(0)} m_{X_3 X_4}^{(0)} X_3 X_4$ of
$m_{X_1 X_2}^{(2)} X_1 X_2$ and
$m_{X_1 X_2}^{(0)} m_{X_3 X_4}^{(0)} X_1 X_2 X_3 X_4$
is now in $\isg{1}$.
So when we then measure $X_3 X_4$, we're in \stabCase{2};
we get deterministic outcome $m_{X_3 X_4}^{(2)} = m_{X_1 X_2}^{(2)} m_{X_1 X_2}^{(0)} m_{X_3 X_4}^{(0)}$,
and $\isg{2}$ remains unchanged:
\begin{equation}\label{eq:eg:ISG_2}
\isg{2} = \langle
m_{X_1 X_2}^{(2)} X_1 X_2,~
m_{Z_1 Z_3}^{(1)} m_{Z_2 Z_4}^{(1)} Z_1 Z_2 Z_3 Z_4,~
m_{X_1 X_2}^{(0)} m_{X_3 X_4}^{(0)} X_1 X_2 X_3 X_4 \rangle
\end{equation}
By multiplying the third generator by the first and using
$m_{X_3 X_4}^{(2)} = m_{X_1 X_2}^{(2)} m_{X_1 X_2}^{(0)} m_{X_3 X_4}^{(0)}$,
we get the following more convenient presentation:
\begin{equation}\label{eq:eg:ISG_2'}
\isg{2} = \langle
m_{X_1 X_2}^{(2)} X_1 X_2,~
m_{X_3 X_4}^{(2)} X_3 X_4,~
m_{Z_1 Z_3}^{(1)} m_{Z_2 Z_4}^{(1)} Z_1 Z_2 Z_3 Z_4 \rangle
\end{equation}
And indeed this has rank 3 again, as promised.
If we continued this sort of analysis, we'd see that we get:
\begin{equation}\label{eq:eg:isg_code_isg_full}
    \isg{t} = \begin{cases}
        \{\one\} &\text{if } t < 0 \\[3pt]
        \langle
        m_{X_1 X_2}^{(0)} X_1 X_2,~
        m_{X_3 X_4}^{(0)} X_3 X_4
        \rangle &\text{if } t = 0 \\[3pt]
        \langle
        m_{Z_1 Z_3}^{(t)} Z_1 Z_3,~
        m_{Z_3 Z_4}^{(t)} Z_2 Z_4,~
        m_{X_1 X_2}^{(0)} m_{X_3 X_4}^{(0)} X_1 X_2 X_3 X_4
        \rangle &\text{if $t > 0$ odd} \\[3pt]
        \langle
        m_{X_1 X_2}^{(t)} X_1 X_2,~
        m_{X_3 X_4}^{(t)} X_3 X_4,~
        m_{Z_1 Z_3}^{(1)} m_{Z_3 Z_4}^{(1)} Z_1 Z_2 Z_3 Z_4
        \rangle &\text{if $t > 0$ even} \\
    \end{cases}
\end{equation}

Note that the ISG here is \textit{dynamic} even after establishment at $T=1$;
it changes from one timestep to the next.
As pointed out in the table in \Cref{fig:isg_relationships},
this is characteristic of a subsystem code.
In fact, we can be even more specific;
for any subsystem code,
the way in which the ISG changes between timesteps (after establishment)
is more commonly called \textit{gauge fixing} -
we give more detail on this in \Cref{subsec:gauge_fixing}.

\subsection{Detectors}\label{subsec:isg_example_detectors}

By analysing the evolution of the ISG of the code,
we also uncover the code's detectors.
Specifically, anytime we measure a Pauli $p$ and find ourselves in $\stabCase{2}$,
we learn a detector of the code.
Moreover, a generating set of detectors of the code can be found this way.
To see why this is, recall that a detector is defined to be
a set of measurement outcomes whose product is deterministic in the absence of noise.
In any ISG code, at any timestep $t$, every element of the ISG $\isg{t}$ is of the form
$m_{p_1}^{(t_1)} \ldots m_{p_\ell}^{(t_\ell)} p_1 \ldots p_\ell$,
where each $p_j$ is a Pauli and $m_{p_j}^{(t_j)}$ is the outcome of measuring $p_j$ at timestep $t_j$.
If at time $t$ we measure some Pauli $p$ and land in $\stabCase{2}$,
it means that $\pm p$ was already in $\isg{t}$.
In other words, $\pm p = m_{p_1}^{(t_1)} \ldots m_{p_\ell}^{(t_\ell)} p_1 \ldots p_\ell$.
So the measurement outcome $m_p^{(t)}$ is deterministically equal to $m_{p_1}^{(t_1)} \ldots m_{p_\ell}^{(t_\ell)}$;
that is, the product $m_p^{(t)} m_{p_1}^{(t_1)} \ldots m_{p_\ell}^{(t_\ell)}$ is deterministic (it's 1, in this case).
Hence the formal product $m_p^{(t)} m_{p_1}^{(t_1)} \ldots m_{p_\ell}^{(t_\ell)}$ is a detector.

For example, in the analysis above, the measurement of $X_3 X_4$ at time $t=2$ was handled by $\stabCase{2}$.
We found that the measurement outcome $m_{X_3 X_4}^{(2)}$ was deterministically equal to
the product $m_{X_1 X_2}^{(2)} m_{X_1 X_2}^{(0)} m_{X_3 X_4}^{(0)}$ of previous measurements.
Thus the formal product $m_{X_1 X_2}^{(2)} m_{X_3 X_4}^{(2)} m_{X_1 X_2}^{(0)} m_{X_3 X_4}^{(0)}$ is a detector.
Indeed, were we to continue this analysis,
we would find detectors:
\begin{equation}\label{eq:isg_example_detectors}
    \begin{aligned}
        m_{X_1 X_2}^{(t)} m_{X_3 X_4}^{(t)} m_{X_1 X_2}^{(t-2)} m_{X_3 X_4}^{(t-2)}
        & \text{\quad at every even timestep } t \geq 2 \\
        m_{Z_1 Z_3}^{(t)} m_{Z_2 Z_4}^{(t)} m_{Z_1 Z_3}^{(t-2)} m_{Z_2 Z_4}^{(t-2)}
        & \text{\quad at every odd timestep } t \geq 2 \\
    \end{aligned}
\end{equation}

The reason we said this gives us a \textit{generating set} of detectors is because detectors form a group.
Suppose we have two detectors $m_1 \ldots m_u$ and $m_1' \ldots m_v'$.
Since by definition these products are deterministic,
the total product $m_1 \ldots m_u m_1' \ldots m_v'$ must be deterministic too.
Hence the formal product $m_1 \ldots m_u m_1' \ldots m_v'$ is a detector.
The group's unit is the trivial formal product,
which corresponds to an empty set of measurement outcomes.
We can write this as $\prod_{m_j \in \emptyset} m_j = 1$.
Every detector is then its own inverse; as a formal product whose powers are taken mod $2$, we have
$m_1 \ldots m_u m_1 \ldots m_u = m_1^2 \ldots m_u^2 = m_1^0 \ldots m_u^0 = 1$.

\subsection{Logical Pauli group}\label{subsec:isg_example_lpg}

Tracking the evolution of the logical Pauli group $\lpg{t}$ has a very similar feel as for $\isg{t}$,
but there's a little bit more to picking the right presentation for the group now,
in order to be able to apply the normalizer formalism.

As a reminder, we're considering the schedule
$\MMM = [\langle X_1 X_2,~ X_3 X_4 \rangle, \langle  Z_1 Z_3,~ Z_2 Z_4 \rangle]$.
Since the initial ISG $\isg{-1}$ is trivial,
the most obvious presentation for the initial logical Pauli group is:
\begin{equation}\label{eq:isg_code_ops_initial}
\lpg{-1} = \langle~ \ol{i},~ \ol{X_1},~ \ol{Z_1},~ \ol{X_2},~ \ol{Z_2},~ \ol{X_3},~ \ol{Z_3},~ \ol{X_4},~ \ol{Z_4} ~\rangle
\end{equation}
But at $t=0$ we measure $X_1 X_2$ and $X_3 X_4$,
so in order to use the normalizer formalism, we need a different presentation.
First, we consider measuring $X_1 X_2$.
We know we'll be in \normCase{1};
we can see this either by observing that $\ol{X_1 X_2}$ is an element of $\lpg{-1}$ other than $\ol{\one}$ or $-\ol{\one}$,
or equivalently by observing that $X_1 X_2$ is not in $\isg{-1}$ but commutes with all its generators
(vacuously, because there are none).
To apply the formalism, we can first multiply $\ol{X_1}$ by $\ol{X_2}$,
so that $\ol{X_1 X_2}$ is a generator.
One might think that we're then good to go, but we must be careful:
this multiplication means the generators no longer satisfy the Pauli commutativity relations.
To account for this, we can multiply $\ol{Z_2}$ by $\ol{Z_1}$, to get:
\begin{equation}\label{eq:isg_code_ops_initial'}
    \lpg{-1} = \langle
        \ol{i},~
        \ol{X_1 X_2},~
        \ol{Z_1},~
        \ol{X_2},~
        \ol{Z_1 Z_2},~
        \ol{X_3},~
        \ol{Z_3},~
        \ol{X_4},~
        \ol{Z_4}
    \rangle
\end{equation}
Now we can apply the formalism!
Doing this simply removes $\ol{X_1 X_2}$ and $\ol{Z_1}$ from the generating set.
We can repeat something similar for the measurement of $X_3 X_4$;
namely, we can multiply $\ol{X_3}$ by $\ol{X_4}$ and $\ol{Z_4}$ by $\ol{Z_3}$,
then apply \normCase{1} to get:
\begin{equation}\label{eq:isg_code_ops_0}
    \lpg{0} = \langle~ \ol{i},~ \ol{X_2},~ \ol{Z_1 Z_2},~ \ol{X_4},~ \ol{Z_3 Z_4} ~\rangle
\end{equation}
where $\isg{0} = \langle m_{X_1 X_2}^{(0)} X_1 X_2,~ m_{X_3 X_4}^{(0)} X_3 X_4 \rangle$.
At the next timestep $t=1$, we first measure $Z_1 Z_3$.
This puts us in \normCase{3}, and transforms $\isg{0}$ to
$\isg{0}' = \langle m_{Z_1 Z_3}^{(1)} Z_1 Z_3,~ m_{X_1 X_2}^{(0)} m_{X_3 X_4}^{(0)} X_1 X_2 X_3 X_4 \rangle$.
Since all representatives of the generators of $\lpg{0}$ already commute with $Z_1 Z_3$,
we don't need to change our presentation to be able to apply the formalism.
We get:
\begin{equation}\label{eq:isg_code_ops_0'}
N(\isg{0}') / \isg{0}' = \langle~ \ol{i},~ \ol{X_2},~ \ol{Z_1 Z_2},~ \ol{X_4},~ \ol{Z_3 Z_4} ~\rangle
\end{equation}
The other measurement we make at timestep $t=1$ is $Z_2 Z_4$.
This means we're in \normCase{1}, but we need an alternative presentation first,
in which $\ol{Z_2 Z_4}$ is a generator.
We can do this by multiplying $\ol{Z_3 Z_4}$ by $\ol{Z_1 Z_2}$ to get $\ol{Z_1 Z_2 Z_3 Z_4}$,
then, since $m_{Z_1 Z_3}^{(1)} Z_1 Z_3 \in \isg{0}'$, we can pick a new representative $m_{Z_1 Z_3}^{(1)} Z_2 Z_4$.
If necessary, we can multiply this by $\ol{i}^2$ to get just $\ol{Z_2 Z_4}$, for convenience.
But now our generators don't satisfy the Pauli commutativity relations again.
Fortunately we can fix this by multiplying $\ol{X_2}$ by $\ol{X_4}$,
picking a new representative $m_{X_1 X_2}^{(0)} m_{X_3 X_4}^{(0)} X_1 X_3$,
and for convenience multiplying by $\ol{i}^2$ if needed to get just $\ol{X_1 X_3}$.
Altogether, this gives:
\begin{equation}\label{eq:isg_code_ops_0''}
    N(\isg{0}') / \isg{0}' = \langle~ \ol{i},~ \ol{X_1 X_3},~ \ol{Z_1 Z_2},~ \ol{X_4},~ \ol{Z_2 Z_4} ~\rangle
\end{equation}
We can now directly apply \normCase{1}; the resulting new logical Pauli group is:
\begin{equation}\label{eq:isg_code_ops_1}
\lpg{1} = \langle~ \ol{i},~ \ol{X_1 X_3},~ \ol{Z_1 Z_2} ~\rangle
\end{equation}
And this is how things stay: for every measurement $p$ performed at any future timestep $t > 1$,
we're either in \textit{Case 2} or \textit{Case 3} - and when it's the latter,
the given representatives of generators of $\lpg{1}$ even commute with $p$ already.
Hence:
\begin{equation}\label{eq:eg:isg_code_ops_full}
\lpg{t} = \begin{cases}
        \langle~
            \ol{i},~
            \ol{X_1},~ \ol{Z_1},~
            \ol{X_2},~ \ol{Z_2},~
            \ol{X_3},~ \ol{Z_3},~
            \ol{X_4},~ \ol{Z_4} ~
        \rangle \cong \PPP_4
        &\text{if } t < 0 \\
        \langle~
            \ol{i},~
            \ol{X_2},~ \ol{Z_1 Z_2},~
            \ol{X_4},~ \ol{Z_3 Z_4} ~
        \rangle \cong \PPP_2
        &\text{if } t = 0 \\
        \langle~
            \ol{i},~
            \ol{X_1 X_3},~ \ol{Z_1 Z_2} ~
        \rangle \cong \PPP_1
        &\text{if $t > 0$}
    \end{cases}
\end{equation}

We can thus see that the logical Pauli group is \textit{static} after becoming established at time $T=1$.
Again, as noted in the table in \Cref{fig:isg_relationships}, this is characteristic of a subsystem code.
    \section{Remarks on ISG codes}\label{sec:isg_code_remarks}

Here we comment on some aspects of ISG codes that we didn't have space for in the main text.
First, we note that $|\lpg{t}|$ is a non-increasing sequence in $t$.
That is, the number of qubits encoded in an ISG code can only ever stay the same or decrease, but never increase.
One can see this from the normalizer formalism of the previous section;
$\lpg{t}$ is initially isomorphic to $\PPP_n$, so encodes $n$ qubits,
then evolves exclusively by measuring Paulis $p \in \PPP_n$.
Such a measurement can only ever cause the number of encoded qubits to
decrease by one (\hyperlink{normalizer_formalism_case_1}{\textit{Case 1}})
or stay the same (\hyperlink{normalizer_formalism_case_2}{\textit{Cases 2 and 3}}).

Second, we note that, for all $t$, $\isg{t}$ has rank $r \iff \lpg{t} \cong \PPP_{n-r}$.
Thus when we defined `establishment' as being
the existence of a $T$ such that $\isg{t}$ has fixed rank $r$ for all $t \geq T$,
we could equally have defined it in terms of logical qubits,
as a $T$ such that $\lpg{t} \cong \PPP_k$ for fixed $k = n - r$ for all $t \geq T$.

We confess that we have used a very naive definition of distance;
namely, the minimum weight of any element of $\lpg{t}$ over all $t \geq T$.
This only considers problematic Pauli operators within a single timestep,
and fails to take into account that there can be combinations of Pauli operators over multiple timesteps
which anti-commute with logical operators but do not violate any detectors.
Our definition should perhaps be downgraded to a \textit{spacelike distance},
and a better definition of distance could instead be directly related to detecting and operating regions.

An interesting unanswered question concerns when two ISG codes should be considered equivalent.
For example, suppose we have an ISG code
$\MMM = [\ldots, \MMM_t, \MMM_{t+1}, \ldots]$.
If $\MMM_t \MMM_{t+1}$ is Abelian
(i.e.\ all Paulis in $\MMM_t$ and $\MMM_{t+1}$ commute with each other),
should this be considered equivalent to the ISG code
$\MMM' = [\ldots, \MMM_t \MMM_{t+1}, \ldots]$?
Certainly the ISGs $\isg{t+1}$ of $\MMM$ and $\isg{t}'$ of $\MMM'$ will be identical,
and thus so too the logical Pauli groups $\lpg{t+1}$ and $N(\isg{t}') / \isg{t}'$.
If we think the answer to this question is yes,
then any stabilizer code $\MMM = [\MMM_0, \MMM_1, \ldots, \MMM_t, \ldots]$
is equivalent to the period-1 code $\MMM = [\MMM_0 \MMM_1 \ldots \MMM_t \ldots]$.
In particular, this would make every stabilizer code a Floquet code,
demanding we update our Venn diagram from \Cref{fig:isg_relationships} to:
\begin{equation}\label{eq:venn_diagram_alternative}
    \scalebox{0.9}{\begin{tikzpicture}
	\begin{pgfonlayer}{nodelayer}
		\node [style=none] (2) at (8.75, 0) {};
		\node [style=none] (3) at (-6.75, 0) {};
		\node [style=none] (4) at (4.25, 0) {};
		\node [style=none] (5) at (-4.75, 0) {};
		\node [style=none] (6) at (4, 0) {};
		\node [style=none] (7) at (-0.5, 0) {};
		\node [style=none] (8) at (8.25, 0) {};
		\node [style=none] (9) at (-0.75, 0) {};
		\node [style=none] (10) at (-5.75, 0) {ISG};
		\node [style=none] (12) at (1.75, 0) {Stabilizer};
		\node [style=none] (13) at (6.25, -0.05) {\textcolor{gray}{Floquet}};
		\node [style=none] (14) at (-2.75, -0.05) {Subsystem};
	\end{pgfonlayer}
	\begin{pgfonlayer}{edgelayer}
		\draw [bend left=90, looseness=0.75] (2.center) to (3.center);
		\draw [bend left=90, looseness=0.75] (3.center) to (2.center);
		\draw [bend left=90] (4.center) to (5.center);
		\draw [bend left=90] (5.center) to (4.center);
		\draw [bend left=90] (6.center) to (7.center);
		\draw [bend left=90] (7.center) to (6.center);
		\draw [color=gray, bend left=90] (8.center) to (9.center);
		\draw [color=gray, bend left=90] (9.center) to (8.center);
	\end{pgfonlayer}
\end{tikzpicture}
}
\end{equation}
Other cases also arise -
should the periodic ISG code $\MMM = [\MMM_0, \MMM_1, \ldots, \MMM_{\ell-2}, \MMM_{\ell-1}]$
be considered equivalent to one like $\MMM' = [\MMM_1, \MMM_2, \ldots, \MMM_{\ell-1}, \MMM_0]$,
which has the same measurement sequence only cyclically permuted?
We would be interested to hear more opinions on this matter.

One could also point out that `establishment' isn't a strictly necessary part of the definition of an ISG code.
For example, one could imagine a measurement schedule $\MMM = [\MMM_0, \MMM_1, \ldots]$ such that
the sequence $|\lpg{t}|$ never becomes constant,
but decreases so slowly that one can still use a subset of the encoded logical qubits to perform quantum computation.
We would be very interested to hear of any non-trivial such $\MMM$.
Finally, we remark that the definition could (and probably should) be extended such that,
in addition to Pauli measurements, 
Clifford operations are also allowed to be applied to the qubits of the code. 
    \section{Remarks on subsystem codes}\label{sec:subsystem_code_remarks}

Here, we give more detail on the relationship between subsystem codes and ISG codes.
In \Cref{defn:subsystem_code_as_ISG_code}, we defined a subsystem code as a type of ISG code,
and confessed that this definition actually deviates slightly from the usual one.
We now discuss this deviation at length.
We first introduce new notation -
for any group $\GGG \leq \PPP_n$,
we use $\modPhases{\GGG}$ to denote $\GGG / \langle i \rangle$, the same group modulo phases -
and we remind ourselves of previous notation;
for groups $\GGG_1, \ldots, \GGG_{\ell} \leq \PPP_n$,
we write the product $\GGG_1 \ldots \GGG_{\ell}$ to mean
the group generated by the union of generating sets for $\GGG_1, \ldots, \GGG_{\ell}$,
and we let $\PPP_k^\circ$ denote the `almost Pauli group' $\langle X_1, Z_1, \ldots, X_k, Z_k \rangle$,
which satisfies $\addPhases{\PPP_k^\circ} = \PPP_k$.

At a high level, subsystem codes are
stabilizer codes in which some logical qubits aren't used to store logical information.
So the formalism doesn't give rise to new codes,
but rather gives rise to more flexible ways of correcting errors on these codes~\cite{BaconShorCodeOriginalPaper}.
These unused logical qubits are referred to as \textit{gauge qubits}.
Much like stabilizer codes can be defined entirely by a stabilizer group,
subsystem codes can be defined entirely by a subgroup $\GGG \leq \PPP_n$ of the Pauli group.
The only requirement on this $\GGG$ is that it must contain $i$ (and hence need not be Abelian).
We must also pick a presentation
$\langle i, s_1, \ldots, s_R \rangle$ of $Z(\GGG)$
such that $\SSS = \langle s_1, \ldots, s_R \rangle$ is a stabilizer group,
but the actual choice is arbitrary.
Crucially, at no point above was a measurement schedule required,
which makes this definition more general than ours.
Throughout the rest of this section, we will forget our \Cref{defn:subsystem_code_as_ISG_code};
whenever we mention a subsystem code,
we now just mean a gauge group $\GGG$.

It was originally envisioned that, in order to use a subsystem code for error correction,
one would repeatedly measure a generating set for $\SSS$~\cite{PoulinStabilizerFormalismSubsystemCodes}.
Later, it was pointed out that actually one could potentially measure
a generating set for a group $\GGG'$ satisfying $\addPhases{\GGG'} = \GGG$,
and combine the outcomes in order to infer the measurement outcomes of
a generating set for $\SSS$~\cite{BombinRubyCodes}.
Since the generators of such a $\GGG'$ are typically of lower-weight than those of $\SSS$,
this provides a significant practical advantage,
and is thus generally how subsystem codes are now envisioned to be implemented.
But since $\GGG$ can be non-Abelian,
we can't necessarily measure a full generating set for $\GGG'$ at once.
Instead, if choosing this route, we must specify a schedule
in which to measure mutually commuting sets of generators.
This leads us to define an \textit{associated ISG code} for a subsystem code:

\begin{definition}\label{defn:associated_ISG_code}
    Given a subsystem code $\GGG$, an \textbf{associated ISG code} for it is
    an ISG code $\MMM_\GGG = [\MMM_0, \MMM_1, \ldots]$ that establishes at some time $T \geq 0$,
    such that $\addPhases{\MMM_0 \MMM_1 \ldots} = \GGG$,
    and the ISG $\SSS_t$ satisfies $\modPhases{\SSS} \leq \modPhases{\SSS_t}$ for all $t \geq T$.
\end{definition}

In the other direction, we can define the \textit{associated subsystem code} for an ISG code:

\begin{definition}\label{defn:}
    Given an ISG code $\MMM = [\MMM_0, \MMM_1, \ldots]$,
    we say its \textbf{associated subsystem code} is defined by the gauge group
    $\GGG_\MMM = \addPhases{\MMM_0 \MMM_1 \ldots}$.
\end{definition}

Note that a single subsystem code $\GGG$ can have many associated ISG codes $\MMM_\GGG$,
but an ISG code $\MMM$ has a unique associated subsystem code $\GGG_\MMM$.
As an example, in \Cref{sec:isg_code_example} we said we looked at the distance-two Bacon-Shor code,
which is normally defined via the gauge group
$\GGG = \langle i, X_1 X_2, X_3 X_4, Z_1 Z_3, Z_2 Z_4 \rangle$.
However, since this group is non-Abelian,
if we choose to implement it by measuring generators of a $\GGG'$ satisfying $\addPhases{\GGG'} = \GGG$
(the obvious choice being $\GGG' = \langle X_1 X_2, X_3 X_4, Z_1 Z_3, Z_2 Z_4 \rangle$),
then we require a measurement schedule.
In that section, we thus actually looked at the associated ISG code
$\MMM_\GGG = [\langle X_1 X_2, X_3 X_4 \rangle, \langle Z_1 Z_3, Z_2 Z_4 \rangle]$.
Let's also point out quickly that,
while any partition of generators of $\GGG$ into mutually commuting subsets produces an ISG code,
some fall foul of the definition of an \textit{associated} ISG code for $\GGG$.
The schedule
$\MMM_\GGG = [\langle X_1 X_2 \rangle, \langle Z_2 Z_4 \rangle, \langle X_3 X_4 \rangle, \langle Z_1 Z_3 \rangle]$
is one such example for the Bacon-Shor code above.
Here the ISG at any time $t$ is (up to phases) exactly the group $\MMM_t$ whose generating set was just measured.
That is, letting $m_p^{(t)}$ denote the outcome of measuring Pauli $p$ at time $t \geq 0$, we have:
\begin{equation}\label{eq:subsystem_code_bad_isg}
    \isg{t} = \begin{cases}
        \langle m_{X_1 X_2}^{(t)} X_1 X_2 \rangle &\text{if } t = 0 \text{ mod } 4 \\[3pt]
        \langle m_{Z_2 Z_4}^{(t)} Z_2 Z_4 \rangle &\text{if } t = 1 \text{ mod } 4 \\[3pt]
        \langle m_{X_3 X_4}^{(t)} X_3 X_4 \rangle &\text{if } t = 2 \text{ mod } 4 \\[3pt]
        \langle m_{Z_1 Z_3}^{(t)} Z_1 Z_3 \rangle &\text{if } t = 3 \text{ mod } 4 \\
    \end{cases}
\end{equation}
Since $\modPhases{\SSS}$ = $\modPhases{Z(\GGG)} = \langle \ol{X_1 X_2 X_3 X_4}, \ol{Z_1 Z_3 Z_2 Z_4} \rangle$,
this ISG code doesn't satisfy $\modPhases{\SSS} \leq \modPhases{\SSS_t}$ for any $t$,
so isn't an associated ISG code for $\GGG$.

We are very interested in precisely pinning down the relationship between a subsystem code $\GGG$
and any associated ISG code $\MMM_\GGG$ for it.
To this end, we state below one result and one open question.
For a subsystem code $\GGG$, the logical (almost!) Pauli group is defined as $\sublpg{\GGG}$,
and is isomorphic to $\PPP_k^\circ$, for some $k \in \ZZ_{\geq 0}$.
This $k$ is defined to be the number of logical qubits of the subsystem code.
The following theorem effectively says that the logical qubits of a subsystem code are
a subset of the logical qubits of any associated ISG code.
More precisely:

\begin{theorem}\label{thm:subsystem_logical_qubits_are_ISG_qubits}
    For any subsystem code $\GGG$ and any associated ISG code $\MMM_\GGG$,
    there exist fixed Paulis $x_1, z_1, \ldots, x_k, z_k$ such that
    $\langle x_1 \GGG, z_1 \GGG, \ldots, x_k \GGG, z_k \GGG \rangle$
    is a presentation for $\sublpg{\GGG} \cong \PPP_k^\circ$,
    and, for all $t$,
    $\langle i \isg{t}, x_1 \isg{t}, z_1 \isg{t}, \ldots, x_k \isg{t}, z_k \isg{t} \rangle$
    is a subgroup of $\lpg{t}$ isomorphic to $\PPP_k$.
\end{theorem}

We prove the statement in two steps.
First we show that every non-trivial coset of the subsystem code's logical Pauli group $\sublpg{\GGG}$
has a representative - a \textit{bare logical} -
which is also a non-trivial logical operator for any associated ISG code at all times.
Secondly, we show that inequivalent non-trivial logical operators of the subsystem code are also
inequivalent non-trivial logical operators of any associated ISG code.
Both statements rely on the fact that, for any subsystem code $\GGG$ and associated ISG code
$\MMM_\GGG = [\MMM_0, \MMM_1, \ldots]$,
the ISG is always a subgroup of the gauge group: $\forall t \in \ZZ, \isg{t} \leq \GGG$.
This can be seen directly from the fact that $\SSS_t$ is formed by
measuring generating sets for groups $\MMM_t, \MMM_{t-1}, \ldots, \MMM_0 \leq \GGG$.
Thus every element of $\SSS_t$ is a product of elements of $\GGG$,
multiplied by some product $\pm 1$ of measurement outcomes - both of which are again in $\GGG$.
We start in earnest by importing the following lemma separately,
so that we can refer to it again later:

\begin{lemma}[\cite{PoulinStabilizerFormalismSubsystemCodes}]\label{lem:gauge_group_structure}
    For any subsystem code $\GGG \leq \PPP_n$, the gauge group can be written as
    $\GGG = \SSS \langle i \rangle \langle x_1, z_1, \ldots, x_K, z_K \rangle$,
    for some Paulis $x_1, z_1, \ldots, x_K, z_K \in \PPP_n$,
    with $\langle x_1, z_1, \ldots, x_K, z_K \rangle$ isomorphic to $\PPP_K^\circ$.
\end{lemma}

The group $\langle x_1, z_1, \ldots, x_K, z_K \rangle \cong \PPP_K^\circ$ above corresponds exactly to
the $K$ gauge qubits of the subsystem code.
Next we define bare and logical operators.

\begin{definition}\label{defn:bare_dressed_logicals}
    Given a subsystem code $\GGG$,
    a logical operator (a member of a coset of $\sublpg{\GGG}$)
    is \textbf{bare} if it commutes with all gauge qubit logical operators $x_1, z_1, \ldots, x_K, z_K$,
    and \textbf{dressed} if not.
\end{definition}

Every logical operator coset in $\sublpg{\GGG}$ contains at least one bare logical operator.
We now relate these bare logicals to the \textit{centralizer}
$C(\GGG) = \{p \in \PPP_n \mid \forall g \in \GGG, pg = gp \}$ of $\GGG$ in $\PPP_n$ -
i.e.\ all the Pauli operators that commute with all elements of $\GGG$.

\begin{lemma}\label{lem:bare_logicals}
    Given a subsystem code $\GGG$,
    the set of all non-trivial bare logical operators is exactly $C(\GGG) - \GGG$.
    \begin{proof}
        A logical operator is a member of a coset of $\sublpg{\GGG}$,
        so - ignoring the group structure and just thinking in terms of sets - the non-trivial logical operators are exactly the set $N(\SSS) - \GGG$.
        For stabilizer groups like $\SSS$, the normalizer and centralizer coincide: $N(\SSS) = C(\SSS)$~\cite[Section 3.2]{GottesmanQecLectureNotes}.
        So a logical operator commutes with all elements of $\SSS$.
        A \textit{bare} logical operator by definition also commutes with all gauge qubit logical operators
        $x_1, z_1, \ldots, x_K, z_K$.
        From \Cref{lem:gauge_group_structure} we know $\GGG = \SSS \langle i \rangle \langle x_1, z_1, \ldots, x_K, z_K \rangle$,
        and since every Pauli operator commutes with $i$,
        we deduce that a bare logical operator commutes with all of $\GGG$.
        So all non-trivial bare logical operators are in $C(\GGG) - \GGG$.
        The converse can also be seen to hold;
        every element of $C(\GGG) - \GGG$ is a non-trivial bare logical operator.
    \end{proof}
\end{lemma}

We now tie this all together.

\begin{proof}[Proof of \Cref{thm:subsystem_logical_qubits_are_ISG_qubits}]
    A non-trivial subsystem code logical operator is an element of $C(\GGG) - \GGG$, by \Cref{lem:bare_logicals}.
    Since $\SSS_t \leq \GGG$ for all $t$, we have $C(\GGG) \leq C(\SSS_t)$,
    and hence $C(\GGG) - \GGG \leq C(\SSS_t) - \SSS_t$.
    If we again use the fact that centralizers and normalizers coincide for stabilizer groups like $\SSS_t$,
    we see that any non-trivial bare logical operator is
    a member of $N(\SSS_t) - \SSS_t$, and is hence a non-trivial logical operator of the ISG code at all timesteps.

    It remains to show that inequivalent logical operators of the subsystem code correspond to
    inequivalent logical operators of any associated ISG code.
    Let $p,q \in N(\SSS) - \GGG$ be inequivalent logicals of the subsystem code -
    i.e.\ $p\GGG \neq q\GGG$ in $\sublpg{\GGG}$.
    Since $\isg{t}$ is a subgroup of $\GGG$ for all $t \in \ZZ$,
    we can directly infer that $p\isg{t} \neq q\isg{t}$.
    So we're done; we've shown that choosing \textit{bare} logical representatives $x_1, z_1, \ldots, x_k, z_k$
    of $\langle x_1 \GGG, z_1 \GGG, \ldots, x_k \GGG, z_k \GGG \rangle = \sublpg{\GGG}$
    gives us a presentation
    $\langle i \isg{t} x_1 \isg{t}, z_1 \isg{t}, \ldots, x_k \isg{t}, z_k \isg{t} \rangle \cong \PPP_k$
    of a subgroup of $\lpg{t}$
    for all $t \in \ZZ$.
\end{proof}

As a corollary, we can now bound the distance of any associated ISG code.
For any subsystem code $\GGG$,
let $d_\GGG$ be its \textit{bare distance}; that is,
the minimum weight of any bare logical operator of the code.
Then the distance $d$ of any associated ISG code must be at most $d_\GGG$.
Similarly, letting $k_\GGG$ be the number of logical qubits in the subsystem code,
the number $k$ of logical qubits of any associated ISG code must be at least $k_\GGG$.
For this latter bound, it is natural to ask whether we can always find an associated ISG code which saturates it -
in other words,
can we always find an associated ISG code whose logical qubits are effectively exactly those of the subsystem code?

\begin{openquestion}\label{openq:subsystem_code_as_ISG_code_open_question}
    For every subsystem code $\GGG$,
    does there always exist an associated ISG code $\MMM_\GGG$ that establishes at time $T$,
    and fixed Paulis $x_1, z_1, \ldots, x_k, z_k$, such that
    $\langle x_1 \GGG, z_1 \GGG, \ldots, x_k \GGG, z_k \GGG \rangle$
    is a presentation for $\sublpg{\GGG} \cong \PPP_k^\circ$,
    and, for all $t$,
    $\langle i \isg{t}, x_1 \isg{t}, z_1 \isg{t}, \ldots, x_k \isg{t}, z_k \isg{t} \rangle$
    is a presentation for $\lpg{t} \cong \PPP_k$?
\end{openquestion}

For many subsystem codes, we can find such an associated ISG.
In the small Bacon-Shor code of \Cref{sec:isg_code_example}, for example,
we analysed the associated ISG code
$[\langle X_1 X_2, X_3 X_4 \rangle, \langle Z_1 Z_3, Z_2 Z_4 \rangle]$,
and found that after establishment it had logical Pauli group
$\lpg{t} = \langle i \isg{t}, X_1 X_3 \isg{t}, Z_1 Z_2 \isg{t} \rangle$ -
see \Cref{eq:eg:isg_code_ops_full}.
The same representatives generate
$\sublpg{\GGG} = \langle X_1 X_3 \GGG, Z_1 Z_2 \GGG \rangle$ too.
For some other subsystem codes, however,
we have not managed to find such an associated ISG code -
examples include Bombín's \textit{topological subsystem code} introduced in \Ccite[Section IV]{BombinRubyCodes}.
We plan to investigate this subsystem code through the lens of ISG codes in future work.

\subsection{Gauge fixing}\label{subsec:gauge_fixing}

We close this section with a word on what \textit{gauge fixing} means from an ISG code perspective.
For any subsystem code $\GGG \leq \PPP_n$ implemented by
measuring generators of some $\GGG'$ satisfying $\addPhases{\GGG'} = \GGG$,
gauge fixing is the name given to tracking what happens to the group stabilizing the system of $n$ qubits
as we measure these generators.
The name comes from the following:
as shown in \Cref{lem:gauge_group_structure},
$\GGG$ can be written as $\SSS \langle i \rangle \langle x_1, z_1, \ldots, x_K, z_K \rangle$,
where the subgroup $\langle x_1, z_1, \ldots, x_K, z_K \rangle \leq \GGG$
isomorphic to $\PPP_k^\circ$ is the group of logical operators for the gauge qubits.
When we measure one of these logical operators,
we are then said to be fixing the overall system into an eigenstate of one of these gauge qubits.
Hence: \textit{gauge fixing}.
If we instead view these measurements of generators of $\GGG'$ as
implementing an associated ISG code for $\GGG$,
gauge fixing is just another way of saying `tracking the ISG $\isg{t}$ over time'.
Note that any associated ISG code will potentially only gauge fix a subset of the gauge qubits.
This would then lead to a higher number of logical qubits in the ISG code compared to the subsystem code,
a distinctive feature of Floquet codes~\cite{DynamicallyGeneratedLogicalQubits,FloquetCodesWithoutParents,TylerFloquetifiedColourCode}.

    \section{Double hexagon code}\label{sec:double_hexagon_code}

We give here a little more detail on the double hexagon code of \Cref{sec:floquetifying_422},
showing that it still has distance two, like the $\qcode{4, 2, 2}$ code that it Floquetifies.
In \Cref{fig:double_hexagon_x_1,fig:double_hexagon_z_1},
we show operating regions corresponding to logical operators $\ol{x_1}$ and $\ol{z_1}$ respectively.
Recall (\Cref{subsubsec:stabilizers_and_logicals}) that given an operating region,
the corresponding logical operator is found by looking at which output legs are highlighted with which colours.
For example, in the rightmost diagram of \Cref{fig:double_hexagon_x_1},
we show seven timesteps $t \in \{0, 1, \ldots, 6\}$ of the period-six double hexagon code.
So looking at the output wires of this diagram tells us a representative of the logical $\ol{x_1}$ at timestep $t=6$.
Namely, re-using the notation that assigns label $(i, j)$ to the qubit with the $i$-th colour and $j$-th style,
where $i$ and $j$ are taken modulo 6 and 2 respectively, and the ordered lists of colours and styles are
$[
\mathbf{purple}, ~
\mathbf{pink}, ~
\mathbf{orange}, ~
\mathbf{yellow}, ~
\mathbf{brown}, ~
\mathbf{blue}]$
and
$[
\mathbf{solid}, ~
\mathbf{dashed}]$,
we see the corresponding logical operator at this timestep is
$X_{(0, 0)} X_{(0, 1)} X_{(5, 0)} X_{(5, 1)}$.

To see the corresponding logical operator at a different timestep $t$,
we can just truncate the diagram after timestep $t$ and again look at the highlighted output edges.
This is shown for timesteps $t \in \{0, 1, 2\}$ in \Cref{fig:double_hexagon_x_1_truncations} below;
the corresponding operators are
$X_{(0, 0)} X_{(0, 1)} X_{(5, 0)} X_{(5, 1)}$ after timesteps $t=0$ and $t=1$, and
$X_{(2, 0)} X_{(2, 1)} X_{(1, 0)} X_{(1, 1)}$ after timestep $t=2$.
In fact, truncating the rightmost diagram of \Cref{fig:double_hexagon_x_1} after \textit{any} timestep $t$
gives an operating region whose corresponding $\ol{x_1}$ representative has weight four.
The same is true for all possible truncations of the $\ol{z_1}$ operating region in \Cref{fig:double_hexagon_z_1},
and would also be true if we repeated this process for the other logical operator cosets $\ol{x_2}$ and $\ol{z_2}$.
That is, truncating after timestep $t$ would always give us an operating region whose corresponding operator had weight four.
It would be tempting to conclude, therefore, that the double hexagon code has distance four,
and that Floquetifying the $\qcode{4, 2, 2}$ code has increased its distance.
But this would be wrong.

The reason is that while this process gives us logical operators in the Floquetification,
it doesn't necessarily give us minimum-weight ones,
even though we started with minimum-weight operating regions of the $\qcode{4, 2, 2}$ code in
the leftmost diagrams of \Cref{fig:double_hexagon_x_1,fig:double_hexagon_z_1}.
Indeed, one can find stabilizers at a given timestep $t$ with which to multiply these logical operators,
in order to reduce their weight to two.
One can do this algebraically,
making use of the fact that we already worked out the ISG $\SSS_t$ of the double hexagon code for all $t$
(see \Cref{eq:floquetified_4_2_2_s_t} and the preceding sentence).
Alternatively, one can do this graphically,
using the fact that Pauli webs on the same diagram form a group
and can thus be multiplied together~\cite[Appendix A.3]{MattTimeDynamics}.
Graphically, multiplying together two (unsigned CSS) Pauli webs just means
laying one diagram on top of the other,
up to the rule that if a wire is highlighted twice by the same colour,
these cancel out and the wire becomes unhighlighted\footnote{%
    If laying one diagram on top of another results in a wire being highlighted both red and green,
    then the resulting Pauli web is no longer CSS, in the language of this paper.
    It is, however, a valid Pauli web according to the more general definition used in
    \Ccite{ZxFaultTolerancePsiQuantum,MattTimeDynamics}.}.
In \Cref{fig:double_hexagon_x_1_truncations_multiplied} below,
we thus multiply the weight-four operating regions by appropriate stabilizing regions
to get weight-two operating regions.

Since this code is fairly small, one can use these ideas to check by hand that for all $t$,
no logical operator at time $t$ has weight one.
Hence, since representatives of weight two exist for the logical operator coset $\ol{x_1}$,
the code has distance two, as promised.


\begin{figure}[p]
    \centering
    \includegraphics[width=450pt]{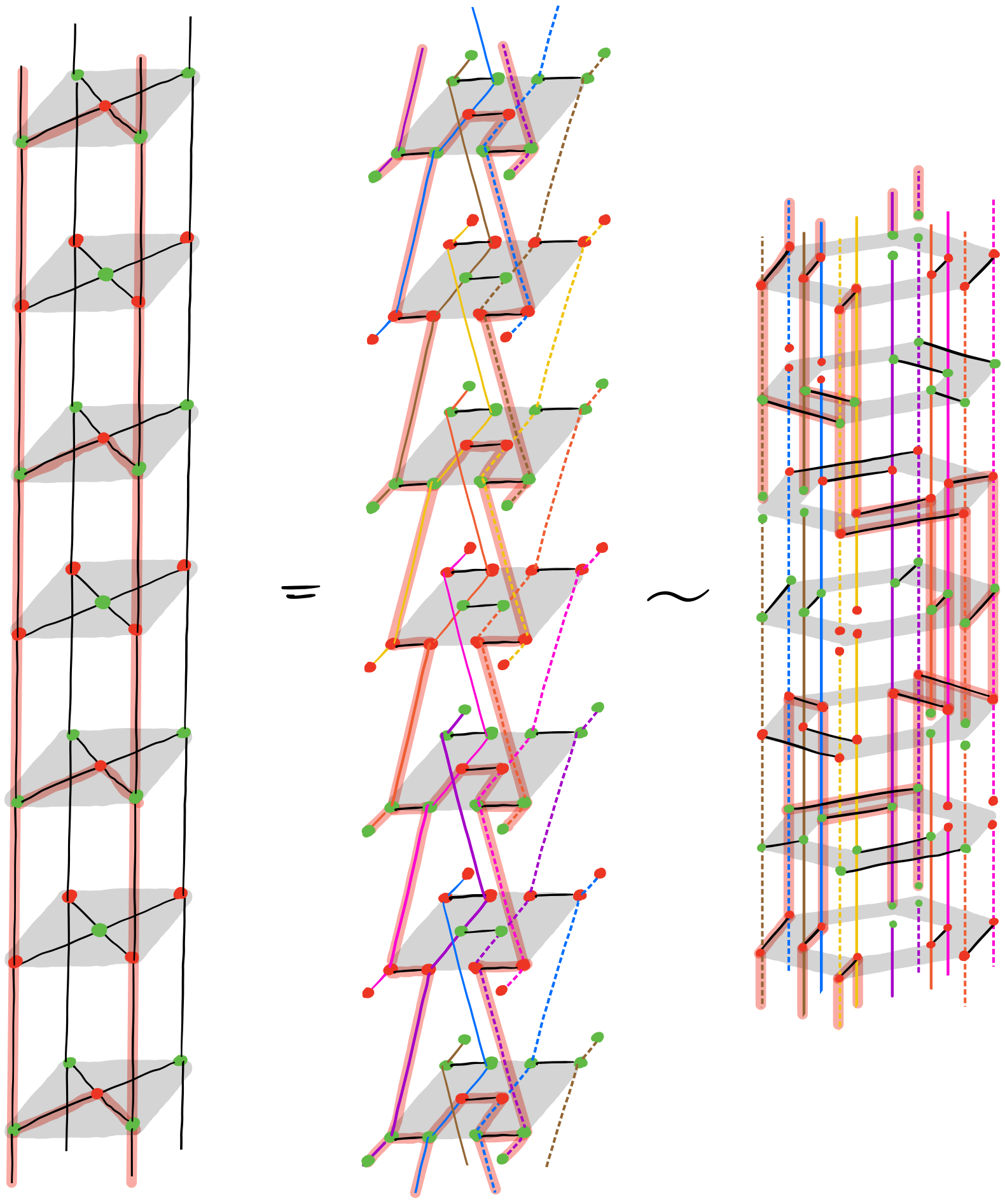}
    \caption{%
        An operating region corresponding to a representative of the logical $\ol{x_1}$ in
        the $\qcode{4, 2, 2}$ code (left and middle)
        and in its Floquetification, the double hexagon code (right).}
    \label{fig:double_hexagon_x_1}
\end{figure}

\begin{figure}[p]
    \centering
    \includegraphics[width=450pt]{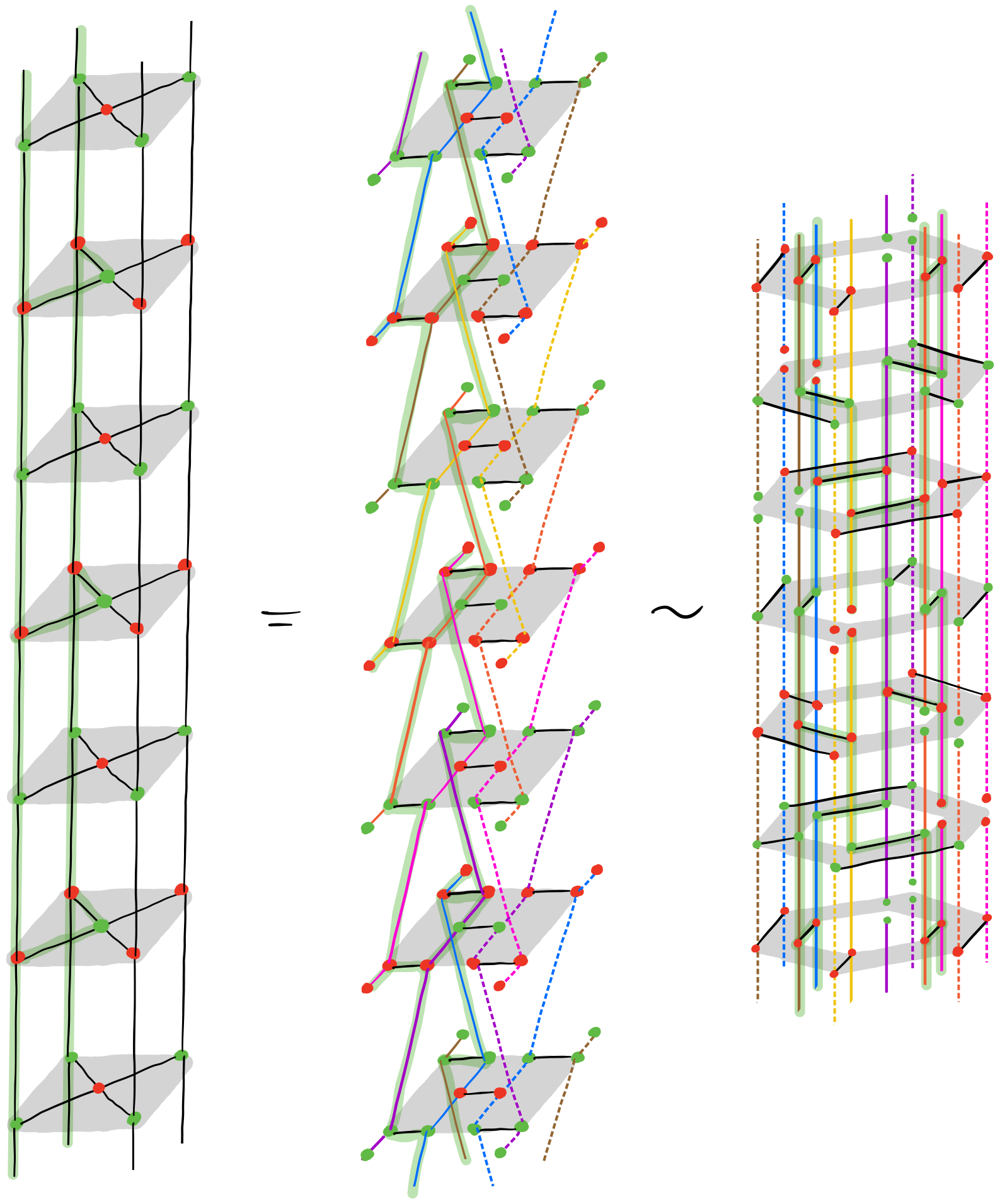}
    \caption{%
        An operating region corresponding to a representative of the logical $\ol{z_1}$ in
        the $\qcode{4, 2, 2}$ code (left and middle)
        and in its Floquetification, the double hexagon code (right).}
    \label{fig:double_hexagon_z_1}
\end{figure}

\begin{figure}[t]
    \centering
    \includegraphics[width=450pt, valign=c]{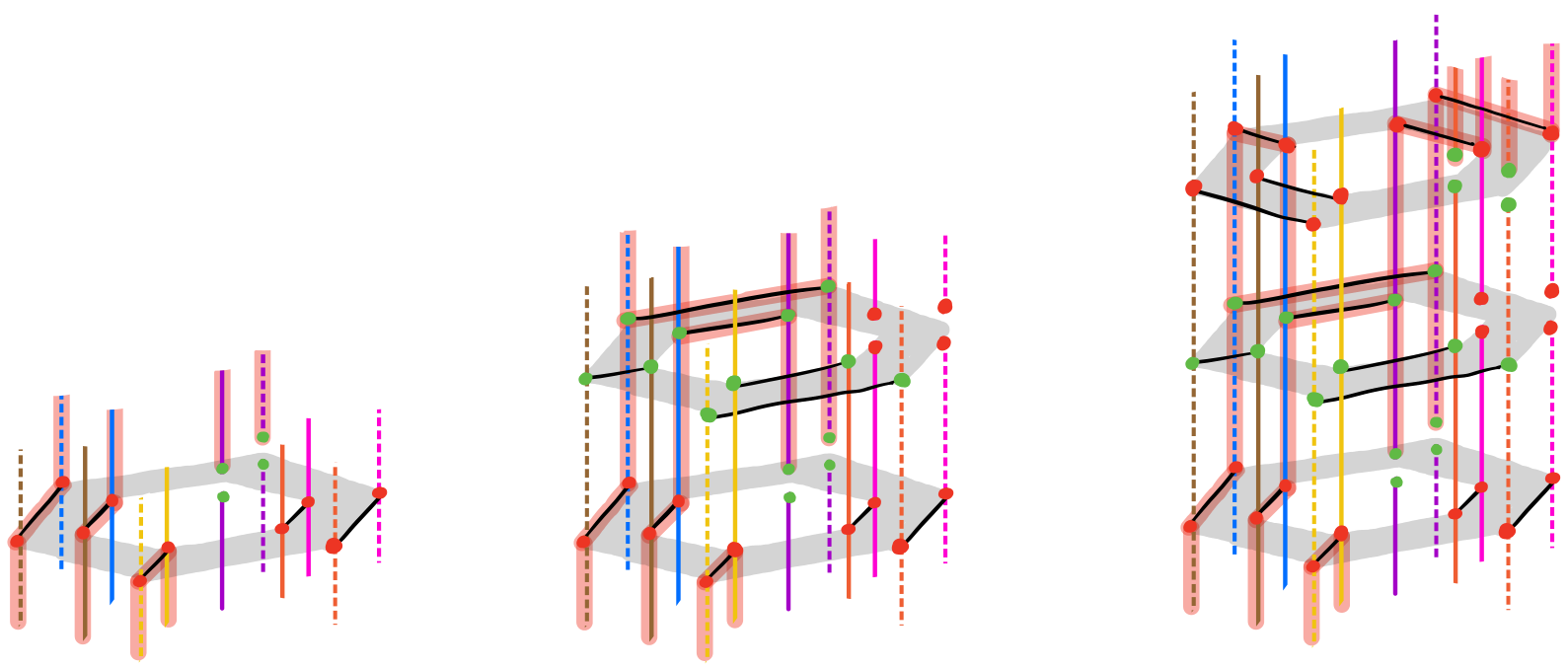}
    \caption{%
        The double hexagon code after timesteps $t=0$, $t=1$ and $t=2$ respectively,
        annotated with an operating region corresponding to a representative of the logical $\ol{x_1}$.}
    \label{fig:double_hexagon_x_1_truncations}
\end{figure}

\begin{figure}[b]
    \centering
    \includegraphics[width=430pt]{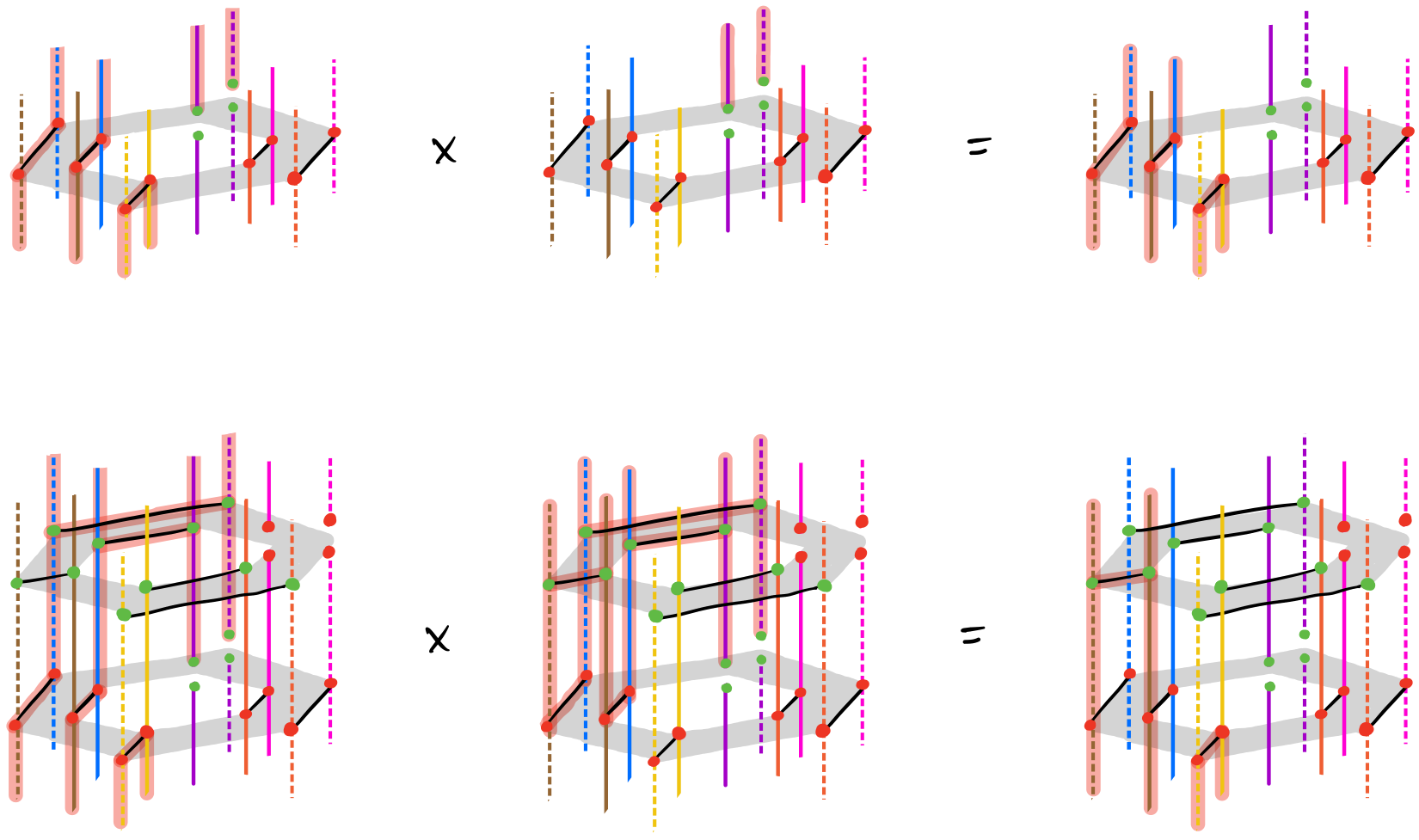}
    \caption{%
        For timesteps $t=0$ and $t=1$ respectively,
        we multiply the weight-four operating regions of \Cref{fig:double_hexagon_x_1_truncations}
        by appropriate stabilizing regions, which gives us weight-two operating regions.
        This corresponds exactly to composition of the corresponding Pauli operators -
        for example, the topmost graphical equation corresponds to the algebraic equation
        $X_{(0, 0)} X_{(0, 1)} X_{(5, 0)} X_{(5, 1)} \after X_{(0, 0)} X_{(0, 1)} = X_{(5, 0)} X_{(5, 1)}$.}
    \label{fig:double_hexagon_x_1_truncations_multiplied}
\end{figure}

    \clearpage
    \section{Floquetifying the colour code: beyond the bulk}\label{sec:floquetifying_colour_code_beyond_bulk}

In this section, we expand on what we mean by a code exhibiting \textit{drift},
and sketch some ideas as to how it might be avoided.
Then we describe what goes wrong when attempting to Floquetify a non-planar code,
before turning our attention back to our specific example of the colour code.

\subsection{Drift}\label{subsec:drift}

In a drifting Floquet code, at any timestep a new qubit might need to be initialised and entangled into the code,
or an old qubit might be measured out and not used again.
Were such a code to be implemented on hardware that required qubits to be in fixed positions
(e.g.\ on a superconducting chip),
the code might appear to slowly drift across the hardware in some direction.
For example, if the qubits of the code roughly form a square,
and qubits on the top boundary of this square are being regularly measured out,
while on the bottom boundary new qubits are regularly being entangled in,
then the code as a whole will seem to drift downwards.
Viewing our Floquetification procedure as tilting the time direction in a ZX-diagram
shows that drift is a natural consequence - see the left hand side of \Cref{fig:drifting_spacetime_volume}.

\begin{figure}[b]
    \centering
    \includegraphics[width=430pt]{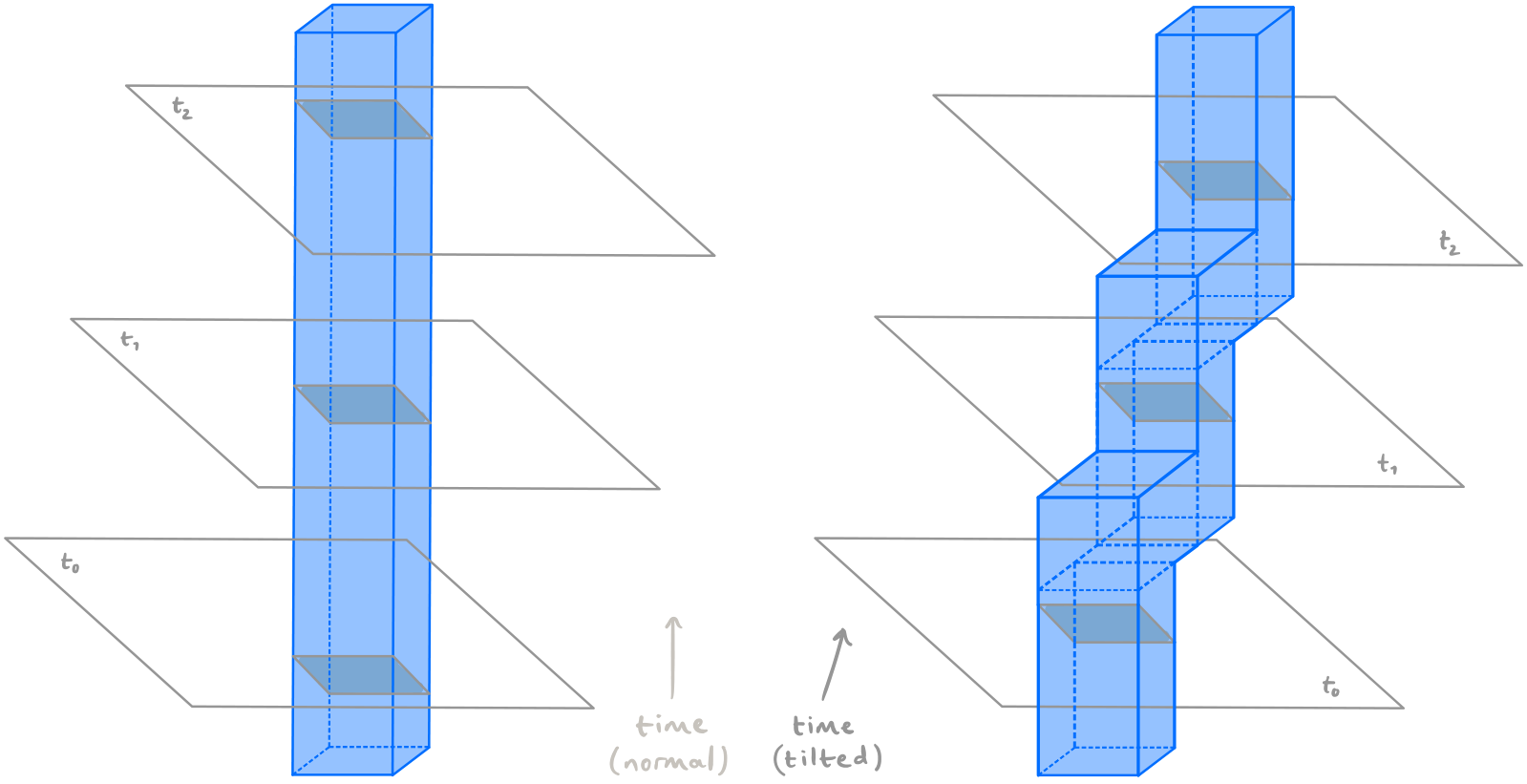}
    \caption{%
        Here blue prisms represent ZX-diagrams of codes,
        while grey planes are timeslices $t_0 < t_1 < t_2$ after the time direction has been tilted.
        The shaded grey squares show the intersection of the code's ZX-diagram with a particular timestep.
        On the left, we see that taking a stabilizer code
        and tilting the time direction to get an equivalent Floquet code naturally causes drift -
        notice how the shaded grey intersections change position within the timeslices $t_0$, $t_1$ and $t_2$.
        On the right, we show that one way to avoid drifting boundaries is to Floquetify a code that is itself moving;
        here the grey intersections are always in the same place within the timeslices $t_0$, $t_1$ and $t_2$.}
    \label{fig:drifting_spacetime_volume}
\end{figure}

There are many ways we can try to avoid drifting boundaries.
In fact, we saw one of them already,
when we Floquetified the $\qcode{4, 2, 2}$ code in \Cref{sec:floquetifying_422}.
Ordinarily, the resulting code would exhibit drift,
in exactly the fashion visualised on the left hand side of \Cref{fig:drifting_spacetime_volume}.
However, our Floquetification re-used qubits cyclically;
once the pair of qubits on the top boundary of the code were measured out,
they were immediately re-entangled into the bottom boundary of the code at the next timestep -
see the rightmost diagram of \Cref{fig:4_2_2_rewritten}.
In this sense, we made our code drift around in a circle, which gives it its double hexagon shape.
While this works nicely for the very small $\qcode{4, 2, 2}$ code,
it's not a good general solution from a practical perspective,
since for hardware with qubits in fixed positions,
an unusual connectivity would be required in order to implement a large code that drifts circularly.
(For other types of hardware, where qubits can be moved around and re-used,
this might be less of a problem).

Another potential solution is to Floquetify a code that is already moving,
so as to cancel out the drift that results from Floquetifying it.
For example, the rotated surface code can be grown outwards from one boundary (the top one, say),
then contracted from the opposite one (the bottom),
such that it appears to move in a particular direction (upwards) -
see, for example, \Ccite[Figure 40(c)]{GameOfSurfaceCodes}.
The same idea can be used to move other stabilizer and subsystem codes.
If such a code can be moved in an appropriate direction and at an appropriate speed,
the resulting Floquetification will remain put;
this is sketched abstractly on the right hand side of \Cref{fig:drifting_spacetime_volume}.

Other systematic ways one might hope to obtain non-drifting boundaries involve changing the measurement schedule.
Suppose a drifting Floquetification has schedule
$\MMM = [\MMM_0, \MMM_1, \ldots, \MMM_{\ell-1}]$.
Consider a new Floquet code with schedule
$\MMM = [\MMM_0, \MMM_1, \ldots, \MMM_{\ell-1}, \MMM_{\ell-1}, \MMM_{\ell-2}, \ldots, \MMM_0]$.
One could hope that this new code drifts in one direction in timesteps $0 \leq t < \ell$,
but then drifts in exactly the opposite direction in timesteps $\ell \leq t < 2\ell$,
and so overall stays within a bounded region,
while still having the required error correcting properties.

\subsection{Periodic boundaries}\label{subsec:periodic_boundaries}

Recall once again that we can interpret our Floquetification procedure as tilting the time direction in a ZX-diagram.
If we attempt to do this for a code with periodic boundary conditions,
this tilting can have the effect of making time run circularly.
The diagram below illustrates the problem:
a `thickened' torus $T^2 \times I$ is shown in blue
(faces of the cube marked with the same letters $A$ and $B$ are glued together),
and three planes through it are drawn in different shades of grey.
The thickened torus is an abstraction for a ZX-diagram showing the time evolution of a code that lives on a torus,
while the three grey planes are potential new timeslices after Floquetifying this code.
However, it's impossible to label these three planes with unique integers
such that a line perpendicular to them passes through them an increasing sequence:
\begin{equation}\label{eq:toric_spacetime_volume_timeslice}
    \begin{aligned}
        \includegraphics[width=100pt]{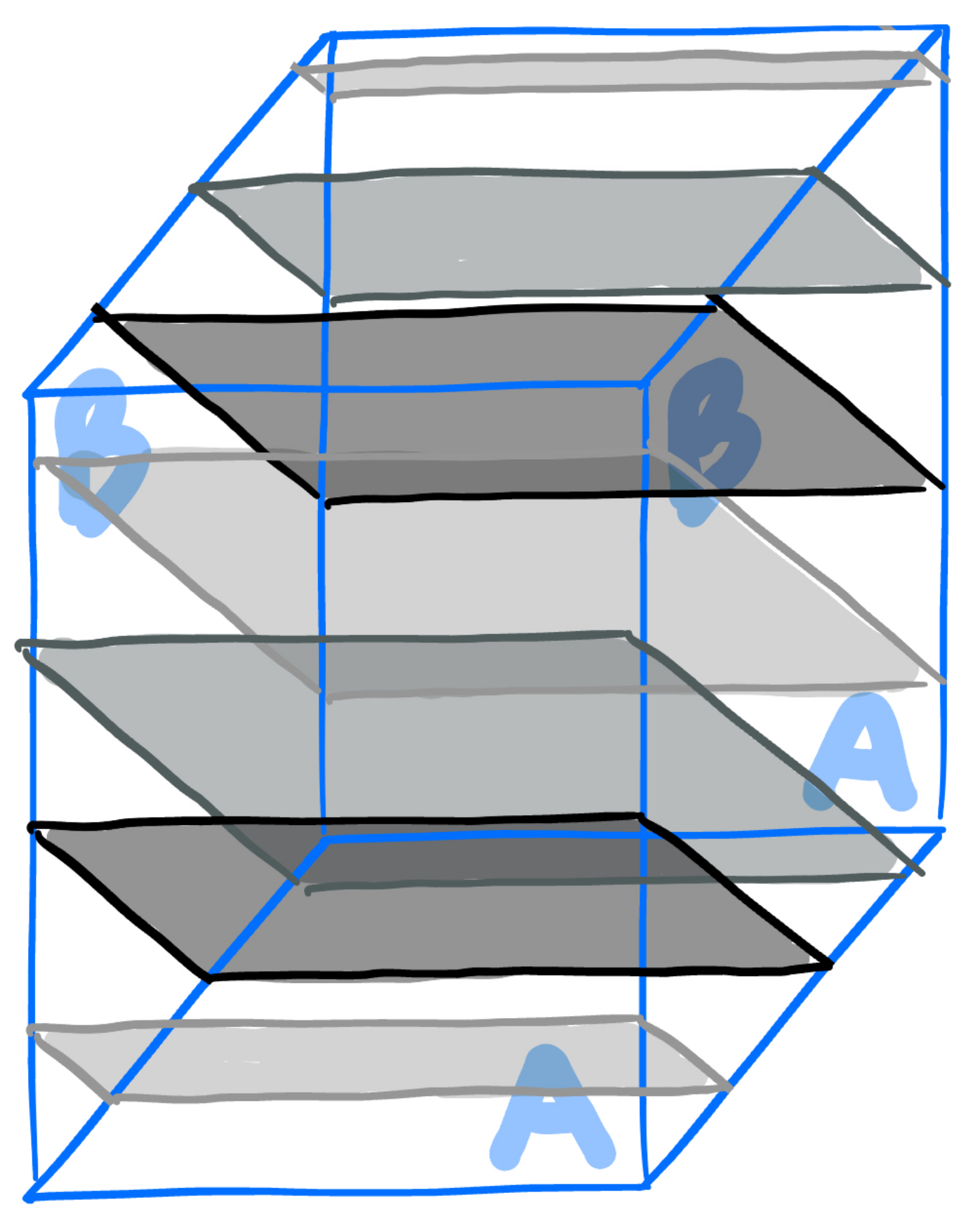}
    \end{aligned}
\end{equation}

A possible alternative strategy in the case that we're given a code that lives on a torus is the following:
one can first try to Floquetify a patch of the bulk,
by reinterpreting which wires in the code's ZX-diagram represent qubit world-lines.
After doing this for as large a portion of the bulk as needed,
one can then try to impose periodic boundary conditions on this Floquetified bulk patch,
and see whether the resulting ZX-diagram encodes the expected logical qubits.

\subsection{The colour code}\label{subsec:colour_code}

Having discussed drift and periodic boundaries in general,
we now discuss strategies for completing our Floquetification of the colour code bulk from the main text.
Option one is to Floquetify a planar colour code, which will give us a Floquet code with drifting boundaries.
We show this explicitly in
\Cref{fig:floquetified_planar_colour_code_zx_rewritten,fig:floquetified_planar_colour_code_drift,fig:floquetified_planar_colour_code_zx}
below, for a small colour code that lives on a parallelogram,
but the idea extends to a planar colour code of any size.

\begin{figure}[h!]
    \centering
    \includegraphics[width=450pt]{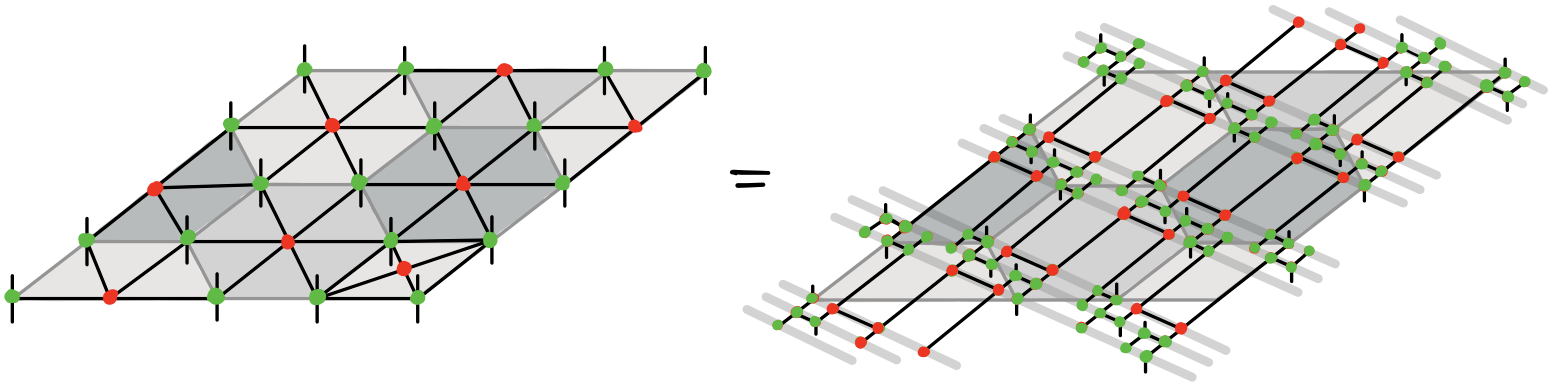}
    \caption{%
        Two equivalent ZX-diagrams for one round of a small planar colour code on a parallelogram,
        in which all $Z \otimes Z \otimes \ldots \otimes Z$ measurements are performed.
        As in \Cref{fig:colour_code_rewritten}, the grey highlighted bars have no meaning in the ZX-calculus,
        and instead denote timesteps of the corresponding Floquetified code.}
    \label{fig:floquetified_planar_colour_code_zx_rewritten}
\end{figure}

\begin{figure}[h!]
    \centering
    \hspace{50pt}
    \includegraphics[width=400pt]{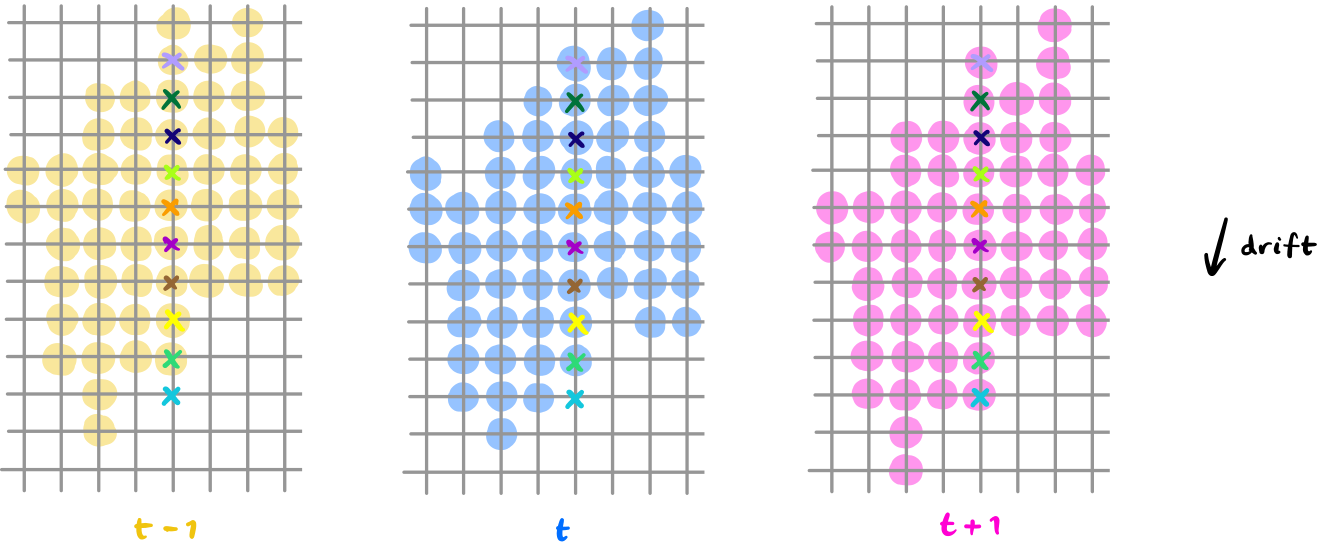}
    \caption{%
        The square lattice layout of the Floquetified planar colour code, with qubits on vertices,
        at consecutive timesteps $t-1$, $t$ and $t+1$.
        Vertices marked with a coloured cross correspond to
        qubit world-lines of the same colour in \Cref{fig:floquetified_planar_colour_code_zx}.
        Vertices highlighted yellow, blue or pink are \textit{active} at the respective timestep $t-1$, $t$ or $t+1$.
        Here \textit{active} means the qubit is initialised for the first time on or before this timestep,
        and is measured out for the last time on or after this timestep.
        If we were to plot the active qubits across all timesteps,
        we would see that the code drifts downwards and slightly to the left.}
    \label{fig:floquetified_planar_colour_code_drift}
\end{figure}

\begin{figure}
    \centering
    \includegraphics[width=387pt]{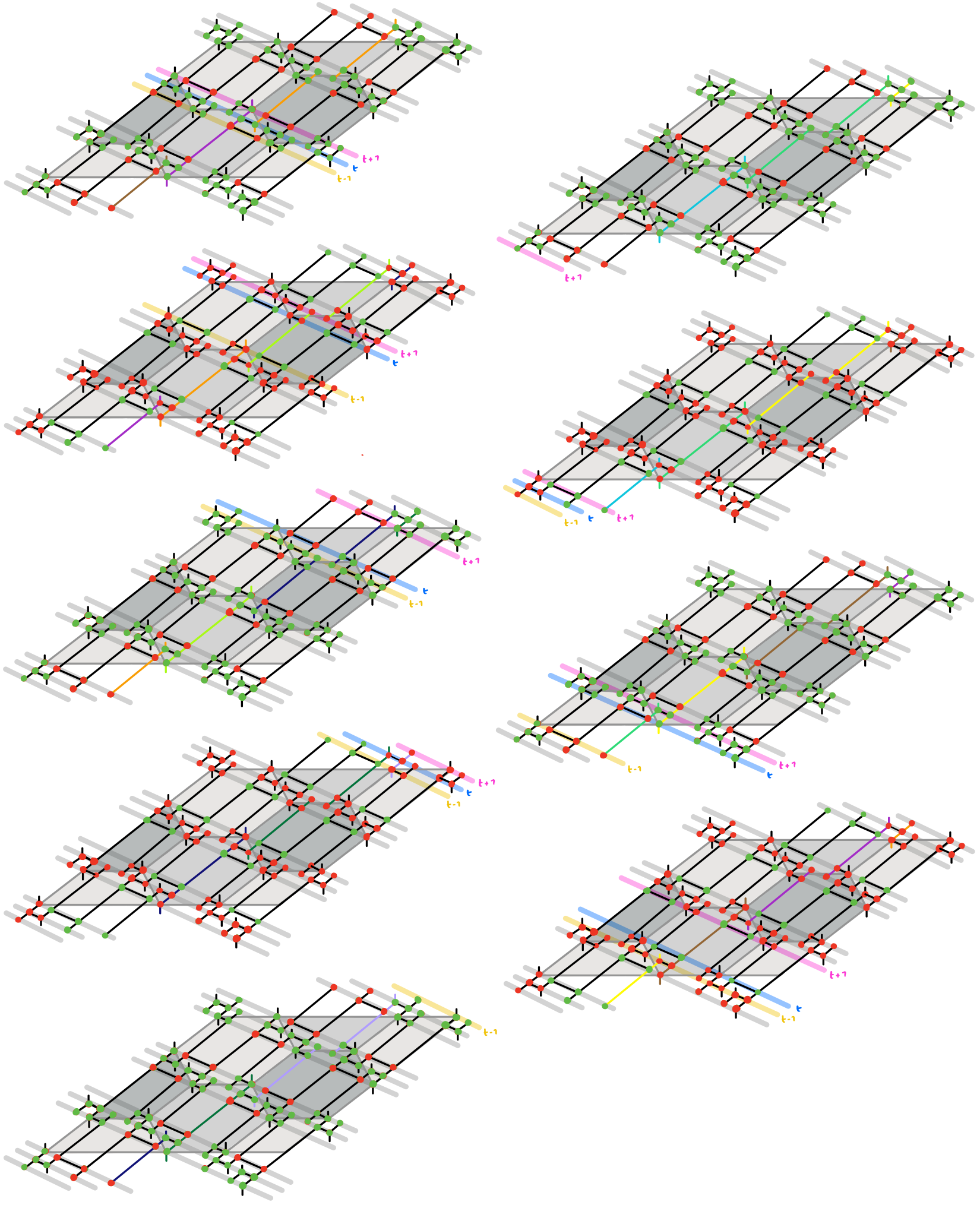}
    \caption{%
        ZX-diagram of nine rounds of the planar colour code.
        This figure should be read columnwise up the page; first the left column, then the right.
        Much like in \Cref{fig:colour_code_rewritten},
        vertical ZX-wires between consecutive layers have been omitted for clarity.
        Instead we state in words that, for a vertex $v$ of the underlying honeycomb lattice,
        the small vertical upwards wire closest to $v$ going upwards from one layer
        is connected to the small vertical wire closest to $v$ protruding downwards from the next layer.
        Certain wires have been coloured - these indicate the world-lines of different qubits.
        Specifically, these are the world-lines of qubits that in a single column
        on the square lattice of the Floquetified code,
        as denoted by coloured crosses in \Cref{fig:floquetified_planar_colour_code_drift}.
        Grey highlighted bars again denote timesteps of the Floquetified code;
        three of them have been coloured yellow, blue and pink respectively,
        denoting three consecutive timesteps $t-1$, $t$ and $t+1$.}
    \label{fig:floquetified_planar_colour_code_zx}
\end{figure}

The sequence of measurements that a particular qubit undergoes
will be almost the same as that described in \Cref{sec:floquetifying_colour_code},
with the exception that qubits near a boundary can be measured out earlier or initialised later than usual.
Additionally, as a consequence, a pairwise measurement that would ordinarily have been performed is omitted
if one of the two qubits it would have applied to has been measured out early or initialised late.
Detectors and logicals can again be found via Pauli webs,
by drawing them on the ZX-diagram for the planar colour code and seeing how they map to the Floquetification.
In the bulk, they will look as described in \Cref{sec:floquetifying_colour_code}.

If we wish to prevent the Floquetified code from drifting,
we could instead force our parallelogrammatic colour code to regularly
expand outwards from its top and right boundaries,
and contract inwards from its bottom and left boundaries,
such that this cancels out the drift that results from Floquetifying it.
The expansion is performed by preparing pairs of qubits in Bell states,
then entangling them into the code by measuring the stabilizers of a bigger colour code.
Similarly, contraction is done by measuring out pairs of qubits in the Bell basis.
We must be a little careful here to maintain the code's timelike distance;
when measuring the stabilizers of the new, bigger code during an expansion phase,
we must repeatedly measure these stabilizers for $d$ rounds, where $d$ is the distance of the original code.
Contraction, however, can be done instantaneously.
The distance by which we need to expand is determined by how fast the world-lines of qubits in the Floquetified code
move forwards and rightwards through the ZX-diagram for the underlying colour code,
as in \Cref{fig:floquetified_planar_colour_code_zx} or the right hand side of \Cref{fig:colour_code_rewritten}.
Let \textit{rows} and \textit{columns} in the honeycomb lattice be defined as
the pink and orange lines respectively in the following diagram:

\begin{equation}\label{eq:honeycomb_rows_columns}
    \includegraphics[width=200pt, valign=c]{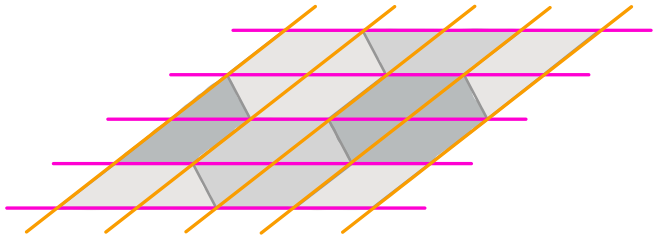}
\end{equation}

Suppose we now have a parallelogrammatic colour code with distance $d$.
The geometry of the parallelogram forces $d$ to be even.
Such a code occupies $\frac{3d}{2} - 1$ rows and $\frac{3d}{2} - 1$ columns in the honeycomb lattice.
After every $d$ rounds of stabilizer measurements,
the code needs to have moved $4d$ rows forwards and $d$ columns rightwards in the honeycomb lattice
in order to cancel out the drift the Floquetification process would bring.
It's most convenient to expand and contract the colour code by multiples of $3$ rows or columns,
so let's also assume that $d$ is an even multiple of $3$.
Thus, during an expansion phase, the code occupies $\frac{3d}{2} - 1 + 4d$ rows and $\frac{3d}{2} - 1 + d$ columns.
The distance of the code during such a phase is therefore
$\frac{11d}{3}$ in the forwards direction and $\frac{5d}{3}$ in the rightwards direction.
We sketch this in \Cref{fig:floquetified_planar_colour_code_cancelling_drift}.

The extra complication that this idea entails is that we now have also timelike boundaries to Floquetify.
These are points at which pairs of qubits either become Bell states and get entangled into the code,
or get measured out in the Bell basis.
We could just do the same thing in our Floquetified code, but while the pairs of qubits to be entangled/measured
are all neighbours in the honeycomb lattice, they are not all neighbours in the square lattice of the Floquetified code.
This is shown in \Cref{fig:floquetified_planar_colour_code_measuring_out}.
Fortunately, as shown in that \Cref{fig:floquetified_planar_colour_code_local_readout_circuit},
we can construct a circuit on groups of four qubits which achieves the same thing
but respects the nearest neighbour connectivity of the square lattice.

\begin{figure}
    \centering
    \includegraphics[width=250pt]{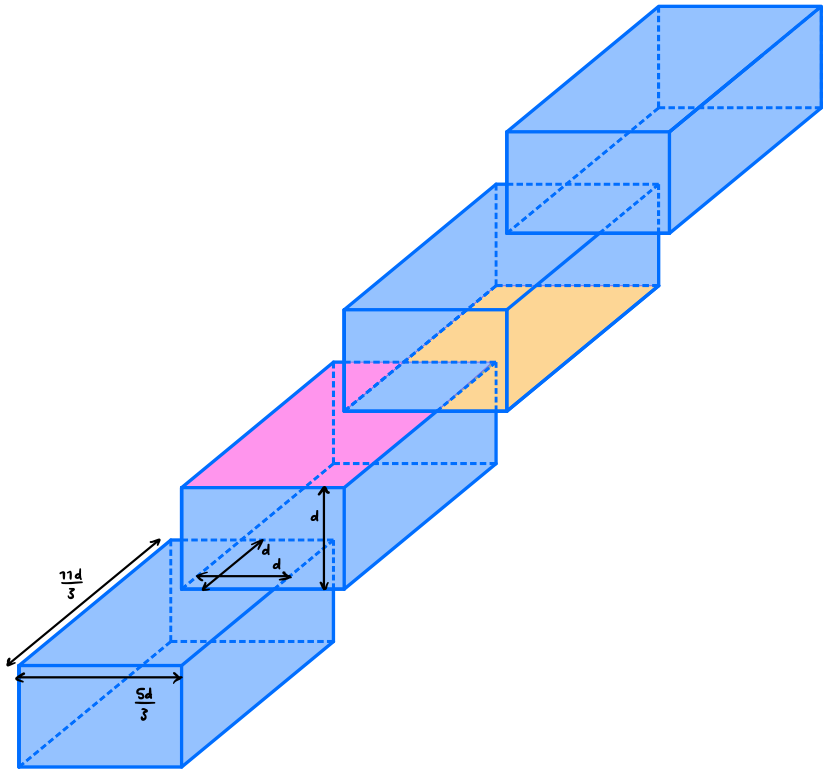}
    \caption{%
        An abstraction of a ZX-diagram of a parallelogrammatic colour code
        which expands and contracts in the fashion specified in the text.
        In this figure, time goes directly upwards.
        The distance of the original colour code is denoted $d$.
        Two timelike boundaries are shaded;
        at the pink one, pairs of qubits are measured in the Bell basis,
        while at the orange one, Bell pairs are entangled into the code.}
    \label{fig:floquetified_planar_colour_code_cancelling_drift}
\end{figure}

\begin{figure}
    \centering
    \begin{subfigure}{\textwidth}
        \centering
        \includegraphics[width=0.9\textwidth, valign=c]{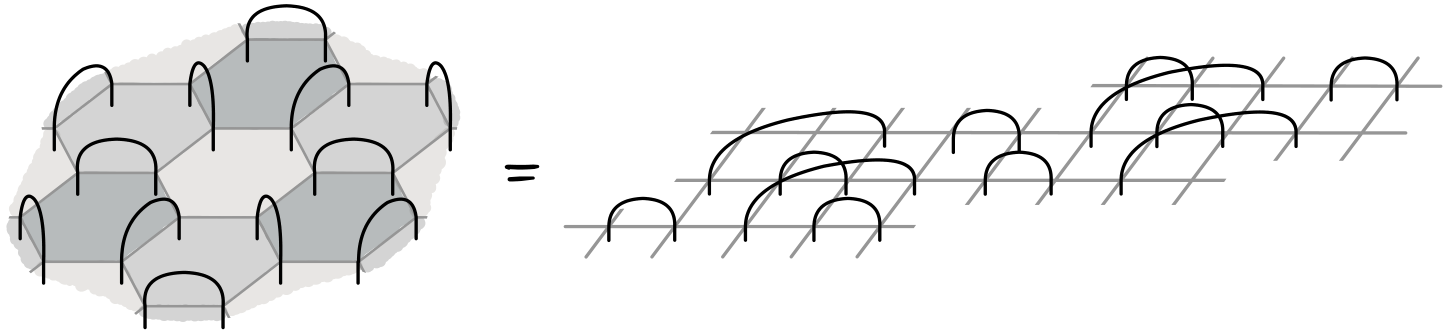}
        \caption{\ }
        \label{fig:floquetified_planar_colour_code_measuring_out}
    \end{subfigure}
    \begin{subfigure}{0.15\textwidth}
        \centering
        \includegraphics[width=30pt, valign=c]{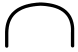}
        \caption{\ }
        \label{fig:floquetified_planar_colour_code_destructive_bell_measurement}
    \end{subfigure}
    \begin{subfigure}{0.15\textwidth}
        \centering
        \includegraphics[width=30pt, valign=c]{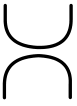}
        \caption{\ }
        \label{fig:floquetified_planar_colour_code_non_destructive_bell_measurement}
    \end{subfigure}
    \begin{subfigure}{0.6\textwidth}
        \centering
        \includegraphics[width=200pt, valign=c]{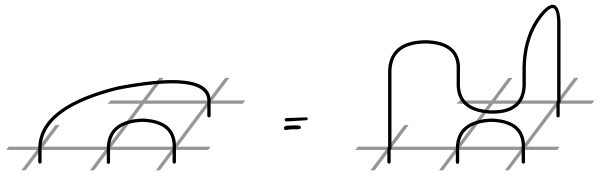}
        \caption{\ }
        \label{fig:floquetified_planar_colour_code_local_readout_circuit}
    \end{subfigure}
    \caption{
        \textbf{(a)}
        The Bell basis measurements on a timelike boundary of the colour code,
        such as the pink one in \Cref{fig:floquetified_planar_colour_code_cancelling_drift},
        and their image in the square lattice Floquetification.
        \textbf{(b)}
        A destructive Bell basis measurement (or rather, post-selection) in the ZX-calculus.
        \textbf{(c)}
        A non-destructive Bell basis measurement (or rather, post-selection) in the ZX-calculus.
        \textbf{(d)}
        Two equivalent ZX-diagrams for performing a pair of Bell basis measurements on four qubits.
        Whereas the diagram on the left involves a measurement between two non-neighbouring qubits on the square lattice,
        the diagram on the right respects this connectivity constraint,
        but costs one extra Bell basis measurement.
    }
    \label{fig:floquetified_planar_colour_code_timelike_boundaries}
\end{figure}

\begin{figure}[p]
    \centering
    \includegraphics[width=350pt]{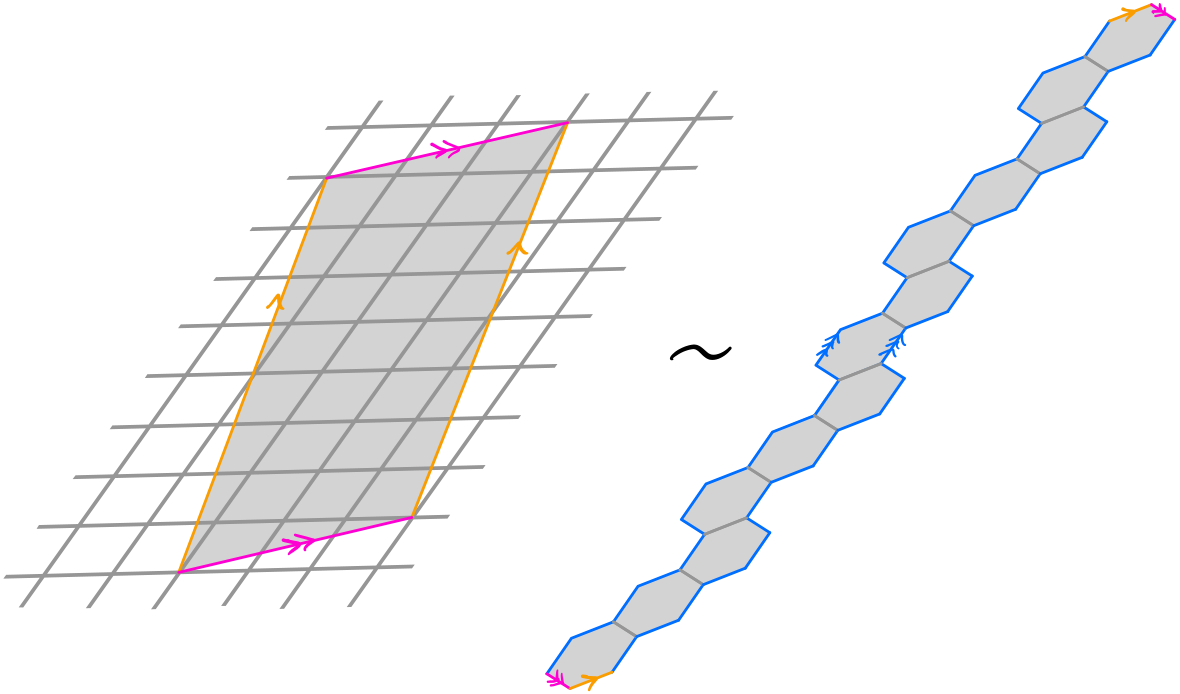}
    \caption{%
        Imposing periodic boundary conditions on a single rectangle of the Floquetified colour code bulk,
        as in \Cref{fig:floquetified_colour_code_measurement_schedule},
        defines a Floquet code with no logical qubits.
        As a patch of the Floquetified colour code bulk,
        this rectangle corresponds to the strip of colour code bulk on the right.
        Since this strip is not three-colourable,
        imposing toric boundary conditions on it prevents it from supporting a colour code.}
    \label{fig:floquetified_toric_colour_code_tile}
\end{figure}

\begin{figure}[p]
    \centering
    \includegraphics[width=\linewidth]{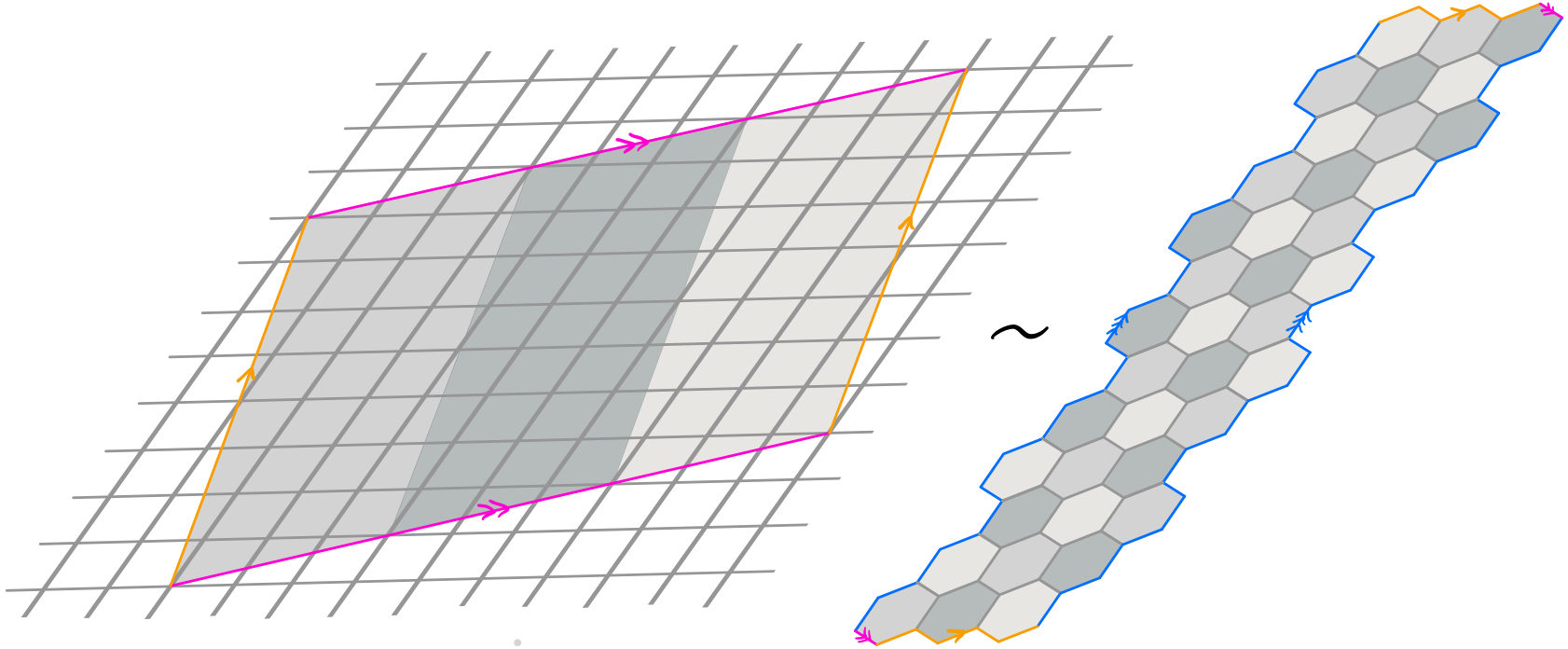}
    \caption{%
        If we instead tile together a multiple of three columns of these rectangles,
        we define a Floquet code with four logical qubits -
        the same number as for the toric colour code.
        The left diagram of this figure is thus the support of
        the smallest member of our family of toric Floquetified colour codes.
        As a patch of the Floquetified colour code bulk,
        this corresponds to the wider strip of colour code bulk on the right.
        Unlike in \Cref{fig:floquetified_toric_colour_code_tile}, this strip is three-colourable,
        and thus supports a colour code when toric boundary conditions are imposed on it.}
    \label{fig:floquetified_toric_colour_three_tiles}
\end{figure}

An alternative option is
to impose periodic boundary conditions on a patch of Floquetified colour code bulk.
On the right hand side of \Cref{fig:floquetified_colour_code_measurement_schedule}
we drew a
grey rectangle,
and said that by tiling copies of these rectangles across the square lattice,
we can define the whole bulk of the Floquetified colour code.
An obvious candidate, then, for the patch of bulk on which to impose a toric geometry, is this rectangle.
But it turns out that we can't just tile the torus any way we like with these rectangles;
in order to actually encode any logical information,
we must choose a tiling in which the number of columns of such rectangles is a multiple of three.
That is, if we choose a tiling in which this is not the case,
then we find that the logical Pauli group $\lpg{t}$ becomes trivial after a full period of measurements.
When we choose a tiling in which it is the case, on the other hand,
we find that $\lpg{t} \cong \PPP_4$ after establishment, just as for the toric colour code.
One explanation for this restriction is the fact that the underlying colour code is three-coloured;
indeed, analogously, one can't define a colour code on a toric honeycomb lattice that isn't three-colourable.
We sketch this correspondence in
\Cref{fig:floquetified_toric_colour_code_tile,fig:floquetified_toric_colour_three_tiles}.

    \section{Measurement and post-selection}\label{sec:measurement_vs_post_selection}

Throughout this paper, we've used post-selection instead of measurement in all of our ZX-diagrams.
Here we justify this.
ISG codes detect and correct errors using detectors, as defined in \Cref{subsubsec:detectors}.
The logical information they encode is characterised by their logical Pauli group $\lpg{t}$.
Suppose we draw the ZX-diagram for a given ISG code, and actually use measurement rather than post-selection.
That is, we parametrise spiders representing measurements by $\frac{1-m}{2}\pi$,
where $m$ is the measurement outcome
(we showed an example of this in \Cref{eq:zx_CSS_measurements} for CSS measurements).
The set of all detecting regions in this ZX-diagram
corresponds exactly to the set of all detectors for the code.
Likewise the set of all operating regions corresponds exactly to the logical Pauli group $\lpg{t}$.
Thus these two types of Pauli webs capture
all of the error-correcting and information-encoding capabilities of an ISG code.

The admissible Pauli webs for a Pauli spider are independent of whether that spider's phase is $0$ or $\pi$.
Hence detecting regions and operating regions are independent of this too.
In particular, they are unaffected by whether a spider representing a measurement has phase
$\frac{1-1}{2}\pi = 0$ or $\frac{1-(-1)}{2}\pi = \pi$.
Thus from an error correction perspective there is no loss of generality in assuming all the
measurements have outcome $m=1$, say - i.e.\ in using post-selection, rather than measurement.
In other contexts, of course, this is not justifiable.

This is hinted at in \Ccite[Appendix A]{MattTimeDynamics}
and stated even more clearly in \Ccite[III. Checks]{ZxFaultTolerancePsiQuantum} in their own language.
Craig Gidney has also tweeted about the idea of a `Pauli-free' ZX-calculus~\cite{PauliFreeZxCalculusTweet}
suitable for use in error correction,
in which a spider with phase $\alpha$ is equivalent to one with phase $\alpha + k\pi$, for any $k \in \ZZ$.

    \section{`Colour code' vs `color code'}

For fear of having their British passport revoked,
one third of the authors of this paper insisted on using `colour' instead of `color' throughout.
We apologise for any distress caused.
    \section{`Floquet': how to say it}\label{sec:how_to_say_it}
It rhymes with `okay'!
\end{document}